\RequirePackage{anyfontsize}
\RequirePackage{amsthm}
\documentclass[pdflatex,sn-basic]{sn-jnl}

\usepackage{amsmath}
\usepackage{mathrsfs}
\usepackage{amsthm}
\usepackage{amssymb}
\usepackage{amsfonts}
\usepackage{bm}

\usepackage[title]{appendix}%
\usepackage{mathtools}
\usepackage{enumerate}
\usepackage[dvipsnames]{xcolor}
\usepackage{tikz}
\usepackage{graphicx}

\newcommand{\E}{\mathbb{E}}
\renewcommand{\Pr}{\mathbb{P}}

\newcommand{\bX}{\mathbf{X}}

\newcommand{\bY}{\mathbf{Y}}

\newcommand{\bU}{\mathbf{U}}

\newcommand{\es}{\enspace\enspace\:}

\newcommand{\Qfunc}{Q}
\newcommand{\Comp}{K}
\newcommand{\comp}{k}
\newcommand{\Compset}{\mathcal{K}}

\theoremstyle{thmstyleone}%
\newtheorem{thm}{Theorem}
\newtheorem{lemma}{Lemma}
\newtheorem{prop}{Proposition}

\theoremstyle{thmstyletwo}%
\newtheorem{assump}{Assumption}
\newtheorem{remark}{Remark}

\theoremstyle{thmstylethree}%
\newtheorem{example}{Example}
\newtheorem{defn}{Definition}

\allowdisplaybreaks  % allow align environment to be broken into two pages

\begin{document}
\title{\bf \!\!\!\!\!Information Compression in Dynamic Games\!\!\!\!\!}
    
\author*[1]{\fnm{Dengwang} \sur{Tang}}
\email{dwtang@umich.edu}
    
\author[2]{\fnm{Vijay} \sur{Subramanian}}
\email{vgsubram@umich.edu}

\author[2]{\fnm{Demosthenis} \sur{Teneketzis}}
\email{teneket@umich.edu}
		
\affil[1]{\orgdiv{Ming Hsieh Department of Electrical and Computer Engineering}, \orgname{University of Southern California}, \orgaddress{\city{Los Angeles}, \state{CA}, \postcode{90089-2560}, \country{USA}}}

\affil[2]{\orgdiv{Electrical and Computer Engineering Division, Electrical Engineering and Computer Science Department}, \orgname{University of Michigan}, \orgaddress{\city{Ann Arbor}, \state{MI}, \postcode{48109}, \country{USA}}}

\abstract{
    One of the reasons why stochastic dynamic games with an underlying dynamic system are challenging is because strategic players have access to enormous amount of information which leads to the use of extremely complex strategies at equilibrium. One approach to resolve this challenge is to simplify players’ strategies by identifying appropriate compression of information maps so that the players can make decisions solely based on the compressed version of information, called the information state. Such maps allow players to implement their strategies efficiently. For finite dynamic games with asymmetric information, inspired by the notion of information state for single-agent control problems, we propose two notions of information states, namely  mutually sufficient information (MSI) and unilaterally sufficient information (USI). Both these information states 
    are obtained by applying information compression maps that are independent of the strategy profile. We show that Bayes-Nash Equilibria (BNE) and Sequential Equilibria (SE) exist when all players use MSI-based strategies. We prove that when all players employ USI-based strategies the resulting sets of BNE and SE payoff profiles are the same as the sets of BNE and SE payoff profiles resulting when all players use full information-based strategies. We prove that when all players use USI-based strategies the resulting set of weak Perfect Bayesian Equilibrium (wPBE) payoff profiles can be a proper subset of all wPBE payoff profiles.
    We identify MSI and USI in specific models of dynamic games in the literature. 
    We end by presenting an open problem: Do there exist strategy-dependent information compression maps that guarantee the existence of at least one equilibrium or maintain all equilibria that exist under perfect recall? We show, by a counterexample, that a well-known strategy-dependent information compression map used in the literature does not possess any of the properties of the strategy-independent compression maps that result in MSI or USI.
}

% C72 Noncooperative Games
% C73 Stochastic and Dynamic Games; Evolutionary Games; Repeated Games
% D80: Information, Knowledge, and Uncertainty (general)

\keywords{Non-cooperative Games, Dynamic Games, Information State, Sequential Equilibrium, Markov Decision Process}

\pacs[JEL Classification]{C72, C73, D80}
\pacs[MSC Classification]{90C40, 91A10, 91A15, 91A25, 91A50}

\pacs[Acknowledgements]{The authors would like to thank Yi Ouyang, Hamidreza Tavafoghi, Ashutosh Nayyar, Tilman B\"{o}rgers, and David Miller for helpful discussions.}

\maketitle
\section{Introduction}\label{sec:suffinfo:intro}
The model of stochastic dynamic games has found application in many engineering and socioeconomic settings, such as transportation networks, power grid, spectrum markets, and online shopping platforms. In these {settings}, multiple agents/players make decisions over time on top of an ever-changing environment with players having different goals and asymmetric information. For example, in transportation networks, individual drivers make routing decisions based on information from online map services in order to reach their respective destinations as fast as possible. Their actions then collectively affect traffic conditions in the future. Another example involves online shopping platforms, where buyers leave reviews to inform potential future buyers, while sellers update prices and make listing decisions based on the feedback from buyers. In these systems, players' decisions are generally not only interdependent, but also affect the underlying environment as well as future decisions and payoffs of all players in complex ways.

Determining the set of equilibria, or even solving for one equilibrium, in a given stochastic dynamic game can be a challenging task. The main challenges include: (a) the presence of an underlying environment/system that can change over time based on the actions of all players; (b) incomplete and asymmetric information; (c) large number of players, states, and actions; and (d) growing amount of information over time which results in a massive strategy space. {As a result of the advances in} technology, stochastic dynamic games today are often played 
by players (e.g. big corporations) that have access to substantial computational resources along with a large amount of data for decision making. 
{Nevertheless, even these players are computationally constrained}, and they must make decisions in real-time, hence complicated strategies may not be feasible for them. Therefore, it is important to determine computationally efficient strategies for players to play at equilibria. Compression of players' information and then use of the strategies based on the compressed information is a well-heeled methodology that results in computationally efficient strategies. In this paper we address some of the above-mentioned challenges. We concentrate on the challenges associated with information compression, namely the existence of equilibria under information compression, and the preservation of all equilibrium payoff profiles under information compression. We leave as a topic of future investigation the discovery of efficient algorithms for the computation of equilibria based on strategies that use compressed information.

{Specifically, our} goal is to identify appropriate strategy-independent
\footnote{Strategy independent information compression maps are maps that are not parameterized by a strategy profile. Examples of strategy-independent information compression maps include those that use a fixed-subset of the game's history (e.g. the most recent observation) or some statistics based on the game's history (e.g. the number of times player $i$ takes a certain action). Strategy-dependent maps are parameterized by a strategy profile (see Section \ref{sec:openproblems}).} information compression maps in dynamic games so that the resulting compressed information has properties/features sufficient to satisfy the following requirements: (R1) existence of equilibria when all players use strategies based on the {compressed information}; (R2) equality of the set of all equilibrium payoff profiles that are achieved when all players use full information based-strategies with the set of all equilibrium payoff profiles that are achieved when all players use strategies based on the compressed information.

Inspired by the literature on single-agent decision/control problems, particularly the notion of information state, we develop notions of information state (compressed information) that satisfy requirements (R1) and (R2). Specifically, we introduce the notions of Mutually Sufficient Information (MSI) and Unilaterally Sufficient Information (USI). We show that MSI has properties/features sufficient to satisfy (R1), whereas USI has properties sufficient to satisfy (R2) under several different equilibrium concepts.

The remainder of the paper is organized as follows: In Section \ref{sec:litreview} we briefly review related literature in stochastic control and game theory. In Section \ref{sec:contributions} we list our contributions. In Section \ref{sec:notations} we introduce our notation. In Section \ref{sec:suffinfo:gamemodel} we formulate our game model. In Section \ref{sec:msi} and Section \ref{sec:usi} we introduce the notion of mutually sufficient information and unilaterally sufficient information respectively. We {present our main} results in Section \ref{sec:isbe}. We discuss {these} results in Section \ref{sec:siapplications}. 
%Strategy-dependent compression is discussed in Section \ref{sec:openproblems}. % Not any more
{We discuss an open problem, primarily associated with strategy-dependent information compression, in Section \ref{sec:openproblems}.}
{We provide} supporting results in Appendix \ref{app:aux:infostate}. {We present} alternative characterizations of sequential equilibria in Appendix \ref{app:SE}. {We provide proofs of the results of Sections \ref{sec:twoinfostates} and \ref{sec:isbe}}
%strategy-independent compression maps and their associated discussions are provided 
in Appendix \ref{app:proofsmain}. {We present the details of the discussions in Section \ref{sec:siapplications} and Section \ref{sec:openproblems} in Appendix \ref{app:siapplications}.} 

\subsection{Related Literature}\label{sec:litreview}
{We first present a brief literature survey on information compression in single-agent decision problems because it has inspired several of the key ideas presented in this paper.}

Single-agent {decision/control} problems are problems where one agent chooses actions over time on top of an ever-changing system to maximize {their} total reward. These problems have been extensively studied in the control theory \citep{kumar1986stochastic}, operations research \citep{powell2007approximate}, computer science \citep{russell2002artificial}, and mathematics \citep{bellman1966dynamic} literature. Models like Markov Decision Process (MDP) and Partially Observable Markov Decision Process (POMDP) have been analyzed and applied widely in real-world systems. It is well known that in an MDP, the agent can use a Markov strategy---making decisions based on the current {state}---without loss of optimality. 
A Markov strategy can be seen as a strategy based on compressed information: the full {information---state and action history}---is compressed into only the current state. Furthermore, in finite horizon problems, such optimal Markov strategies can be found through a sequential decomposition procedure. It is also well known that any POMDP can be transformed into an MDP with an appropriate belief acting as the underlying state~\citep[Chapter 6]{kumar1986stochastic}. As a result, the agent can use a belief-based strategy without loss of optimality. A belief-based strategy compresses the full information into the conditional belief of the current state. Critically, this information compression is strategy-independent \citep{aastrom1965optimal,smallwood1973optimal,sondik1978optimal,kumar1986stochastic}.
%\vjmargincomment{Cite the original POMDP papers too.}
For general single-agent control problems, sufficient conditions that guarantee optimality of compression-based strategies have been proposed under {the names} of \emph{sufficient statistic} \citep{shiryaev1964markov,striebel1965sufficient,whittle1969sequential,hinderer1970sufficient,striebel1975lecture} and \emph{information state} \citep{kumar1986stochastic,mahajan2016decentralized,subramanian2020approximate}. In these works, the authors transform single-agent control problems with partial observations into equivalent problems with complete observations with the sufficient statistic/information state acting as the underlying state.

{Multi-agent dynamic decision problems are either teams where all agents have the same objective, or games where agents have different objectives and are strategic. Information compression in dynamic teams has been investigated in  \cite{varaiya1978delayed,nayyar2010optimal,nayyar2013decentralized,mahajan2016decentralized,tavafoghi2021unified,subramanian2020approximate,kao2022common}, and many other works (see \cite{tavafoghi2021unified} and \cite{subramanian2020approximate} for a list of references).} 
Dynamic games can be divided into two categories: those with a static underlying environment (e.g. repeated games), and those with an underlying dynamic system. Over the years, economics researchers have studied repeated games extensively (e.g. see \cite[Chapter 7]{myerson2013game}). As our focus is on dynamic games with an underlying dynamic system, we will not discuss {the literature on repeated games}. Among models for dynamic games with an underlying dynamic system, the model of zero-sum games, as a particular class which possesses special properties, has been analyzed in  \cite{shapley1953stochastic,mertens1981stochastic,rosenberg1998duality} and many others (see \cite{ouyang2024approach} for a list of references). 
{Non-zero-sum} games with an underlying dynamic system and symmetric information have also been studied extensively \citep{bacsar1998dynamic,filar2012competitive}. For such dynamic games with perfect information, the authors of \cite{maskin2001markov} introduce the concept of Markov Perfect Equilibrium (MPE), where each player compresses their information into a Markov state. Dynamic games with asymmetric information {have been} analyzed in \cite{mertens2003equilibria,maskin2013youtube,nayyar2012dynamic,nayyar2013common,gupta2014common,gupta2016dynamic,ouyang2015oligopoly,ouyang2016dynamic,tavafoghi2016stochastic,hamidthesis,vasal2019spbe,tang2022dynamic,ouyang2024approach}.
In \cite{nayyar2013common}, the authors introduce the concept of Common Information Based Markov Perfect Equilibrium (CIB-MPE), which is an extension of MPE in partially observable systems. In a CIB-MPE, all players choose their actions at each time based on the {Common-Information-Based (CIB)} belief {(a compression of the common information)} and private information instead of full information. The authors establish the existence of CIB-MPE under the assumption that the CIB belief is strategy-independent. Furthermore, the authors develop a sequential decomposition procedure to solve for such equilibria. In \cite{ouyang2016dynamic}, the authors extend the result of \cite{nayyar2013common} to a particular model where the CIB beliefs are strategy-dependent. 
{They introduce the concept of Common Information Based Perfect Bayesian Equilibrium (CIB-PBE). In a CIB-PBE all players choose their actions based on the CIB belief and their private information.} They show that such equilibria can be found through a sequential decomposition whenever the decomposition has a solution. The authors conjecture the existence of such equilibria. 
{The authors of \cite{tang2022dynamic} extend the model of \cite{ouyang2016dynamic} to games among teams. They} consider two compression maps and their associated equilibrium concepts. 
For the first compression map, which is strategy-independent, they establish {preservation of equilibrium payoffs}. For the second information compression map, which is strategy-dependent, 
%\demosreplace{the authors propose a sequential decomposition procedure where such equilibria can be found if the procedure admits a solution. However, using an example, the authors show that equilibria under the second concept do not always exist. The same example also proves the existence conjecture of \cite{ouyang2016dynamic} to be false.}
{they propose a sequential decomposition of the game. If the decomposition admits a solution, then there exists a CIB-BNE based on the compressed information. Furthermore, they provide an example where CIB-BNE based on this specific compressed information do not exist. The example also proves that the conjecture about the existence of CIB-PBEs, made in \cite{ouyang2016dynamic}, is false.}

%\demosreplace{Furthermore, there have been two lines of work in the economics theory literature that consider games between players with imperfect recall.}
{In addition to the methods on information compression that appear in \cite{maskin2001markov,nayyar2013common,ouyang2016dynamic,tang2022dynamic}, there are two lines of work on games where the players' decisions are based on limited information.} In the first line of work, players face exogenous hard constraints on the information that can be used to choose actions \citep{piccione1997interpretation,battigalli1997dynamic,GROVE199751,HALPERN199766,AUMANN1997102}.
%\vjedit{(and other articles in in the special issue of Games and Economic Behavior, Volume 20, Issue 1, July 1997)}
In the second line of work, players can utilize any finite automaton with any number of states to choose actions, however more complex automata are assumed to be more costly \citep{abreu1988structure,banks1990repeated}. In our work, we also deal with finite automaton based strategies. However, there is a critical difference between our work and {both lines of literature}{both of the above-mentioned lines of work}: Our primary interest is to study conditions under which a compression based strategy profile can form an equilibrium under standard equilibrium concepts {when \emph{unbounded} rationality and perfect recall are allowed}. Under these equilibrium concepts, we do not restrict the strategy of any player, nor do we impose any penalty on complicated strategies. In other words, a compression based strategy needs to be a best response compared to all possible strategies with full recall in terms of the payoff alone.
The methodology for information compression presented in this paper is similar in spirit to that of \cite{maskin2001markov,nayyar2013common,ouyang2016dynamic,vasal2019spbe,ouyang2024approach}. {However}, this paper is significantly different from those works as it deals with the discovery of information compression maps that lead not only to the existence {(in general)} of various types of compressed information based equilibria but also to the preservation of all equilibrium payoff profiles (a topic not investigated in \cite{maskin2001markov,nayyar2013common,ouyang2016dynamic,vasal2019spbe,ouyang2024approach}).
This paper builds on \cite{tang2022dynamic}; it identifies embodiments of the two information compression maps studied in \cite{tang2022dynamic} for a much more general class of games than that of \cite{tang2022dynamic}, and a broader set of equilibrium concepts.

\subsection{Contributions}\label{sec:contributions}
Our main contributions are the following:
\begin{enumerate}
    \item We propose two notions of information states/compressed information for dynamic games with asymmetric information that result in from strategy-independent compression maps: Mutually Sufficient Information (MSI) and Universally Sufficient Information (USI) --- Definitions \ref{def:msi} and \ref{def:usi}, respectively. We present an example that highlights the differences between MSI and USI.
    \item We show that in finite dynamic games with asymmetric information, Bayes--Nash Equilibria (BNE) and Sequential Equilibria (SE) exist when all players use MSI-based strategies --- Theorems \ref{thm:msiexist} and \ref{thm:msiseexist}, respectively.
    \item We prove that when all players employ USI-based strategies the resulting sets of BNE and SE payoff profiles are same as the sets of BNE and SE payoff profiles resulting when all players use full information based strategies --- Theorems \ref{thm:usiequiv} and \ref{thm:usiseequiv}, respectively.
    \item We prove that when all players use USI-based strategies the resulting set of weak Perfect Bayesian Equilibrium (wPBE) payoff profiles can be a proper subset of the set of all wPBE payoff profiles --- Proposition \ref{prop:exwpbe}. A result similar to that of Proposition \ref{prop:exwpbe} is also true under Watson's PBE \citep{watson2017general}. 

    Figure \ref{fig:vennpayoff} depicts the results stated in Contributions 3 and 4 above.

    \begin{figure}[!ht]
    	\centering
    	\begin{tikzpicture}
    		\draw[thick] (0, 5) rectangle (8.2, 0);
    		\draw[thick] (1, 4) rectangle (8.0, 0.2);
    		\draw[thick] (2, 3) rectangle (7.8, 0.4);
    		\draw[thick] (3, 2) rectangle (7.6, 0.6);
    		
    		\draw[thick] (0, 4.5) node[anchor=west] {USI-based BNE = All BNE};
    		\draw[thick] (1, 3.5) node[anchor=west] {All wPBE};
    		\draw[thick] (2, 2.5) node[anchor=west] {USI-based wPBE};
    		\draw[thick] (3, 1.5) node[anchor=west] {USI-based SE = All SE};
    	\end{tikzpicture}
    	\caption{A Venn diagram showing the relationship of the sets of payoff profiles for different equilibrium concepts using either unilateral sufficient information (USI) based strategy profiles or general strategies.} \label{fig:vennpayoff}
    \end{figure}
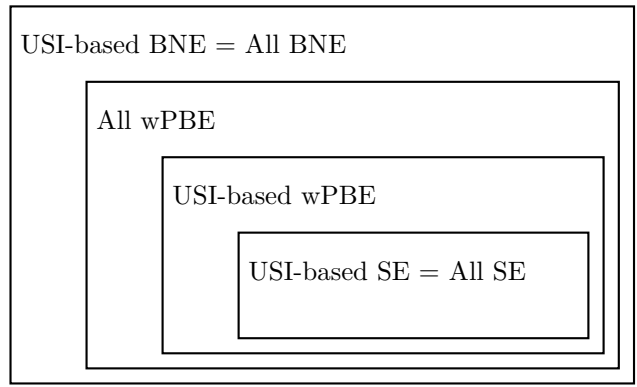
    
    \item We present several examples --- Examples \ref{ex:repeatedgames} through \ref{ex:ouyang} --- of finite dynamic games with asymmetric information where we identify MSI and USI.\\
\end{enumerate}

    Additional contributions of this work are:\\~

\begin{enumerate}
    \item A set of alternative definitions of SE --- Appendix \ref{app:SE}. These definitions are equivalent to the original definition of SE given in \cite{kreps1982sequential} and help simplify some of the proofs of the main results in this paper.

    \item A new methodology for establishing existence of equilibria. The methodology is based on a best response function defined through a dynamic program for a single-agent control problem.

    \item A counterexample showing that a well-known strategy-dependent compression map, resulting in sufficient private information along with common information-based beliefs, does not guarantee existence of equilibria based on the above-stated compressed information.

\end{enumerate}

\subsection{Notation}\label{sec:notations}
We follow the notational convention of stochastic control literature (i.e. using random variables to define the system, representing information as random variables, etc.) instead of the convention of game theory literature (i.e. game trees, nodes, information sets, etc.) unless otherwise specified. This allows us to apply techniques from stochastic control, which we rely heavily upon, in a more natural way.
We use capital letters to represent random variables, bold capital letters to denote random vectors, and lower case letters to represent realizations. We use superscripts to indicate players, and subscripts to indicate time. We use $i$ to represent a typical player and $-i$ to represent all players other than $i$. We use $t_1:t_2$ to indicate the collection of timestamps $(t_1, t_1+1, \cdots, t_2)$. For example, $X_{1:4}^i$ stands for the random vector $(X_1^1, X_2^i, X_3^i, X_4^i)$. For random variables or random vectors represented by Latin letters, we use the corresponding script capital letters to denote the space of values these random vectors can take. For example, $\mathcal{H}_t^i$ denotes the space of values the random vector $H_t^i$ can take. The products of sets refers to Cartesian products. We use $\Pr(\cdot)$ and $\E[\cdot]$ to denote probabilities and expectations, respectively. We use $\Delta(\varOmega)$ to denote the set of probability distributions on a finite set $\varOmega$. For a distribution $\nu\in \Delta(\varOmega)$, we use $\mathrm{supp}(\nu)$ to denote the support of $\nu$. When writing probabilities, we will omit the random variables when the lower case letters that represent the realizations clearly indicate the random variable it represents. For example, we will use $\Pr(y_t^i|x_t, u_t)$ as a shorthand for $\Pr(Y_t^i = y_t^i|X_t = x_t, U_t=u_t)$. When $\lambda$ is a function from $\varOmega_1$ to $\Delta(\varOmega_2)$, with some abuse of notation we write $\lambda(\omega_2|\omega_1):=(\lambda(\omega_1))(\omega_2)$ as if $\lambda$ is a conditional distribution. We use $\bm{1}_A$ to denote the indicator random variable of an event $A$.

In general, probability distributions of random variables in a dynamic system are only well defined after a complete strategy profile is specified. We specify the strategy profile that defines the distribution in superscripts, e.g. $\Pr^g(x_t^i|h_t^0)$. When the conditional probability is independent of a certain part of the strategy $(g_t^i)_{(i, t)\in \varOmega}$, we may omit this part of the strategy in the notation, e.g. $\Pr^{g_{1:t-1}}(x_{t}|y_{1:t-1}, u_{1:t-1})$, $\Pr^{g^{i}}(u_{t}^i|h_t^i)$ or $\Pr(x_{t+1}|x_{t}, u_{t})$. We say that a realization of some random vector (for example $h_t^i$) is \emph{admissible} under a partially specified strategy profile (for example $g^{-i}$) if the realization has strictly positive probability under some completion of the partially specified strategy profile (In this example, that means $\Pr^{g^i, g^{-i}}(h_t^i) > 0$ for some $g^i$). Whenever we write a conditional probability or conditional expectation, we implicitly assume that the condition has non-zero probability under the specified strategy profile. When only part of the strategy profile is specified in the superscript, we implicitly assume that the condition is admissible under the specified partial strategy profile. In this paper, we make heavy use of value functions and reward-to-go functions. Such functions will be clearly defined within their context with the following convention: $Q$ stands for state-action value functions; $V$ stands for state value functions; and $J$ stands for reward-to-go functions for a given strategy profile (as opposed to $Q$ or $V$, 
{both of which are typically defined via} a maximum over all strategies).

\section{Game Model and Objectives}\label{sec:suffinfo:gamemodel}
\subsection{Game Model}
In this section we formulate a general model for a finite horizon dynamic game with finitely many players. 

Denote the set of players by $\mathcal{I}$. Denote the set of timestamps by $\mathcal{T}=\{1,2,\cdots, T\}$. At time $t$, player $i\in\mathcal{I}$ takes action $U_t^i$, obtains instantaneous reward $R_t^i$, and then learns new information $Z_t^i$. Player $i$ may not necessarily observe the instantaneous rewards $R_t^i$ directly. The reward is observable only if it is part of $Z_t^i$. Define $Z_t=(Z_t^i)_{i\in\mathcal{I}}, U_t=(U_t^i)_{i\in\mathcal{I}}$, and $R_t=(R_t^i)_{i\in\mathcal{I}}$. We assume that there is an underlying state variable $X_t$ and
\begin{align}
	(X_{t+1}, Z_t, R_t) &= f_t(X_t, U_t, W_t), \qquad t\in\mathcal{T},
\end{align}
where $(f_t)_{t\in \mathcal{T}}$ are fixed functions. The primitive random variable $X_1$ represents the initial move of nature. The primitive random vector $H_1=(H_1^i)_{i\in\mathcal{I}}$ {represents} the initial information of the players. The initial state and information $X_1$ and $H_1$ are, in general, correlated. The random variables $(W_t)_{t=1}^T$ are mutually independent primitive random variables representing nature's move. The vector $(X_1, H_1)$ is assumed to be mutually independent with $W_1, W_2, \cdots, W_{T}$. The distributions of the primitive random variables are common knowledge to all players.

{Define $\mathcal{X}_t, \mathcal{U}_t, \mathcal{Z}_t, \mathcal{W}_t, \mathcal{H}_1$ to be the sets of possible values of $X_t, U_t, Z_t, W_t, H_1$ respectively. The sets $\mathcal{X}_t, \mathcal{U}_t, \mathcal{Z}_t, \mathcal{W}_t, \mathcal{H}_1$ are assumed to be common knowledge among all players.} In this work, in order to focus on conceptual difficulties instead of technical issues, we make the following assumption. 

\begin{assump}
	$\mathcal{X}_t, \mathcal{U}_t, \mathcal{Z}_t, \mathcal{W}_t, \mathcal{H}_1$ are finite sets, and $R_t^i$ is supported on $[-1, 1]$.
\end{assump}

We assume perfect recall, i.e. the information player $i$ has at time $t$ is $H_t^i = (H_1^i, Z_{1:t-1}^i)$, and player $i$'s action $U_t^i$ is {contained} in the new information $Z_t^i$. 
A behavioral strategy $g^i=(g_t^i)_{t\in\mathcal{T}}$ of player $i$ is a collection of functions $g_t^i\colon \mathcal{H}_t^i\mapsto \Delta(\mathcal{U}_t^i)$, where $\mathcal{H}_t^i$ is the space where $H_t^i$ takes values. Under a behavioral strategy profile $g=(g^i)_{i\in\mathcal{I}}$, the total reward/payoff of player $i$ in this game is given by
\begin{equation}
	J^i(g):=\E^g\left[\sum_{t=1}^{T} R_t^i \right].
\end{equation}

\begin{remark}
	This is not a restrictive model: By choosing appropriate state representation $X_t$ and instantaneous reward vector $R_t$, it can be used to model any finite-node extensive form sequential game with perfect recall.
\end{remark}

We initially consider two solution concepts for dynamic games with asymmetric information: Bayes--Nash Equilibrium (BNE) and Sequential Equilibrium (SE). We define BNE and SE below.

\begin{defn}[Bayes-Nash Equilibrium]
A behavioral strategy profile $g$ is said to form a Bayes-Nash equilibrium (BNE) if for any player $i$ and any behavioral strategy $\tilde{g}^i$ of player $i$, we have $J^i(g)\geq J^i(\tilde{g}^i, g^{-i})$.
\end{defn}

\begin{defn}[Sequential Equilibrium]\label{def:KSEinmaintext}
	Let $g=(g^i)_{i\in\mathcal{I}}$ be a behavioral strategy profile. Let $\Qfunc=(\Qfunc_t^i)_{i\in\mathcal{I},t\in\mathcal{T}}$ be a collection of history-action value functions, i.e. $\Qfunc_t^i\colon\mathcal{H}_t^i \times \mathcal{U}_t^i \mapsto \mathbb{R}$. The strategy profile $g$ is said to be sequentially rational under $\Qfunc$ if for each $i\in\mathcal{I}, t\in\mathcal{T}$ and each $h_t^i\in\mathcal{H}_t^i$, 
	\begin{equation}
		\mathrm{supp}(g_t^i(h_t^i))\subseteq \underset{u_t^i}{\arg\max}~ \Qfunc_t^i(h_t^i, u_t^i).
	\end{equation}
	
	$\Qfunc$ is said to be fully consistent with $g$ if there exist a sequence of pairs of strategies and history-action value functions $(g^{(n)}, \Qfunc^{(n)})_{n=1}^\infty$ such that
	\begin{enumerate}[(1)]
		\item $g^{(n)}$ is fully mixed, i.e. every action is chosen with positive probability at every information set.
		\item $\Qfunc^{(n)}$ is consistent with $g^{(n)}$, i.e.,
		\begin{align}
			\Qfunc_\tau^{(n), i}(h_\tau^i, u_\tau^i)&=\E^{g^{(n)}}\left[\sum_{t=\tau}^T R_t^i\Big|h_\tau^i, u_\tau^i\right],
		\end{align}
		for each $i\in\mathcal{I}, \tau\in\mathcal{T}, h_\tau^i\in \mathcal{H}_\tau^i, u_\tau^i\in \mathcal{U}_\tau^i$.
		\item $(g^{(n)}, \Qfunc^{(n)})\rightarrow (g, \Qfunc)$ as $n\rightarrow\infty$.
	\end{enumerate}
	A tuple $(g, \Qfunc)$ is said to be a sequential equilibrium if $g$ is sequentially rational under $\Qfunc$ and $\Qfunc$ is fully consistent with $g$.
\end{defn}	

{Whereas Definition \ref{def:KSEinmaintext} of SE is different from that of \cite{kreps1982sequential}, we show in Appendix \ref{app:SE} that it is equivalent to the concept in \cite{kreps1982sequential}. We use Definition \ref{def:KSEinmaintext} as it is more suitable for the development of our results.}

In this paper, we are interested in analyzing the performance of strategy profiles that are based on some form of compressed information.
Let $\Comp_t^i$ be a function of $H_t^i$ that can be sequentially updated, i.e. there exist functions $(\iota_t^i)_{t\in\mathcal{T}}$ such that
	\begin{align}
		\Comp_1^i &= \iota_1^i(H_1^i),\\
		\Comp_{t}^i &= \iota_{t}^i(\Comp_{t-1}^i, Z_{t-1}^i),\qquad t\in\mathcal{T}\backslash\{1\}.
	\end{align}
Write $\Comp^i=(\Comp_t^i)_{t\in\mathcal{T}}$ and $\Comp=(\Comp^i)_{i\in\mathcal{I}}$. We will refer to $\Comp^i$ as the compression of player $i$'s information under $\iota^i = (\iota_t^i)_{t\in\mathcal{T}}$. A $\Comp^i$-based (behavioral) strategy $\rho^i=(\rho_t^i)_{t\in\mathcal{T}}$ is a collection of functions $\rho_t^i\colon\Compset_t^i\mapsto \Delta(\mathcal{U}_t^i)$. A strategy profile where each player $i$ uses a $\Comp^i$-based strategy is called a $\Comp$-based strategy profile. If a $\Comp$-based strategy profile forms an Bayes-Nash (resp. sequential) equilibrium, then it is called a $\Comp$-based Bayes-Nash (resp. sequential) equilibrium. Note that unlike \cite{piccione1997interpretation,battigalli1997dynamic,GROVE199751,HALPERN199766,AUMANN1997102}, we require the $\Comp$-based BNE and $\Comp$-based SE to contain no profitable deviation among \emph{all full-history-based} strategies.

\subsection{Objectives}
Our goal is to discover properties/features of the compressed information $\Comp$ sufficient to guarantee that (i) there exists $\Comp$-based BNE and SE; (ii) the set of $\Comp$-based BNE (resp. SE) payoff profiles is equal to the set of (general strategy based) BNE (resp. SE) profiles under perfect recall. 

To achieve the above-stated objectives we proceed as follows: First, we introduce two notions of information state, namely MSI and USI (Section \ref{sec:twoinfostates}). Then, we investigate the existence of MSI-based and USI-based BNE and SE, as well as the preservation of the set of all BNE and SE payoff profiles when USI-based strategies are employed by all players (Section \ref{sec:isbe}).

\begin{remark}
{A key challenge in achieving the above-stated goal is the following: Unlike the case of perfect recall,} one may not be able to recover $\Comp_{t-1}^i$ from $\Comp_t^i$. 
Therefore, $\Comp^i$-based (behavioral) strategies are not equivalent to mixed strategies supported on the set of $\Comp^i$-based pure strategies. This fact creates difficulty for analyzing $\Comp^i$-based strategies since the standard technique of using Kuhn's Theorem \citep{kuhn2016extensive} to transform mixed strategies to behavioral strategies does not apply. To resolve this challenge, we developed stochastic control theory-based techniques that allow us to work with $\Comp^i$-based behavioral strategies directly rather than transforming from a mixed strategy. 
\end{remark}

\begin{remark}
    In the following sections, when referring to the compressed information $\Comp_t^i$, we will consider the compression mappings $\iota^i$ to be fixed and given, so that $\Comp_t^i$ is fixed given $H_t^i$. The space of compressed information $\Compset_t^i$ is a fixed, finite set given $\iota^i$. When we use $\comp_t^i$ to represent a realization of $\Comp_t^i$, we assume that it corresponds to the compression of $H_t^i=h_t^i$ under the fixed $\iota^i$.
\end{remark}

\section{Two Definitions of Information State}\label{sec:twoinfostates}
Before we define notions of information state in dynamic games we introduce the notion of information state for one player when other players' strategies are fixed. The following definition is an extension of the definition of information state in \cite{subramanian2020approximate}.

\begin{defn}\label{def:infostate}
	Let $g^{-i}$ be a behavioral strategy profile of players other than $i$. We say that $\Comp^i$ is an \emph{information state under $g^{-i}$} if there exist functions $(P_t^{i,g^{-i}})_{t\in\mathcal{T}}, (r_t^{i,g^{-i}})_{t\in\mathcal{T}}$, where $P_t^{i,g^{-i}}\colon\Compset_t^i\times \mathcal{U}_t^i \mapsto \Delta(\Compset_{t+1}^i)$ and $r_t^{i,g^{-i}}\colon\Compset_t^i\times \mathcal{U}_t^i \mapsto [-1, 1]$, such that
	\begin{enumerate}[(1)]
		\item $\Pr^{g^i, g^{-i}}(\comp_{t+1}^i|h_t^i, u_t^i) = P_t^{i,g^{-i}}(\comp_{t+1}^i|\comp_t^i, u_t^i)$ for all $t\in\mathcal{T}\backslash\{T\}$;
		\item $\E^{g^i, g^{-i}}[R_t^i|h_t^i, u_t^i] = r_t^{i,g^{-i}}(\comp_t^i, u_t^i)$ for all $t\in \mathcal{T}$,
	\end{enumerate}
	for all $g^i$, and all $(h_t^i, u_t^i)$ admissible under $(g^i, g^{-i})$. (Both $P_t^{i,g^{-i}}$ and $r_t^{i,g^{-i}}$ may depend on $g^{-i}$, but they do not depend on $g^i$.)
\end{defn}

In the absence of other players, the above definition is exactly the same as the definition of information state for player $i$'s control problem. When other players are present, the parameters of player $i$'s control problem, in general, depend on the strategy of other players. As a consequence, an information state under one strategy profile $g^{-i}$ may not be an information state under a different strategy profile $\tilde{g}^{-i}$.

\subsection{Mutually Sufficient Information}\label{sec:msi}
\begin{defn}[Mutually Sufficient Information]\label{def:msi}
	We say that $\Comp=(\Comp^i)_{i\in\mathcal{I}}$ is \emph{mutually sufficient information} (MSI) if for all players $i\in\mathcal{I}$ and all $\Comp^{-i}$-based strategy profiles $\rho^{-i}$, $\Comp^i$ is an information state under $\rho^{-i}$.
\end{defn}

In words, MSI represents mutually consistent compression of information in a dynamic game: Player $i$ could compress their information to $\Comp^i$ without loss of performance when other players are compressing their information to $\Comp^{-i}$. Note that MSI imposes interdependent conditions on the compression maps of all players: It requires the compression maps of all players to be consistent with each other.

The following lemma provides a sufficient condition for a compression maps to yield mutually sufficient information.

\begin{lemma}\label{lem:msi}
	If for all $i\in\mathcal{I}$ and all $\Comp^{-i}$-based strategy profiles $\rho^{-i}$, there exist functions $(\Phi_t^{i, \rho^{-i}})_{t\in\mathcal{T}}$ where $\Phi_t^{i, \rho^{-i}}\colon\Compset_t^i \mapsto \Delta(\mathcal{X}_t\times \Compset_t^{-i})$ such that
	\begin{equation}
		\Pr^{g^i, \rho^{-i}}(x_t, \comp_t^{-i}|h_t^i) = \Phi_t^{i,\rho^{-i}}(x_t, \comp_t^{-i}|\comp_t^i),
	\end{equation}
	for all behavioral strategies $g^i$, all $t\in\mathcal{T}$, and all $h_t^i$ admissible under $(g^i, \rho^{-i})$, then $\Comp=(\Comp^i)_{i\in\mathcal{I}}$ is mutually sufficient information.
\end{lemma}
\begin{proof}
    See Appendix \ref{app:lem:msi}.
\end{proof}

In words, the condition of Lemma \ref{lem:msi} means that $\Comp_t^i$ has the same predictive power as $H_t^i$ in terms of forming a belief on the current state and other players' compressed information \emph{whenever other players are using compression-based strategies}. This belief is sufficient for player $i$ to predict other player's actions and future state evolution. Since other players are using compression-based strategies, player $i$ does not have to form a belief on other player's full information in order to predict other players' actions.

\subsection{Unilaterally Sufficient Information}\label{sec:usi}

\begin{defn}[Unilaterally Sufficient Information]\label{def:usi}
	We say that $\Comp^i$ is \emph{unilaterally sufficient information} (USI) for player $i\in\mathcal{I}$ if there exist functions $(F_t^{i, g^i})_{t\in\mathcal{T}}$ and $(\Phi_t^{i, g^{-i}})_{t\in\mathcal{T}}$ where $F_t^{i, g^i}\colon\Compset_t^i\mapsto \Delta(\mathcal{H}_t^i), \Phi_t^{i, g^{-i}}\colon\Compset_t^i \mapsto \Delta(\mathcal{X}_t\times \mathcal{H}_t^{-i})$ such that
	\begin{equation}\label{eq:usicondition}
		\Pr^{g}(x_t, h_t|\comp_t^i) = F_t^{i, g^i}(h_t^i|\comp_t^i)\Phi_t^{i,g^{-i}}(x_t, h_t^{-i}|\comp_t^i),
	\end{equation}
	for all behavioral strategy profiles $g$, all $t\in\mathcal{T}$, and all $\comp_t^i$ admissible under $g$.\footnote{In the case where random vectors $X_t$, $H_t^i$ and $H_t^{-i}$ share some common components, \eqref{eq:usicondition} should be interpreted in the following way: $x_t$, $h_t^i$ and $h_t^{-i}$ are three separate realizations that are not necessarily congruent with each other (i.e. they can disagree on their common parts). In the case of incongruency, the left-hand side equals 0. The equation needs to be true for all combinations of $x_t\in \mathcal{X}_t$, $h_t^i\in\mathcal{H}_t^i$ and $h_t^{-i}\in\mathcal{H}_t^{-i}$.}
\end{defn}

The definition of USI can be separated into two parts: The first part states that the conditional distribution of $H_t^i$, player $i$'s full information, given $\Comp_t^i$, the compressed information, does not depend on other players' strategies. This is similar to the idea of sufficient statistics in the statistics literature \citep{kay1993fundamentals}: If player $i$ would like to use their ``data'' $H_t^i$ to estimate the ``parameter'' $g^{-i}$, then $\Comp_t^i$ is a sufficient statistic for this parameter estimation problem. The second part states that $\Comp_t^i$ has the same predictive power as $H_t^i$ in terms of forming a belief on the current state and other players' full information. 
In contrast to the definition of mutually sufficient information, if $\Comp^i$ is unilaterally sufficient information, then $\Comp^i$ is sufficient for player $i$'s decision making regardless of whether other players are using any information compression map.

\subsection{Comparison}
Using Lemma \ref{lem:msi} it can be shown that if $\Comp^i$ is USI for each $i\in\mathcal{I}$, then $\Comp=(\Comp^i)_{i\in\mathcal{I}}$ is MSI. The {converse} is not true. The following example illustrates the difference between MSI and USI. 
\begin{example}\label{ex:msinotusi}
	Consider a two stage stateless (i.e. $X_t=\varnothing$) game of two players: Alice (A) moves first and Bob (B) moves afterwards. There is no initial information (i.e. $H_1^A=H_1^B=\varnothing$).
	
	At time $t=1$, Alice chooses $U_1^A\in \{0, 1\}$. The instantaneous rewards of both players are given by
	\begin{equation}
		R_1^A = U_1^A, R_1^B = -U_1^A.
	\end{equation}
	
	The new information of both Alice and Bob at time $1$ is $Z_1^A=Z_1^B = U_1^A$, i.e. Alice's action is observed.
	
	At time $t=2$, Bob chooses $U_2^B \in \{-1, 1\}$.
	The instantaneous rewards of both players are given by
	\begin{equation}
		R_2^A = U_2^B, R_2^B = 0.
	\end{equation}
	
	Set $\Comp_t^A=H_t^A$ and $\Comp_t^B=\varnothing$ for both $t\in \{1, 2\}$. It can be shown that $\Comp$ is mutually sufficient information. However, $\Comp^B$ is not unilaterally sufficient information: We have $\Pr^g(h_2^B|\comp_2^B)=\Pr^g(u_1^A) = g_1^A(u_1^A|\varnothing)$, while the definition of USI requires that $\Pr^g(h_2^B|\comp_2^B) = F_t^{B, g^B}(h_2^B|\comp_2^B)$ for some function $F_t^{B, g^B}$ that does not depend on $g^A$.
\end{example}

% ==========================================================================================================================

\section{Information-State Based Equilibrium}\label{sec:isbe}
In this section, we formulate our result on MSI and USI based equilibria for two equilibrium concepts: Bayes--Nash equilibria and sequential equilibria. 

\subsection{Information-State Based Bayes--Nash Equilibrium}
\begin{thm}\label{thm:msiexist}
	If $\Comp$ is mutually sufficient information, then there exists at least one $\Comp$-based BNE.
\end{thm}
\begin{proof}
    See Appendix \ref{app:thm:msiexist}.
\end{proof}

The main idea for the proof of Theorem \ref{thm:msiexist} is {the definition of} a best-response correspondence through the dynamic program for an underlying single-agent control problem.

\begin{thm}\label{thm:usiequiv}
	If $\Comp=(\Comp^i)_{i\in\mathcal{I}}$ where $\Comp^i$ is unilaterally sufficient information for player $i$, then the set of $\Comp$-based BNE payoffs is the same as that of all BNE.
\end{thm}
\begin{proof}
    See Appendix \ref{app:thm:usiequiv}.
\end{proof}

The intuition behind Theorem \ref{thm:usiequiv} is that one can think of player $i$'s information {that is not included in} the unilaterally sufficient information $\Comp_t^i$ as a private randomization device for player $i$: When player $i$ is using a strategy that depends on their information outside of $\Comp_t^i$, it is as if they are using a randomized $\Comp^i$-based strategy. The main idea for the proof of Theorem \ref{thm:usiequiv} is to show that for every BNE strategy profile $g$, player $i$ can switch to an ``equivalent'' randomized $\Comp^i$-based strategy $\rho^i$ while maintaining the equilibrium and payoffs.\footnote{Besides the connection of USI to sufficient statistics, the idea behind the construction of the equivalent $\Comp^i$-based strategy is also closely related to the idea of the Rao--Blackwell estimator \citep{kay1993fundamentals}, where a new estimator is obtained by taking the conditional expectation of the old estimator given the sufficient statistics.} The theorem then follows from iteratively switching the strategy of each player.

Example \ref{ex:msinotusi} can also be used to illustrate that when $\Comp$ is an MSI but not an USI, $\Comp$-based BNE exist but $\Comp$-based strategies do not attain all equilibrium payoffs. 
\theoremstyle{thmstylethree}
\newtheorem*{example*}{Example \ref{ex:msinotusi}}
\begin{example*}[Continued]
    {In this example, $\Comp_t^A = H_t^A, \Comp_t^B = \varnothing$ for $t=1,2$ is MSI. Furthermore, it can be shown that the following strategy profiles are BNE of the game: (E1) Alice plays $U_1^A=1$ at time 1 and Bob plays $U_2^B=1$ irrespective of Alice's action at time 1; and (E2) Alice plays $U_1^A=0$ at time 1; Bob plays $U_2^B = 1$ if $U_1^A=0$ and $U_2^B=-1$ if $U_1^A=1$. Equilibrium (E1) is a $\Comp$-based equilibrium. However, (E2) cannot be attained by $\Comp$-based strategy profile for the following reason: In any $\Comp$-based equilibrium, Bob plays the same mixed strategy irrespective of Alice's action and his expected payoff at the end of the game is $-1$. At (E2), Bob's expected payoff at the end of the game is $0$. Therefore, the payoff at (E2) cannot be attained by any $\Comp$-based strategy profile.} 
\end{example*}

\subsection{Information-State Based Sequential Equilibrium}\label{sec:sieqref} 
\begin{thm}\label{thm:msiseexist}
	If $\Comp$ is mutually sufficient information, then there exists at least one $\Comp$-based sequential equilibrium.
\end{thm}
\begin{proof}
    See Appendix \ref{app:thm:msiseexist}.
\end{proof}

The proof of Theorem \ref{thm:msiseexist} follows steps similar to that of Theorem \ref{thm:msiexist}. The difference is that we explicitly construct a sequence of conjectured history-action value functions $\Qfunc^{(n)}$ (as defined in Definition \ref{def:KSEinmaintext}) using the dynamic program of player $i$'s decision problem. Then we argue that the strategies and the conjectures satisfies Definition \ref{def:KSEinmaintext}.

\begin{thm}\label{thm:usiseequiv}
	If $\Comp=(\Comp^i)_{i\in\mathcal{I}}$ where $\Comp^i$ is unilaterally sufficient information for player $i$, then the set of $\Comp$-based sequential equilibrium payoffs is the same as that of all sequential equilibria.
\end{thm}
\begin{proof}
    See Appendix \ref{app:thm:usiseequiv}.
\end{proof}

The proof of Theorem \ref{thm:usiseequiv} mostly follows the same ideas for Theorem \ref{thm:usiequiv}: for each sequential equilibrium strategy profile $g$, we construct an ``equivalent'' $\Comp^i$-based strategy $\rho^i$ for player $i$ with similar construction as in Theorem \ref{thm:usiequiv}. The critical part is to show that $\rho^i$ is still sequentially rational under the concept of sequential equilibrium. 

% ===============================================================================================================================

\section{Discussion}\label{sec:siapplications}
{In this section we first investigate if USI can preserve the set of equilibrium payoffs achievable under perfect recall when refinements of BNE other than SE, namely, various versons of Perfect Bayesian Equilibrium (PBE), are considered. Then, we identify MSI and USI in specific models that appeared in the literature.}

\subsection{Other Equilibrium Concepts}
{We first present Example \ref{ex:wpbeproblem} to show that the result of Theorem \ref{thm:usiseequiv} is not true when we replace SE with the concept of weak Perfect Bayesian Equilibrium (wPBE) \citep{mas1995microeconomic} which is a refinement of BNE that is weaker than SE. Then, we discuss how the result of Proposition \ref{prop:exwpbe}, that is, part of Example \ref{ex:wpbeproblem} and appears below, applies or does not apply to other versions of PBE, namely, those defined in \cite{watson2017general} and \cite{battigalli1996strategic}.}

The concept of wPBE is defined as follows: 
Let $(g, \mu)$ be an assessment, where $g$ is a behavioral strategy profile as specified in Section \ref{sec:suffinfo:gamemodel} and $\mu$ is a system of functions representing player's beliefs in the extensive-form game representation. Then, $(g, \mu)$ is said to be a weak perfect Bayesian equilibrium %\vjdeletecomment{Already defined.}{(wPBE)} 
\citep{mas1995microeconomic} if $g$ is sequentially rational to $\mu$ and $\mu$ satisfies Bayes rule with respect to $g$ on the equilibrium path. The concept of wPBE does not impose any restriction on beliefs off the equilibrium path.

\begin{example}\label{ex:wpbeproblem}
	Consider a two-stage game with two players: Bob (B) moves at stage 1; Alice (A) and Bob move simultaneously at stage 2. Let $X_1^A, X_1^B$ be independent uniform random variables on $\{-1, +1\}$ representing the types of the players. The state satisfies $X_1=(X_1^A, X_1^B)$ and $X_2=X_1^B$.  The set of actions are $\mathcal{U}_1^B = \{-1, +1\}$, $\mathcal{U}_2^A = \mathcal{U}_2^B = \{-1, 0, +1\}$. The information structure is given by
	\begin{align}
		H_1^A &= X_1^A,\quad H_1^B = X_1^B;\\
		H_2^A &= (X_1^A, U_1^B),\quad H_2^B = (X_1^B, U_1^B),
	\end{align}
	i.e. types are private and actions are observable.
	
	The instantaneous payoffs of Alice are given by
	\begin{align*}
		R_1^A &= \begin{cases}
			-1,&\text{if }U_1^B = -1;\\
			0,&\text{otherwise},
		\end{cases}\qquad
		R_2^A = \begin{cases}
			1,&\text{if }U_2^A = X_2\text{ or }U_2^A = 0;\\
			0,&\text{otherwise}.
		\end{cases}.
	\end{align*}
	
	The instantaneous payoffs of Bob are given by
	\begin{align*}
		R_1^B &= \begin{cases}
			0.2,&\text{if }U_1^B = -1;\\
			0,&\text{otherwise},
		\end{cases}\qquad
		R_2^B = \begin{cases}
			-1,&\text{if }U_2^A = U_2^B;\\
			0,&\text{otherwise}.
		\end{cases}.
	\end{align*}
	
	Define $\Comp_1^A = X_1^A$ and $\Comp_2^A = U_1^B$. It can be shown that $\Comp^A$ is unilaterally sufficient information for Alice.\footnote{In fact, this example can be seen as an instance of the model described in Example \ref{ex:ouyang} which we introduce later.} Set $\Comp_t^B = H_t^B$, i.e. no compression for Bob's information. Then, $\Comp^B$ is trivially unilaterally sufficient information for Bob.
	
	\begin{prop}\label{prop:exwpbe}
		In the game defined in Example \ref{ex:wpbeproblem}, the set of $\Comp$-based wPBE payoffs is a proper subset of that of all wPBE payoffs.
	\end{prop}
    \begin{proof}
        See Appendix \ref{app:prop:exwpbe}.
    \end{proof}
	
	Note that since any wPBE is first and foremost a BNE, by Theorem \ref{thm:usiequiv}, any general strategy based wPBE payoff profile can be attained by a $\Comp$-based BNE. However, Proposition \ref{prop:exwpbe} implies that there exists a wPBE payoff profile such that none of the $\Comp$-based BNEs attaining this payoff profile are wPBEs.
\end{example}

Intuitively, the reason for some wPBE payoff profiles to be unachievable under $\Comp$-based wPBE payoffs in this example can be explained as follows. The state $X_1^A$ in this game can be thought of as a private randomization device of Alice that is payoff irrelevant (i.e. a private coin flip) that should not play a role in the outcome of the game. However, under the concept of wPBE, the presence of $X_1^A$ facilitates Alice to implement off-equilibrium strategies that are otherwise not sequentially rational. This holds due to the following: For a fixed realization of $U_1^B$, the two realizations of $X_1^A$ give rise to two different information sets. Under the concept of wPBE, if the two information sets are both off equilibrium path, Alice is allowed to form different beliefs and hence justify the use of different mixed actions under different realizations of $X_1^A$. Therefore, the presence of $X_1^A$ can expand Alice's set of ``justifiable'' mixed actions off-equilibrium. By restricting Alice to use $\Comp^A$-based strategies, i.e. choosing her mixed action not depending on $X_1^A$, Alice loses the ability to use some mixed actions off-equilibrium in a ``justifiable'' manner, and hence losing her power to sustain certain equilibrium outcomes. This phenomenon, however, does not happen under the concept of sequential equilibrium, since SE (quite reasonably) would require Alice to use the same belief on two information sets if they only differ in the realization of $X_1^A$.

With similar approaches, one can establish the analogue of Proposition \ref{prop:exwpbe} for the perfect Bayesian equilibrium concept defined in \cite{watson2017general} (which we refer to as ``Watson's PBE''). Simply put, this is since Watson's PBE imposes conditions on the belief update for each pair of successive information states in a separated manner. There exist no restrictions \emph{across} different pairs of successive information states. As a result, for a fixed realization of $U_1^B$, Alice is allowed to form different beliefs under two realizations of $X_1^A$ just like under wPBE as long as both beliefs are reasonable on their own. In fact, in the proof of Proposition \ref{prop:exwpbe}, the two off-equilibrium belief updates both satisfy Watson's condition of plain consistency \citep{watson2017general}. 

Approaches similar to those in the proof of Proposition \ref{prop:exwpbe}, however, do not apply to the PBE concept defined with the independence property of conditional probability systems specified in \cite{battigalli1996strategic} (which we refer to as ``Battigalli's PBE''). In fact, Battigalli's PBE is equivalent to sequential equilibrium if the dynamic game consists of only two strategic players \citep{battigalli1996strategic}. We conjecture that in general games with three or more players, if $\Comp$ is USI, then the set of all $\Comp$-based Battigalli's PBE payoffs is the same as that of all Battigalli's PBE payoffs. However, establishing this result can be difficult due to the complexity of Battigalli's conditions. 

\subsection{Information States in Specific Models}
In this section, we identify MSI and USI in specific game models studied in the literature. Whereas we recover some existing results using our framework, we also develop some new results.

\begin{example}\label{ex:repeatedgames}
	Consider stateless dynamic games with observable actions, i.e. $X_t=\varnothing, H_1^i=\varnothing, Z_t^i=U_t$ for all $i\in\mathcal{I}$. One instance of such games is the class of repeated games \citep{fudenberg1991game}. In this game, $H_t^i=U_{1:t-1}$ for all $i\in\mathcal{I}$. Let $(\iota_t^0)_{t\in\mathcal{T}}$ be an arbitrary, common update function and let $\Comp^i=\Comp^0$ be generated from $(\iota_t^0)_{t\in\mathcal{T}}$. Then $\Comp$ is mutually sufficient information since Lemma \ref{lem:msi} is trivially satisfied. As a result, Theorem \ref{thm:msiexist} holds for $\Comp$, i.e. there exist at least one $\Comp$-based BNE.
	
	However, in general, $\Comp$ is not unilaterally sufficient information. To see that, one can consider the case when player $j\neq i$ is using a strategy that chooses different mixed actions for different realizations of $U_{1:t-1}$. In this case $\Pr^{g^i, g^{-i}}(\tilde{\comp}_{t+1}^i|h_t^i, u_t^i)$ would potentially depend on $U_{1:t-1}$ as a whole. This means that $\Comp^i$ is not an information state for player $i$ under $g^{-i}$, which violates Lemma \ref{lem:usiisinfostate}.
	
	Furthermore for $\Comp$, the result of Theorem \ref{thm:usiequiv} does not necessarily hold, i.e. the set of $\Comp$-based BNE payoffs may not be the same as that of all BNE. Example \ref{ex:msinotusi} can be used to show this.
	
\end{example}

\begin{example}\label{ex:maskintirole}
	The model of \citep{maskin2001markov} is a special case of our dynamic game model where $Z_t^i=(X_{t+1}, U_t)$, i.e. the (past and current) states and past actions are observable. In this case, $\Comp=(\Comp_t^i)_{t\in\mathcal{T}, i\in\mathcal{I}}$ with $\Comp_t^i=X_t$ is mutually sufficient information; note that $H_t^i=(X_{1:t}, U_{1:t-1})$. Consider a $\Comp^{-i}$-based strategy profile $\rho^{-i}$, i.e. $\rho_t^j\colon\mathcal{X}_t\mapsto \Delta(\mathcal{U}_t^j)$ for $t\in\mathcal{T}, j\in\mathcal{I}\backslash\{i\}$. We have
	\begin{align}
		\Pr^{g^i, \rho^{-i}}(\tilde{x}_t, \tilde{\comp}_t^{-i}|h_t^i) &=\Pr^{g^i, \rho^{-i}}(\tilde{x}_t, \tilde{\comp}_t^{-i}|x_{1:t}, u_{1:t-1})\\
		&= \bm{1}_{\{\tilde{x}_t = x_t\} } \prod_{j\neq i} \bm{1}_{\{\tilde{\comp}_t^j = x_t\} }\\
		&=:\Phi_t^{i, \rho^{-i}}(\tilde{x}_t, \tilde{\comp}_t^{-i}|x_t).
	\end{align}
	
	Hence $\Comp$ is mutually sufficient information by Lemma \ref{lem:msi}. As a result, there exists at least one $\Comp$-based BNE.
	
	Similar to Example \ref{ex:repeatedgames}, in general, $\Comp$ is not unilaterally sufficient information, and the set of $\Comp$-based BNE payoffs may not be the same as that of all BNE. The argument for both claims can be carried out in an analogous way to Example \ref{ex:repeatedgames}.
\end{example}

\begin{example}\label{ex:nayyar}
	The model of \cite{nayyar2013common} is a special case of our dynamic model satisfying the following conditions.
	\begin{enumerate}[(1)]
		\item The information of each player $i$ can be separated into the common information $H_t^0$ and private information $L_t^i$, i.e. there exists a strategy-independent bijection between $H_t^i$ and $(H_t^0, L_t^i)$ for all $i\in\mathcal{I}$. 
		\item The common information $H_t^0$ can be sequentially updated, i.e.
		\begin{align}
			H_{t+1}^0 &= (H_t^0, Z_t^0),
		\end{align}
		where $Z_t^0 = \bigcap_{i\in\mathcal{I}} Z_t^i$ is the common part of the new information of all players at time $t$.
		
		\item The private information $L_t^i$ can be sequentially updated, i.e. there exist functions $(\zeta_t^i)_{t=0}^{T-1}$ such that
		\begin{align}
			L_{t+1}^i &= \zeta_t^i(L_t^i, Z_t^i).
		\end{align}
	\end{enumerate}

	In \cite{nayyar2013common}, the authors impose the following assumption.
	\begin{assump}[Strategy independence of beliefs]\label{assump:cibstraindep}
		There exist a function $P_t^0$ such that
		\begin{align}
			\Pr^{g}(x_t, l_t|h_t^0) = P_t^0(x_t, l_t|h_t^0),
		\end{align}
		for all behavioral strategy profiles $g$ whenever $\Pr^g(h_t^0) > 0$, where $l_t=(l_t^i)_{i\in\mathcal{I}}$.
	\end{assump}
	
	In this model, if we set $\Comp_t^i=(\Pi_t, L_t^i)$ where $\Pi_t \in \Delta(\mathcal{X}_t\times \mathcal{S}_t)$ is a function of $H_t^0$ defined through
	\begin{equation}
		\Pi_t(x_t, l_t) :=  P_t^0(x_t, l_t|H_t^0),
	\end{equation}
	then $\Comp=(\Comp^i)_{i\in\mathcal{I}}$ is mutually sufficient information. First note that $\Comp_t^i$ can be sequentially updated as $\Pi_t$ can be sequentially updated using Bayes rule. Then
	\begin{align}
		\Pr^{g^i, \rho^{-i}}(\tilde{x}_t, \tilde{l}_t^{-i}|h_t^i) &= \Pr^{g^i, \rho^{-i}}(\tilde{x}_t, \tilde{l}_t^{-i}|h_t^0, l_t^i)\\ &=\dfrac{\Pr^{g^i, \rho^{-i}}(\tilde{x}_t, l_t^i, \tilde{l}_t^{-i}|h_t^0)}{\Pr^{g^i, \rho^{-i}}(l_t^i|h_t^0)}\label{eq:exnayyar:1} \\
	    &=\dfrac{P_t^{0}(\tilde{x}_t, (l_t^i, \tilde{l}_t^{-i})|h_t^0)}{\sum_{\hat{x}_t, \hat{l}_t^{-i}} P_t^0(\hat{x}_t, (l_t^i, \hat{l}_t^{-i})|h_t^0)}\label{eq:exnayyar:2} \\
		&=\dfrac{ \pi_t(\tilde{x}_t, (l_t^i, \tilde{l}_t^{-i})) }{\sum_{\hat{x}_t,\hat{l}_t^{-i}} \pi_t(\hat{x}_t, (l_t^i, \hat{l}_t^{-i})) }\label{eq:exnayyar:3} \\
		&=:\tilde{\Phi}_t^{i, \rho^{-i}}(\tilde{x}_t, \tilde{l}_t^{-i}|\comp_t^i),
	\end{align}
	for some function $\tilde{\Phi}_t^{i, \rho^{-i}}$, where $\pi_t$ is the realization of $\Pi_t$ corresponding to $H_t^0=h_t^0$. {In steps \eqref{eq:exnayyar:1} and \eqref{eq:exnayyar:2} we apply Bayes rule on the conditional probabilities given $h_t^0$, and we use Assumption \ref{assump:cibstraindep} to express the belief with the strategy-independent function $P_t^0$.} %\tdwcomment{Add explanations. } % ::DONE 05/29/2024
	
	Note that $\Comp_t^{-i}$ is {contained in the vector} $(\Comp_t^i, L_t^{-i})$, hence we conclude that
	\begin{align}
		\Pr^{g^i, \rho^{-i}}(\tilde{x}_t, \tilde{\comp}_t^{-i}|h_t^i)&=:\Phi_t^{i, \rho^{-i}}(\tilde{x}_t, \tilde{\comp}_t^{-i}|\comp_t^i),
	\end{align}
	for some function $\Phi_t^{i, \rho^{-i}}$. By Lemma \ref{lem:msi} we conclude that $\Comp$ is mutually sufficient information. Therefore there exists at least one $\Comp$-based BNE.
	
	Similar to Examples \ref{ex:repeatedgames} and \ref{ex:maskintirole}, in general, $\Comp$ is not unilaterally sufficient information, and the set of $\Comp$-based BNE payoffs may not be the same as that of all BNE. The argument for both claims can be carried out in an analogous way to Examples \ref{ex:repeatedgames} and \ref{ex:maskintirole}.
\end{example}

\begin{example}\label{ex:ouyang}
	The following model is a variant of \cite{ouyang2016dynamic} and \cite{vasal2019spbe}.
	\begin{itemize}
		\item Each player $i$ is associated with a local state $X_t^i$, and $X_t=(X_t^i)_{i\in\mathcal{I}}$.
		\item Each player $i$ is associated with a local noise process $W_t^i$, and $W_t=(W_t^{i})_{i\in\mathcal{I}}$.
		\item There is no initial information, i.e. $H_1^i=\varnothing$ for all $i\in\mathcal{I}$.
		\item There is a public noisy observation $Y_t^i$ of the local state. The state transitions, observation processes, and reward generation processes satisfy the following:
		\begin{align}
			(X_{t+1}^i, Y_t^i) &= f_t^i(X_t^i, U_t, W_t^i),\quad\forall i\in\mathcal{I},\\
			R_t^i &= r_t^i(X_t, U_t),\quad\forall i\in\mathcal{I}.
		\end{align}
		\item The information player $i$ has at time $t$ is $H_t^i=(Y_{1:t-1}, U_{1:t-1}, X_{1:t}^i)$ for $i\in\mathcal{I}$, where $Y_t=(Y_t^i)_{i\in\mathcal{I}}$.
		\item All the primitive random variables, i.e. the random variables in the collection $(X_1^i)_{i\in\mathcal{I}}\cup (W_t^{i})_{i\in\mathcal{I}, t\in\mathcal{T}}$, are mutually independent.
	\end{itemize}
\end{example}

\begin{prop}\label{thm:ouyangusi}
	In the model of Example \ref{ex:ouyang}, $\Comp_t^i=(Y_{1:t-1}, U_{1:t-1}, X_t^i)$ is unilaterally sufficient information.\footnote{$\Comp^i$-based strategies in this setting are closely related to the ``strategies of type $s$'' defined in \cite{vasal2019spbe}. In \cite{vasal2019spbe}, the authors showed that strategy profiles of type $s$ can attain all equilibrium payoffs attainable by general strategy profiles. However, the authors did not show that strategy profiles of type $s$ can do so while being an equilibrium.}
\end{prop}
\begin{proof}
    See Appendix \ref{app:thm:ouyangusi}.
\end{proof}

Finally, we note that the concept of USI is useful in the context of games among teams as well. {We omit the details of the following example due to its complicated nature.}

\begin{example}
    In the model of games among teams with delayed intra-team information sharing {analyzed} in \cite{tang2022dynamic}, the authors defined the notion of sufficient private information (SPI). It can be shown (through the arguments in {\cite[Section 4.3]{tang2022dynamic} and \cite[Chapters 4.6.1 and 4.6.2]{tangthesis}}) that $K_t^i = (H_t^0, S_t^i)$, which consists of the common information $H_t^0$ and the SPI $S_t^i$, is unilaterally sufficient information.
\end{example}

\section{An Open Problem}\label{sec:openproblems}
Identifying strategy-dependent compression maps that guarantee existence of at least one equilibrium (BNE or SE) or maintain all equilibria that exist under perfect recall is an open problem.

{A known strategy-dependent compression map is one that compress separately first each agent's private information, (resulting in ``sufficient private information''), and then the agents' common information (resulting in ``common information based (CIB) beliefs''  on the system state and the agents' sufficient private information \citep{ouyang2015oligopoly,ouyang2016dynamic,vasal2019spbe,tang2022dynamic,ouyang2024approach}). 
Such a compression does not possess any of the properties of the strategy-independent compression maps that result in MSI or USI. The following example presents a game where belief-based equilibria, i.e. equilibrium strategy profiles based on the above-described compression, do not exist.}

\begin{example}\label{ex:zerosumobsac}
	{Consider the following two-stage zero-sum game. The players are Alice (A) and Bob (B). Alice acts at stage $t=1$ and Bob at stage $t=2$. The game's}
    initial state $X_1$ is distributed uniformly at random on $\{-1, +1\}$. Let $H_t^A, H_t^B$ denote Alice's and Bob's information at stage $t$, and $U_t^A, U_t^B$ denote Alice's and Bob's actions at stage $t$, $t=1,2$.
    We assume that $H_1^A=X_1, H_1^B=\varnothing$, i.e. Alice knows $X_1$ and Bob does not. At stage $t=1$, Alice chooses $U_1^A\in\{-1, 1\}$, and the state transition is given by $X_2 = X_1 \cdot U_1^A$. At stage $t=2$, we assume that $H_2^A=(X_{1:2}, U_1^A)$ and $H_2^B=U_1^A$, i.e. Bob observes Alice's action but not the state before or after Alice's action. At time $t=2$, Bob picks an action $U_2^B \in\{\mathrm{U}, \mathrm{D}\}$. Alice's instantaneous rewards are given by
	\begin{align}
		R_1^A = \begin{cases}
			c&\text{if }U_1^A=+1;\\
			0&\text{if }U_1^A=-1,
		\end{cases}\qquad \text{and} \qquad
		R_2^A = \begin{cases}
			2&\text{if }X_2=+1, U_2^B=\mathrm{U};\\
			1&\text{if }X_2=-1, U_2^B=\mathrm{D};\\
			0&\text{otherwise},
		\end{cases}
	\end{align}
	where $c\in (0, 1/3)$. The stage reward for Bob is $R_t^B = -R_t^A$ for $t=1,2$.
	
	The above game is a signaling game which can be represented in extensive form as in Figure \ref{fig:zerosumobsac}.
	
	\begin{figure}[!ht]
		\centering\small
		\begin{tikzpicture}[scale=0.7]
			\draw[thick] (1, 1) -- (7, 1);
			\draw[thick] (1, 5) -- (7, 5);
			\draw[thick] (4, 1) node[anchor=north] {Alice} -- (4, 5) node[anchor=south] {Alice};
			
			\draw[thick] (1, 1) -- (-1, 0) node[anchor=east] {$(1, -1)$};
			\draw[thick] (1, 1) -- (-1, 2) node[anchor=east] {$(0, 0)$};
			
			\draw[thick] (1, 5) -- (-1, 4) node[anchor=east] {$(0, 0)$};
			\draw[thick] (1, 5) -- (-1, 6) node[anchor=east] {$(2, -2)$};
			
			\draw[thick] (7, 1) -- (9, 0) node[anchor=west] {$(c, -c)$};
			\draw[thick] (7, 1) -- (9, 2) node[anchor=west] {$(2+c, -2-c)$};
			
			\draw[thick] (7, 5) -- (9, 4) node[anchor=west] {$(1+c, -1-c)$};
			\draw[thick] (7, 5) -- (9, 6) node[anchor=west] {$(c, -c)$};
			
			\foreach \point in {(1, 1), (1, 5), (7, 1), (7, 5), (4, 1), (4, 5), (4, 3), (-1, 0), (-1, 2), (-1, 4), (-1, 6), (9, 0), (9, 2), (9, 4), (9, 6)}
			\draw[fill=black] \point circle [radius=0.07];
			
			\node[anchor=east] at (4.1, 2) {$+1$};
			\node[anchor=west] at (3.9, 2) {$[0.5]$};
			\node[anchor=west] at (4, 3) {N};
			\node[anchor=east] at (4.1, 4) {$-1$};
			\node[anchor=west] at (3.9, 4) {$[0.5]$};
			
			\draw[dotted] (1, 1) -- (1, 5);
			\node[anchor=west] at (0.9, 3) {Bob};
			\draw[dotted] (7, 1) -- (7, 5);
			\node[anchor=west] at (6.9, 3) {Bob};
			
			\node[anchor=south] at (2.5, 0.9) {$-1$};
			\node[anchor=north] at (2.5, 5.1) {$-1$};
			\node[anchor=south] at (5.5, 0.9) {$+1$};
			\node[anchor=north] at (5.5, 5.1) {$+1$};
			
			\node[anchor=north] at (0, 0.5) {D};
			\node[anchor=south] at (0, 1.5) {U};
			\node[anchor=north] at (0, 4.5) {D};
			\node[anchor=south] at (0, 5.5) {U};
			
			\node[anchor=north] at (8, 0.5) {D};
			\node[anchor=south] at (8, 1.5) {U};
			\node[anchor=north] at (8, 4.5) {D};
			\node[anchor=south] at (8, 5.5) {U};
		\end{tikzpicture}
		\caption{Extensive form of the game in Example \ref{ex:zerosumobsac}.}\label{fig:zerosumobsac}
	\end{figure}
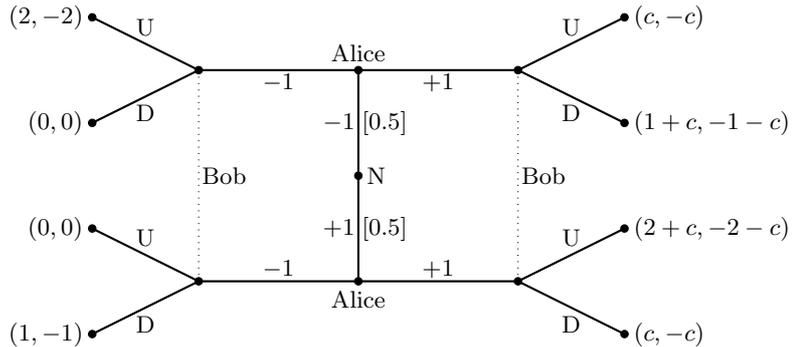

	In order to {define the concept of} belief based equilibrium for this game, we specify the common information $H_t^0$, along with Alice's and Bob's private information, denoted by $L_t^A, L_t^B$, respectively, for $t=1,2$ as follows:
	\begin{align}
		H_1^0 &= \varnothing,\quad L_1^A = X_1,\quad L_1^B = \varnothing,\\
		H_2^0 &= U_1^A,\quad L_2^A = X_2,\quad L_2^B = \varnothing.
	\end{align}
\end{example}

{We prove the following result.}

\begin{prop}\label{prop:zerosumobsac}
	In the game of Example \ref{ex:zerosumobsac} belief-based equilibria do not exist.
\end{prop}
\begin{proof}
    See Appendix \ref{app:prop:zerosumobsac}.
\end{proof}

\section{Conclusion}\label{sec:suffinfo:conclusion}
In this paper, we investigated sufficient conditions for strategy-independent compression maps to be viable in dynamic games. Motivated by the literature on information states for control problems \citep{kumar1986stochastic,mahajan2016decentralized,subramanian2020approximate}, we provided two notions of information state, {both resulting in from strategy-independent information} compression maps for dynamic games, namely mutually sufficient information (MSI) and unilaterally sufficient information (USI). While MSI guarantees the existence of compression-based equilibria, USI guarantees that compression-based equilibria can attain all equilibrium payoff profiles that are achieved when all agents have perfect recall. We established the results under both the concepts of Bayes-Nash equilibrium and sequential equilibrium. We discussed how USI does not guarantee the preservation of payoff profiles under certain other equilibrium refinements. {We considered a strategy-depedent compression map that results in sufficient private information, for each agent, along with a CIB belief. We showed, by an example, that this information compression map does not possess any of the properties of the strategy-independent compression maps that result in MSI or USI.}

{The discovery of strategy-dependent information compression maps that lead to results similar to those of Theorem \ref{thm:msiexist} and \ref{thm:msiseexist} or to those of Theorems \ref{thm:usiequiv} and \ref{thm:usiseequiv} is a challenging open problem of paramount importance. Another important open problem is the discovery of information compression maps under which certain subsets of equilibrium payoff profiles are attained when strategies based on the resulting compressed information are used. The results of this paper have been derived for finite-horizon finite games. The extension of the results to infinite-horizon games and to games with continuous action and state spaces are other interesting technical problems}.

\subsection*{Author Contributions}
This work is a collaborative intellectual effort of the three authors, with Dengwang Tang being the leader. Due to the interconnected nature of the results, it is impossible to separate the contributions of each author.

\subsection*{Funding}
This work is supported by National Science Foundation (NSF) Grant No. ECCS 1750041, ECCS 2025732, ECCS 2038416, ECCS 1608361, CCF 2008130, CMMI 2240981, Army Research Office (ARO) Award No. W911NF-17-1-0232, and Michigan Institute for Data Science (MIDAS) Sponsorship Funds by General Dynamics.

\subsection*{Data Availability}
Not applicable since all results in this paper are theoretical.

\section*{Declarations}

\subsection*{Conflict of Interest}
The authors have no competing interests to declare that are relevant to the content of this article.

\subsection*{Ethical Approval}
Not applicable since no experiments are involved in this work.

% ============================================================================================================================

\appendix
\begin{appendices}

\section{Information State of Single-Agent Control Problems}\label{app:aux:infostate}
In this section we consider single-agent Markov Decision Processes (MDPs) and develop auxiliary results. This section is a recap of \cite{mahajan2016decentralized} with more detailed results and proofs. %The notation used in this section is independent from the rest of the paper.

Let $X_t$ be a controlled Markov Chain controlled by action $U_t$ with initial distribution $\nu_1\in\Delta(\mathcal{X}_1)$ and transition kernels $P=(P_t)_{t\in\mathcal{T}}, P_t\colon\mathcal{X}_t\times \mathcal{U}_t\mapsto \Delta(\mathcal{X}_{t+1})$. Let $r=(r_t)_{t\in\mathcal{T}}, r_t\colon\mathcal{X}_t\times \mathcal{U}_t\mapsto \mathbb{R}$ be a collection of instantaneous reward functions. An MDP is denoted by a tuple $(\nu_1, P, r)$. 

For a Markov strategy $g=(g_t)_{t\in\mathcal{T}}, g_t\colon\mathcal{X}_t\mapsto \Delta(\mathcal{U}_t)$, we use $\Pr^{g, \nu_1, P}$ and $\E^{g, \nu_1, P}$ to denote the probabilities of events and expectations of random variables under the distribution specified by controlled Markov Chain $(\nu_1, P)$ and strategy $g$. When $(\nu_1, P)$ is fixed and clear from the context, we use $\Pr^{g}$ and $\E^{g}$ respectively.

Define the total expected reward in the MDP $(\nu_1, P, r)$ under strategy $g$ by
\begin{align}
	J(g; \nu_1, P, r) := \E^{g, \nu_1, P}\left[\sum_{t=1}^T r_t(X_t, U_t) \right].\label{eq:MDP:J}
\end{align}

Define the value function and state-action quality function by
\begin{align}
	V_\tau(x_\tau; P, r) &:= \max_{g_{\tau:T}}\E^{g_{\tau:T}, P}\left[\sum_{t=\tau}^T r_t(X_t, U_t)|x_\tau\right],\qquad\forall \tau\in [T+1],\\
	\Qfunc_\tau(x_\tau, u_\tau; P, r) &:= r_\tau(x_\tau, u_\tau) + \sum_{\tilde{x}_{\tau+1}} V_{\tau+1}(\tilde{x}_{\tau+1}) P_\tau(\tilde{x}_{\tau+1}|x_\tau, u_\tau),\qquad\forall \tau\in [T].
\end{align}
Note that $V_{T+1}(\cdot;P,r)\equiv 0$.

\begin{defn}\label{def:app:infostate}
	\citep{mahajan2016decentralized} Let $\Comp_t=\Psi_t(X_t)$ for some function $\Psi_t$. Then, $\Comp_t$ is called an information state for $(P, r)$ if there exist functions $P_t^{\Comp}\colon\Compset_t\times \mathcal{U}_t\mapsto \Delta(\Compset_{t+1}), r_t^{\Comp}\colon\Compset_t\times \mathcal{U}_t\mapsto \mathbb{R}$ such that 
	\begin{enumerate}[(1)]
		\item $P_t(\comp_{t+1}|x_t, u_t) = P_t^{\Comp}(\comp_{t+1}|\Psi_t(x_t), u_t)$; and
		\item $r_t(x_t, u_t) = r_t^{\Comp}(\Psi_t(x_t), u_t)$.
	\end{enumerate}
\end{defn}

If $\Comp_t$ is an information state, then $\Comp_t$ is also a controlled Markov Chain with initial distribution $\nu_1^{\Comp}\in\Delta(\Comp_1)$ and transition kernel $P^{\Comp}=(P_t^{\Comp})_{t\in\mathcal{T}}$, where 
\begin{align*}
	\nu_1^{\Comp}(\comp_1) = \sum_{x_1} \bm{1}_{\{\comp_1 = \Psi_1(x_1) \} } \nu_1(x_1).
\end{align*}

The tuple $(\nu_1^{\Comp}, P^{\Comp}, r^{\Comp})$ defines a new MDP. For a \emph{$\Comp$-based strategy} $\rho=(\rho_t)_{t\in\mathcal{T}}, \rho_t\colon\Compset_t\mapsto \Delta(\mathcal{U}_t)$, the $J, V, \Qfunc$ functions can be defined as above for the new MDP.

We state the following standard result (see, for example, Section 2 of \cite{subramanian2020approximate}).

\begin{lemma}\label{lemma:twoMDPequiv}
	 Let $\Comp_t=\Psi_t(X_t)$ be an information state for $(P, r)$. Then
	\begin{enumerate}[(1)]
		\item $V_t(x_t; P, r) = V_t(\Psi_t(x_t); P^{\Comp}, r^{\Comp})$ for all $x_t$;
		\item $\Qfunc_t(x_t, u_t; P, r) = \Qfunc_t(\Psi_t(x_t), u_t; P^{\Comp}, r^{\Comp})$ for all $x_t, u_t$.
	\end{enumerate}
\end{lemma}

\begin{defn}\label{def:app:association}
	Let $g$ be a Markov strategy, an $K$-based strategy $\rho$ is said to be associated with $g$ if
	\begin{align}\label{rhotheinfostatebasedstrategy}
		\rho_t(\comp_t) = \E^{g, \nu_1, P}[g_t(X_t)|\comp_t],
	\end{align}
	whenever $\Pr^{g, \nu_1, P}(\comp_t) > 0$.
\end{defn}

The following lemma will be used in the proofs in Appendix \ref{app:proofsmain}.

\begin{lemma}[Policy Equivalence Lemma]\label{lem:policyeval}
	Let $(\nu_1, P, r)$ be an MDP. Let $\Comp_t$ be an information state for $(P, r)$. Let an $\Comp$-based strategy $\rho$ be associated with a Markov strategy $g$, then
 	\begin{enumerate}[(1)]
 		\item $\Pr^{g, \nu_1, P}(\comp_t)=\Pr^{\rho, \nu_1, P}(\comp_t)$ for all $\comp_t\in\Compset_t$ and $t\in\mathcal{T}$;
 		\item $J(g; \nu_1, P, r) = J(\rho; \nu_1, P, r)$.
 	\end{enumerate}
\end{lemma}

\begin{proof}
	In this proof all probabilities and expectations are assumed to be defined with $(\nu_1, P)$.
	Given a Markov strategy $g$, let $\rho$ be an information state-based strategy that satisfies \eqref{rhotheinfostatebasedstrategy}. 
	
	First, we have
	\begin{align}
		\Pr^g(u_t|\comp_t) = \E^g[g_t(u_t|X_t)|\comp_t] = \rho_t(u_t|\comp_t) \label{eq:policyeval:ugivens},
	\end{align}
	for all $\comp_t$ such that $\Pr^g(\comp_t) > 0$.
	
	\begin{enumerate}[(1)]
		\item Proof by induction:
		
		\textbf{Induction Base:} We have $\Pr^g(\comp_1) = \Pr^\rho(\comp_1)$ since the distribution of $\Comp_1=\Psi_1(X_1)$ is strategy-independent.
		
		\textbf{Induction Step:} Suppose that 
		\begin{align}
			\Pr^g(\comp_t) = \Pr^\rho(\comp_t),\label{eq:policyeval:indhyp}
		\end{align}
		for all $\comp_t\in\Compset_t$. We prove the result for time $t+1$. Combining \eqref{eq:policyeval:ugivens} and \eqref{eq:policyeval:indhyp}, and incorporating the information state transition kernel $P_t^K$ defined in Definition \ref{def:app:infostate}, we have
		\begin{align}
			\Pr^g(\comp_{t+1}) &= \sum_{\tilde{\comp}_t, \tilde{u}_t} \Pr^g(\comp_{t+1}|\tilde{\comp}_t, \tilde{u}_t)\Pr^g(\tilde{u}_t|\tilde{\comp}_t)\Pr^g(\tilde{\comp}_t) \\
			&=\sum_{\tilde{\comp}_t, \tilde{u}_t} P_t^{\Comp}(\comp_{t+1}|\tilde{\comp}_t, \tilde{u}_t)\rho_t(u_t|\tilde{\comp}_t)\Pr^\rho(\tilde{\comp}_t)\\
			&=\Pr^\rho(\comp_{t+1}).
		\end{align}

	    Therefore we have established the induction step.
		\item Using \eqref{eq:policyeval:ugivens}\eqref{eq:policyeval:indhyp} along with the result of part (1), we obtain
		\begin{align}
			\E^g[r_{t}(X_t, U_t)] &= \E^g[r_{t}^{\Comp}(\Comp_t, U_t)] \\
			&=\sum_{\tilde{\comp}_t, \tilde{u}_t} r_t^{\Comp}(\tilde{\comp}_t, \tilde{u}_t) \Pr^g(\tilde{u}_t|\tilde{\comp}_t) \Pr^g(\tilde{\comp}_t) \\
			&=\sum_{\tilde{\comp}_t, \tilde{u}_t} r_t^{\Comp}(\tilde{\comp}_t, \tilde{u}_t) \rho_t(\tilde{u}_t|\tilde{\comp}_t) \Pr^\rho(\tilde{\comp}_t)\\
			&=\E^\rho[r_{t}(X_t, U_t)],
		\end{align}
		for each $t\in\mathcal{T}$. The result then follows from linearity of expectation.
	\end{enumerate}
 This concludes the proof.
\end{proof}

\section{Alternative Characterizations of Sequential Equilibria}\label{app:SE}
This section deals with the game model introduced in Section \ref{sec:suffinfo:gamemodel}. We provide three alternative definitions of sequential equilibria that are equivalent to the original one given by \cite{kreps1982sequential}. These definitions help simplify some of the proofs in Appendix \ref{app:proofsmain}.

We would like to note that several alternative definitions of sequential equilibria are also given in \cite{kreps1982sequential,halpern2009nonstandard}. The definition of {weak perfect equilibrium} in Proposition 6 of \cite{kreps1982sequential} is close to our definitions in spirit in terms of using sequences of payoff functions instead of beliefs as a vehicle to define
sequential rationality.

Notice that fixing the behavioral strategies $g^{-i}$ of players other than player $i$, player $i$'s best response problem (at every information set) can be considered as a Markov Decision Process with state $H_t^i$ and action $U_t^i$, where the transition kernels and instantaneous reward functions depend on $g^{-i}$. Inspired by this observation, we introduce an alternative definition of sequential equilibrium for our model, where we form \emph{conjectures} of transition kernels and reward functions instead of forming beliefs on nodes. This allows us for a more compact representation of the appraisals and beliefs of players. We will later show that this alternative definition is equivalent to the classical definition of sequential equilibrium in \cite{kreps1982sequential}.

For player $i\in\mathcal{I}$, let $P^i=(P_t^i)_{t\in\mathcal{T}\backslash\{T\}}, P_t^i\colon\mathcal{H}_t^i\times \mathcal{U}_t^i\mapsto \Delta(\mathcal{Z}_t^i)$ and $r^i=(r_t^i)_{t\in\mathcal{T}}, r_t^i\colon\mathcal{H}_t^i \times \mathcal{U}_t^i \mapsto [-1, 1]$ be collections of functions that represent conjectures of transition kernels and instantaneous reward functions. For a behavioral strategy profile $g^i$, define the reward-to-go function $J_t^i$ recursively through
\begin{subequations}\label{defofJ}
	\begin{align}
		&\quad~J_{T}^i(g_{T}^i; h_T^i, P^i, r^i) := \sum_{\tilde{u}_T^i} r_T^i(h_T^i, \tilde{u}_T^i) g_T^i(\tilde{u}_T^i|h_T^i);\label{defofJ:VT}\\
		&\quad~J_t^i(g_{t:T}^i; h_t^i, P^i, r^i)\label{defofJ:Vt} \\
		:= &\sum_{\tilde{u}_t^i} \left[ r_t^i(h_t^i, \tilde{u}_t^i) + \sum_{\tilde{z}_t^i} J_{t+1}^i(g_{t+1:T}^i; (h_t^i, \tilde{z}_t^i), P^i, r^i) P_t^i(\tilde{z}_t^i|h_t^i, \tilde{u}_t^i) \right] g_t^i(\tilde{u}_t^i|h_t^i).
	\end{align}
\end{subequations}
\noeqref{defofJ:VT,defofJ:Vt}

\begin{defn}[``Model-based'' Sequential Equilibrium]\label{def:PrSE}
	Let $g=(g^i)_{i\in\mathcal{I}}$ be a behavioral strategy profile. Let $(P, r)=(P^i, r^i)_{i\in\mathcal{I}}$ be a conjectured profile. Then, $g$ is said to be sequentially rational under $(P, r)$ if for each $i\in\mathcal{I}, t\in\mathcal{T}$ and each $h_t^i\in\mathcal{H}_t^i$, 
	\begin{equation}
		J_t^i(g_{t:T}^i; h_t^i, P^i, r^i)\geq J_t^i(\tilde{g}_{t:T}^i; h_t^i, P^i, r^i),
	\end{equation}
	for all behavioral strategies $\tilde{g}_{t:T}^i$. Conjectured profile $(P, r)$ is said to be fully consistent with $g$ if there exist a sequence of behavioral strategy and conjecture profiles  $(g^{(n)}, P^{(n)}, r^{(n)})_{n=1}^\infty$ such that
	\begin{enumerate}[(1)]
		\item $g^{(n)}$ is fully mixed, i.e. every action is chosen with positive probability at every information set.
		\item For each $i\in\mathcal{I}$, $(P^{(n), i}, r^{(n), i})$ is consistent with $g^{(n), -i}$, i.e. for each $i\in\mathcal{I}, t\in\mathcal{T}, h_t^i\in \mathcal{H}_t^i, u_t^i\in \mathcal{U}_t^i$,
		\begin{align}
			P_t^{(n), i}(z_t^i|h_t^i, u_t^i)&=\Pr^{g^{(n), -i}}(z_t^i|h_t^i, u_t^i),\\
			r_t^{(n), i}(h_t^i, u_t^i) &= \E^{g^{(n), -i}}[R_t^i|h_t^i, u_t^i].
		\end{align}
		\item $(g^{(n)}, P^{(n)}, r^{(n)})\rightarrow (g, P, r)$ as $n\rightarrow\infty$.
	\end{enumerate}
	
	A triple $(g, P, r)$ is said to be a \emph{``model-based'' sequential equilibrium}\footnote{Here we borrow the terms ``model-based'' (resp. ``model-free'') from the reinforcement learning literature: ``Model-based'' means that an algorithm constructs the underlying model $(P, r)$, while ``model-free'' usually means that the algorithm directly constructs state-action value functions $\Qfunc$.} if $g$ is sequentially rational under $(P, r)$ and $(P, r)$ is fully consistent with $g$.
\end{defn}

One can also form conjectures directly on the optimal reward-to-go given a state-action pair $(h_t^i, u_t^i)$. 
\begin{defn}[``Model-free'' Sequential Equilibrium, Definition \ref{def:KSEinmaintext} revisited]\label{def:KSE}
	Let $g=(g^i)_{i\in\mathcal{I}}$ be a behavioral strategy profile. Let $\Qfunc=(\Qfunc_t^i)_{i\in\mathcal{I},t\in\mathcal{T}}$ be a collection of functions where $\Qfunc_t^i\colon\mathcal{H}_t^i \times \mathcal{U}_t^i \mapsto [-T, T]$. The strategy profile $g$ is said to be sequentially rational under $\Qfunc$ if for each $i\in\mathcal{I}, t\in\mathcal{T}$ and each $h_t^i\in\mathcal{H}_t^i$, 
	\begin{equation}
		\mathrm{supp}(g_t^i(h_t^i))\subseteq \underset{u_t^i}{\arg\max}~ \Qfunc_t^i(h_t^i, u_t^i).
	\end{equation}
	
	The collection of functions $\Qfunc$ is said to be fully consistent with $g$ if there exist a sequence of behavioral strategy and conjectured profiles  $(g^{(n)}, \Qfunc^{(n)})_{n=1}^\infty$ such that
	\begin{enumerate}[(1)]
		\item $g^{(n)}$ is fully mixed, i.e. every action is chosen with positive probability at every information set.
		\item $\Qfunc^{(n)}$ is consistent with $g^{(n)}$, i.e.,
		\begin{align}
			\Qfunc_\tau^{(n), i}(h_\tau^i, u_\tau^i)&=\E^{g^{(n)}}\left[\sum_{t=\tau}^T R_t^i\Big|h_\tau^i, u_\tau^i\right],
		\end{align}
		for each $i\in\mathcal{I}, \tau\in\mathcal{T}, h_\tau^i\in \mathcal{H}_\tau^i, u_\tau^i\in \mathcal{U}_\tau^i$.
		\item $(g^{(n)}, \Qfunc^{(n)})\rightarrow (g, \Qfunc)$ as $n\rightarrow\infty$.
	\end{enumerate}
	A tuple $(g, \Qfunc)$ is said to be a ``model-free'' sequential equilibrium if $g$ is sequentially rational under $\Qfunc$ and $\Qfunc$ is fully consistent with $g$.
\end{defn}	

A slightly different definition is also equivalent:
\begin{defn}[``Model-free'' Sequential Equilibrium, Version 2]\label{def:KSE2}
A tuple $(g, \Qfunc)$ is said to be a ``model-free'' sequential equilibrium (version 2) if it satisfies Definition \ref{def:KSE} with condition (2) for full consistency replaced by the following condition:
\begin{itemize}
	\item[(2')] For each $i$, $\Qfunc^{(n), i}$ is consistent with $g^{(n), -i}$, i.e.
	\begin{align*}
		\Qfunc_\tau^{(n), i}(h_\tau^i, u_\tau^i)&=\E^{g^{(n), -i}}[R_\tau^i|h_\tau^i, u_\tau^i] + \underset{\tilde{g}_{\tau+1:T}^i}{\max}~\E^{\tilde{g}_{\tau+1:T}^i, g^{(n), -i}}\left[\sum_{t=\tau+1}^T R_t^i\Big|h_\tau^i, u_\tau^i\right],
	\end{align*}
	for each $\tau\in\mathcal{T}, h_\tau^i\in \mathcal{H}_\tau^i, u_\tau^i\in \mathcal{U}_\tau^i$.
\end{itemize}
\end{defn}

{To introduce the last definition of SE, which corresponds to the original definition proposed in \cite{kreps1982sequential}, we first describe the game in Section \ref{sec:suffinfo:gamemodel} as an extensive-form game tree as follows:} 
To convert the game from a simultaneous move game to a sequential game, we set $\mathcal{I}=\{1, 2, \cdots, I\}$, where the index indicates the order of movement. For convenience, for $i\in\mathcal{I}$, we use the superscript $<i$ (resp. $>i$) to represent the set of players $\{1, \cdots, i-1\}$ (resp. $\{i+1, \cdots, I\}$) that moves before (resp. after) player $i$ in any given round. At time $t=0$, nature takes action $w_0=(x_1, h_1)$ and the game enters $t=1$. For each time $t\in\mathcal{T}$, player $1$ takes action $u_t^1$ first, then followed by player $2$ taking action $u_t^2$, and so on, while nature takes action $w_t$ after player $I$ takes action $u_t^I$. 
{In this extensive form game, there are three types of nodes: (1) a node where some player $i\in\mathcal{I}$ takes action (at some time $t\in\mathcal{T}$), (2) a node where nature takes action (at some time $t\in  \{0\}\cup \mathcal{T}$), and (3) a terminal node, where the game has terminated. We denote the set of the first type of nodes corresponding to player $i$ and time $t$ as $\mathcal{O}_t^i$. A node $o_t^i\in\mathcal{O}_t^i$ can also be represented as a vector $o_t^i = (x_{1}, h_1, w_{1:t-1}, u_{1:t-1}, u_{t}^{<i})$ which contains all the moves (by all players and nature) before it. As a result, $o_t^i$ also uniquely determines the states $x_{1:t}$ and information increment vectors $z_{1:t-1}$. We denote the set of the terminal nodes as $\mathcal{O}_{T+1}$. A terminal node $o_{T+1}\in \mathcal{O}_{T+1}$ also has a vector representation $o_{T+1} = (x_{1}, h_1, w_{1:T}, u_{1:T})$.}

Given a terminal node $o_{T+1}$, all the actions of players and nature throughout the game are uniquely determined, hence the realizations of $(R_t)_{t\in\mathcal{T}}$ defined in Section \ref{sec:suffinfo:gamemodel} are also uniquely determined. Let $\Lambda=(\Lambda^i)_{i\in\mathcal{I}}, \Lambda^i\colon\mathcal{O}_{T+1} \mapsto \mathbb{R}$ be the mappings from terminal nodes to total payoffs, i.e. $\Lambda^i(o_{T+1}) = \sum_{t=1}^T r_t^i$, where $r_t^i$ is the realization of $R_t^i$ corresponding to $o_{T+1}$. Also define $\Lambda_{\tau}^i(o_{T+1}) = \sum_{t=\tau}^T r_t^i$ for each $\tau\in\mathcal{T}$.

{Now, as we have constructed the extensive-form game, it is helpful to view the nodes in the game tree as a stochastic process.}
Define $O_t^i$ to be a random variable with support on $\mathcal{O}_t^i$ that represents the node player $i$ is at before taking action at time $t$. Let $O_{T+1}$ be a random variable with support on $\mathcal{O}_{T+1}$ that represents the terminal node the game ends at. 
If we view $(\mathcal{T}\times \mathcal{I}) \cup \{T+1 \}$ as a set of time indices with lexicographic ordering, the random process $(O_t^i)_{(t, i)\in\mathcal{T}\times \mathcal{I}}\cup(O_{T+1})$ is a controlled Markov Chain controlled by action $U_t^i$ at time $(t, i)$. 

\begin{defn}[Classical Sequential Equilibrium \citep{kreps1982sequential}]\label{def:classicalSE}
	An assessment is a pair $(g, \mu)$, where $g$ is a behavioral strategy profile of players (excluding nature) as described in Section \ref{sec:suffinfo:gamemodel}, and $\mu=(\mu_t^i)_{t\in\mathcal{I}, i\in\mathcal{I}}, \mu_t^i\colon\mathcal{H}_t^i \mapsto \Delta(\mathcal{O}_t^i)$ is a belief system. Then, $g$ is said to be sequentially rational given $\mu$ if
	\begin{equation}\label{eq:seqrat}
		\sum_{o_t^i}\E^{g_{t:T}^i, g_t^{>i}, g_{t:T}^{-i}} [\Lambda^i(O_{T+1})|o_t^i]\mu_t^i(o_t^i|h_t^i) \geq \sum_{o_t^i}\E^{\tilde{g}_{t:T}^i, g_t^{>i}, g_{t:T}^{-i}} [\Lambda^i(O_{T+1})|o_t^i]\mu_t^i(o_t^i|h_t^i),
	\end{equation}
	for all $i\in\mathcal{I}, t\in\mathcal{T}, h_t^i\in\mathcal{H}_t^i$, and all behavioral strategies $\tilde{g}_{t:T}^i$. 
   %\vjmargincomment{Since you have defined the notation $\Lambda_\tau^i(o_{T+1})$, perhaps you want to use $\Lambda_{t}^i(O_{T+1}))$ inside the expectations in \eqref{eq:seqrat}. This is even though by the conditioning you get the same thing since terms cancel on both sides.} 
   The belief system $\mu$ is said to be fully consistent with $g$ if there exist a sequence of assessments $(g^{(n)}, \mu^{(n)})_{n=1}^\infty \rightarrow (g, \mu)$ such that $g^{(n)}$ is a fully mixed strategy profile and
	\begin{enumerate}[(1)]
		\item $g^{(n)}$ is fully mixed.
		\item $\mu^{(n)}$ is consistent with $g^{(n)}$, i.e. $\mu_t^{(n), i} (o_t^i|h_t^i) = \Pr^{g^{(n)}}(o_t^i|h_t^i)$ for all $t\in\mathcal{T}, i\in\mathcal{I} ,h_t^i\in\mathcal{H}_t^i$, and $o_t^i\in\mathcal{O}_t^i$.
		
		\item $(g^{(n)}, \mu^{(n)}) \rightarrow (g, \mu)$ as $n\rightarrow\infty$.
	\end{enumerate}
	An assessment $(g, \mu)$ is said to be a (classical) sequential equilibrium if $g$ is sequentially rational given $\mu$ and $\mu$ is fully consistent with $g$.
\end{defn}

\begin{remark}
	Since the instantaneous rewards $R_{1:t-1}^i$ {have already been} realized at time $t$, replacing the total reward $\Lambda$ with reward-to-go $\Lambda_t$ in \eqref{eq:seqrat} would result in an equivalent definition. 
\end{remark}

\begin{thm}\label{thm:SEdefequiv}
	Definitions ~\ref{def:PrSE}, \ref{def:KSE}, \ref{def:KSE2}, and \ref{def:classicalSE} are equivalent for strategy profiles.
\end{thm}

\begin{proof}
    {We complete the proof via four steps: In each step, we show that if $g$ is a strategy profile satisfying one definition of SE, then it satisfy one of the other definitions of SE as well. We follow the following diagram: Definition \ref{def:classicalSE} $\Rightarrow$ Definition \ref{def:PrSE} $\Rightarrow$ Definition \ref{def:KSE} $\Rightarrow$ Definition \ref{def:KSE2} $\Rightarrow$ Definition \ref{def:classicalSE}.}\\~

    \textbf{Step 1:} Classical SE (Definition \ref{def:classicalSE}) $\Rightarrow$ ``Model-based'' SE (Definition \ref{def:PrSE})\\~
		
		Let $(g, \mu)$ satisfy Definition \ref{def:classicalSE}. Let $(g^{(n)}, \mu^{(n)})$ be a sequence of assessments that satisfies conditions (1)-(3) of fully consistency in Definition \ref{def:classicalSE}. 
		
		Set $P_t^{(n), i}(z_t^i|h_t^i, u_t^i) = \Pr^{g^{(n)}}(z_t^i|h_t^i, u_t^i)$ and $r_t^{(n), i}(h_t^i, u_t^i) = \E^{g^{(n)}}[R_t^i|h_t^i, u_t^i]$ for all $h_t^i\in\mathcal{H}_t^i, u_t^i\in\mathcal{U}_t^i$.
		
        {Recall that we can write $O_t^i = (X_1, H_1, W_{1:t-1}, U_{1:t-1}, U_t^{<i})$, and $(X_{1:t}, Z_{1:t-1})$ can be expressed as a function of $O_t^i$.}
        Therefore there exist fixed functions $f_t^{i, Z}, f_t^{i, R}$ such that 
			$Z_t^i = f_t^{i, Z}(O_t^i, U_t^i, U_t^{>i}, W_t)$,
			$R_t^i = f_t^{i, R}(O_t^i, U_t^i, U_t^{>i}, W_t)$. {Furthermore, for all $j>i$, there also exists functions $f_t^{j, i, H}$ such that
			$H_t^j = f_t^{j, i, H}(O_t^i)$ (since $H_t^j=(H_1^i,Z_{1:t-1}^i)$)}.
		Since $\mu_t^{(n), i}(o_t^i|h_t^i) = \Pr^{g^{(n)}}(o_t^i| h_t^i)$ we have 
		\begin{align}
  \begin{split}
			&\es P_t^{(n), i}(z_t^i|h_t^i, u_t^i) \\
			&= \sum_{o_t^i, \tilde{u}_t^{>i}, \tilde{w}_t} \bm{1}_{\{z_t^i = f_t^{i, Z}(o_t^i, u_t^i, \tilde{u}_t^{>i}, \tilde{w}_t) \} } \Pr(\tilde{w}_t) \left(\prod_{j=i+1}^I g_t^{(n), j}(\tilde{u}_t^j|f_t^{j, i, H}(o_t^i))\right) \mu_t^{(n)}(o_t^i|h_t^i),
   \end{split}\\
   \begin{split}
			&\es r_t^{(n), i}(h_t^i, u_t^i) \\
			&= \sum_{o_t^i, \tilde{u}_t^{>i}, \tilde{w}_t} f_t^{i, R}(o_t^i, u_t^i,\tilde{u}_t^{>i}, \tilde{w}_t) \Pr(\tilde{w}_t) \left(\prod_{j=i+1}^I g_t^{(n), j}(\tilde{u}_t^j|f_t^{j, i, H}(o_t^i))\right) \mu_t^{(n)}(o_t^i|h_t^i).
   \end{split}
		\end{align}
	Therefore, as $\mu^{(n)} \rightarrow \mu$, $g^{(n)}\rightarrow g$, we have $(P^{(n)}, r^{(n)} )\rightarrow (P, r)$ for some $(P, r)$.  
		
		Let $\tau\in\mathcal{T}$ and $\tilde{g}_{\tau:T}^i$ be an arbitrary strategy. First, observe that one can represent the conditional reward-to-go $\E^{g^{(n)}}[\sum_{t=\tau}^T R_t^i|h_\tau^i]$ using $\mu^{(n)}$ or $(P^{(n)}, r^{(n)})$. Hence we have
		\begin{equation}\label{eq:classicalSErtg}
		\begin{split}
			\sum_{o_\tau^i}\E^{\tilde{g}_{\tau:T}^i, g_{\tau}^{(n), >i}, g_{\tau+1:T}^{(n), -i}} [\Lambda_\tau^i(O_{T+1})|o_\tau^i]\mu_\tau^{(n), i}(o_\tau^i|\tau_t^i) 
			&= J_t^i(\tilde{g}_{\tau:T}^i; h_\tau^i, P^{(n), i}, r^{(n), i}),
		\end{split}
		\end{equation}
		where $J_t^i$ is as defined in \eqref{defofJ}.
		
		Observe that the left-hand side of \eqref{eq:classicalSErtg} is continuous in $(g_{\tau}^{(n), >i}, g_{\tau+1:T}^{(n), -i}, \mu_\tau^{(n), i})$ since it is a sum of products of components of $(g_{\tau}^{(n), >i}, g_{\tau+1:T}^{(n), -i}, \mu_\tau^{(n), i})$. Also observe that the right-hand side of \eqref{eq:classicalSErtg} is continuous in $(P^{(n), i}, r^{(n), i})$ since it is a sum of products of components of $(P^{(n), i}, r^{(n), i})$ by the definition in \eqref{defofJ}.
		Therefore by taking limit as $n\rightarrow\infty$, we conclude that
		\begin{equation}
			\sum_{o_\tau^i}\E^{\tilde{g}_{\tau:T}^i, g^{-i}} [\Lambda_\tau^i(O_{T+1})|o_\tau^i]\mu_\tau^{i}(o_\tau^i|h_\tau^i) = J_\tau^i(\tilde{g}_{\tau:T}^i; h_\tau^i, P^{i}, r^{i}),\label{eq:rewardequivmuPr}
		\end{equation} 
		for all strategies $\tilde{g}_{\tau:T}^i$. Using sequential rationality of $g$ with respect to $\mu$ and \eqref{eq:rewardequivmuPr} we conclude that
		\begin{align}
			J_t^i(g_{\tau:T}^i; h_\tau^i, P^{i}, r^{i}) \geq  J_t^i(\tilde{g}_{\tau:T}^i; h_\tau^i, P^{i}, r^{i}),
		\end{align}
		for all $\tau\in \mathcal{T}, i\in\mathcal{I}, h_\tau^i \in\mathcal{H}_\tau^i$, i.e. $g$ is also sequentially rational given $(P, r)$.\\~
		
		\textbf{Step 2}: ``Model-based'' SE (Definition \ref{def:PrSE}) $\Rightarrow$ ``Model-free'' SE version 1 (Definition \ref{def:KSE})\\~
		
		Let $(g, P, r)$ be a sequential equilibrium under Definition \ref{def:PrSE}, and let $(g^{(n)}, P^{(n)}, r^{(n)})$ satisfy conditions (1)-(3) of full consistency in Definition \ref{def:PrSE}. Set
		\begin{align}
			\Qfunc_\tau^{(n), i}(h_\tau^i, u_{\tau}^i) = \E^{g^{(n)}}\left[\sum_{t=\tau}^T R_t^i\Big|h_\tau^i, u_{\tau}^i \right], 
		\end{align} 
		for all $\tau\in \mathcal{T}, i\in\mathcal{I}, h_\tau^i \in\mathcal{H}_\tau^i, u_\tau^i\in\mathcal{U}_\tau^i$. Then $\Qfunc^{(n), i}$ satisfies the recurrence relation 
        \begin{subequations}
		\begin{align}
			\Qfunc_T^{(n), i}(h_T^i, u_T^i) &= r_T^{(n), i}(h_T^i, u_T^i),\\
			V_t^{(n), i}(h_t^i) &:= \sum_{\tilde{u}_t^i} \Qfunc_t^{(n), i}(h_t^i, \tilde{u}_t^i)g_t^{(n), i}(\tilde{u}_t^i|h_t^i), \quad \forall t\in\mathcal{T},\\
			\Qfunc_{t}^{(n), i}(h_t^i, u_t^i) &= r_t^{(n), i}(h_t^i, u_t^i) \\
			&+ \sum_{\tilde{z}_{t}^i} V_{t+1}^{(n), i}((h_t^i, \tilde{z}_{t}^i)) P_t^{(n), i}(\tilde{z}_{t}^i|h_t^i, u_t^i),\quad \forall t\in\mathcal{T}\backslash\{T\}.
		\end{align}
        \end{subequations}
	
		Since $(g^{(n)}, P^{(n)}, r^{(n)}) \rightarrow (g, P, r)$ as $n\rightarrow\infty$, we have $\Qfunc^{(n)} \rightarrow \Qfunc$ where $\Qfunc=(\Qfunc_t^i)_{t\in\mathcal{T}, i\in\mathcal{I}}$ satisfies
		\begin{subequations}\label{eq:PrlimitK}
		\begin{align}
			\Qfunc_T^{i}(h_T^i, u_T^i) &= r_T^{i}(h_T^i, u_T^i),\label{eqPrlimitK:K}\\
			V_t^{i}(h_t^i) &:= \sum_{\tilde{u}_t^i} \Qfunc_t^{i}(h_t^i, \tilde{u}_t^i)g_t^{i}(\tilde{u}_t^i|h_t^i),\quad \forall t\in\mathcal{T},\label{eqPrlimitK:V}\\
			\Qfunc_{t}^{i}(h_t^i, u_t^i) &= r_t^{i}(h_t^i, u_t^i) \notag \\
			&+ \sum_{\tilde{z}_{t}^i} V_{t+1}^{i}((h_t^i, \tilde{z}_{t}^i)) P_t^{i}(\tilde{z}_{t}^i|h_t^i, u_t^i),\quad \forall t\in\mathcal{T}\backslash\{T\}. \label{eqPrlimitK:K1}
		\end{align}
		\end{subequations}
% \noeqref{eqPrlimitK:K,eqPrlimitK:V,eqPrlimitK:K1}	
		 Comparing \eqref{eq:PrlimitK} with {the reward-to-go function $J_t^i$ defined in \eqref{defofJ}, we observe that}
        \begin{align}
			V_t^{i}(h_t^i) = J_t^i(g_{t:T}^i; h_t^i, P^i, r^i)\label{eq:appSE:VisJ},
		\end{align}
		for all $t\in \mathcal{T}, i\in\mathcal{I}, h_\tau^i \in\mathcal{H}_\tau^i$.
	
		Let $\tilde{g}_t^i$ be a strategy such that $\hat{g}_t^i(h_t^i) = \eta \in\Delta(\mathcal{U}_t^i)$, then 
		\begin{align}
			&\es J_t^i((\tilde{g}_{t}^i, g_{t+1:T}^i); h_t^i, P^i, r^i)\\
			&= \sum_{\tilde{u}_t}\left( r_t^{i}(h_t^i, \tilde{u}_t^i) + \sum_{\tilde{z}_{t}^i} J_{t+1}^i(g_{t+1:T}^i; (h_t^i, \tilde{z}_t^i), P^i, r^i) P_t^i(\tilde{z}_t^i|h_t^i, \tilde{u}_t^i)\right)\eta(\tilde{u}_t^i)\label{eq:appSE:Jexpansion} \\
			&= \sum_{\tilde{u}_t}\left(r_t^{i}(h_t^i, \tilde{u}_t^i) + \sum_{\tilde{z}_{t}^i} V_{t+1}^i((h_t^i, \tilde{z}_t^i)) P_t^i(\tilde{z}_t^i|h_t^i, \tilde{u}_t^i)\right)\eta(\tilde{u}_t^i)\label{eq:appSE:Vreplace} \\
			&= \sum_{\tilde{u}_t}\Qfunc_{t}^{i}(h_t^i, \hat{u}_t^i) \eta(\tilde{u}_t^i),\label{eq:appSE:Qeta}
		\end{align}
        {where we substitute \eqref{defofJ} in \eqref{eq:appSE:Jexpansion}, \eqref{eq:appSE:VisJ} in \eqref{eq:appSE:Vreplace}, and \eqref{eqPrlimitK:K1} in \eqref{eq:appSE:Qeta}.}
	
		By sequential rationality of $g$ with respect to $(P, r)$, we have $$J_t^i(g_{t:T}^i; h_t^i, P^i, r^i) \geq J_t^i((\tilde{g}_{t}^i, g_{t+1:T}^i); h_t^i, P^i, r^i),$$which means that
		\begin{align}
			\sum_{\tilde{u}_t}\Qfunc_{t}^{i}(h_t^i, \tilde{u}_t^i) g_t^i(\tilde{u}_t^i|h_t^i) \geq \sum_{\tilde{u}_t}\Qfunc_{t}^{i}(h_t^i, \tilde{u}_t^i) \eta(\tilde{u}_t^i),
		\end{align}
		for all $\eta\in\Delta(\mathcal{U}_t^i)$ for all $t\in \mathcal{T}, i\in\mathcal{I}, h_\tau^i \in\mathcal{H}_\tau^i$. Hence $g$ is sequentially rational given $\Qfunc$. Therefore $(g, \Qfunc)$ is a sequential equilibrium under Definition \ref{def:KSE}.\\~
		
		\textbf{Step 3}: ``Model-free'' SE version 1 (Definition \ref{def:KSE}) $\Rightarrow$ ``Model-free'' SE version 2 (Definition \ref{def:KSE2})\\~
		
		Let $(g, \Qfunc)$ be a sequential equilibrium under Definition \ref{def:KSE} and let $(g^{(n)}, \Qfunc^{(n)})$ satisfies conditions (1)-(3) of full consistency in Definition \ref{def:KSE}. Then $\Qfunc^{(n), i}$ satisfies 
        \begin{subequations}
		\begin{align}
			\Qfunc_T^{(n), i}(h_T^i, u_T^i) &= \E^{g^{(n), -i}}[R_T^i|h_T^i, u_T^i],\\
			V_t^{(n), i}(h_t^i) &:= \sum_{\tilde{u}_t^i} \Qfunc_t^{(n), i}(h_t^i, \tilde{u}_t^i)g_t^{(n), i}(\tilde{u}_t^i|h_t^i), \quad \forall t\in\mathcal{T},\\
			\Qfunc_{t}^{(n), i}(h_t^i, u_t^i) &= \E^{g^{(n), -i}}[R_t^i|h_t^i, u_t^i]\notag \\
			&+ \sum_{\tilde{z}_{t}^i} V_{t+1}^{(n), i}((h_t^i, \tilde{z}_{t}^i)) \Pr^{g^{(n), -i}}(\tilde{z}_{t}^i|h_t^i, u_t^i),\quad \forall t\in\mathcal{T}\backslash\{T\},\label{eq:appSE:2:Qnt}
		\end{align}
        \end{subequations}
		and $\Qfunc^{(n)} \rightarrow \Qfunc$ as $n\rightarrow\infty$.	
		Set
		\begin{align}
			\hat{\Qfunc}_\tau^{(n), i}(h_\tau^i, u_\tau^i) = \E^{g^{(n), -i}}[R_\tau^i|h_\tau^i, u_\tau^i] + \underset{\tilde{g}_{\tau+1:T}^i}{\max}~\E^{\tilde{g}_{\tau+1:T}^i, g^{(n), -i}}\left[\sum_{t=\tau+1}^T R_t^i\Big|h_\tau^i, u_\tau^i\right],
		\end{align}
		for each $\tau\in\mathcal{T}, h_\tau^i\in \mathcal{H}_\tau^i, u_\tau^i\in \mathcal{U}_\tau^i$. Then $\hat{\Qfunc}^{(n), i}$ satisfies the recurrence relation
        \begin{subequations}
		\begin{align}
			\hat{\Qfunc}_T^{(n), i}(h_T^i, u_T^i) &= \E^{g^{(n), -i}}[R_T^i|h_T^i, u_T^i],\\
			\hat{V}_t^{(n), i}(h_t^i) &:= \max_{\tilde{u}_t^i} \hat{\Qfunc}_t^{(n), i}(h_t^i, \tilde{u}_t^i), \quad \forall t\in\mathcal{T},\\
			\hat{\Qfunc}_{t}^{(n), i}(h_t^i, u_t^i) &= \E^{g^{(n), -i}}[R_t^i|h_t^i, u_t^i] \notag\\
			&+ \sum_{\tilde{z}_{t}^i} \hat{V}_{t+1}^{(n), i}((h_t^i, \tilde{z}_{t}^i)) \Pr^{g^{(n), -i}}(\tilde{z}_{t}^i|h_t^i, u_t^i),\quad \forall t\in\mathcal{T}\backslash\{T\}.\label{eq:appSE:2:hatQnt}
		\end{align}
        \end{subequations}
	
		\textbf{Claim:} $\hat{\Qfunc}_{t}^{(n)} \rightarrow \Qfunc_{t}^{i}$ as $n\rightarrow\infty$.
		
		Given the claim, we have $(g^{(n)}, \hat{\Qfunc}^{(n)})$ satisfying conditions (1)(2')(3) of full consistency in Definition \ref{def:KSE2}. Therefore $(g, \Qfunc)$ is also a sequential equilibrium under Definition \ref{def:KSE2}, and we complete this part of the proof.
		
		\textbf{Proof of Claim:} By induction on time $t\in\mathcal{T}$.
		
		\textbf{Induction Base:} Observe that $\hat{\Qfunc}_T^{(n)} = \Qfunc_T^{(n)}$ by construction. Since $\Qfunc_T^{(n)}\rightarrow \Qfunc_T$ we also have $\hat{\Qfunc}_T^{(n)} \rightarrow \Qfunc_T$.
		
		\textbf{Induction Step:} Suppose that the result is true for time $t$. We prove it for time $t-1$.
		
		By induction hypothesis and $g^{(n)}\rightarrow g$, we have
		\begin{align}
			\hat{V}_t^{(n), i}(h_t^i) =& \max_{\tilde{u}_t^i} \hat{\Qfunc}_t^{(n), i}(h_t^i, \tilde{u}_t^i)
			\xrightarrow{n\rightarrow\infty} \max_{\tilde{u}_t^i} \Qfunc_t^{i}(h_t^i, \tilde{u}_t^i).\label{eq:hatVconv}
		\end{align}
		
		Since $\Qfunc^{(n)} \rightarrow \Qfunc$ and $g^{(n)}\rightarrow g$, we have
		\begin{align}
			V_t^{(n), i}(h_t^i) =& \sum_{\tilde{u}_t^i} \Qfunc_t^{(n), i}(h_t^i, \tilde{u}_t^i)g_t^{(n), i}(\tilde{u}_t^i|h_t^i)\\
			\xrightarrow{n\rightarrow\infty}& \sum_{\tilde{u}_t^i} \Qfunc_t^{i}(h_t^i, \tilde{u}_t^i)g_t^{i}(\tilde{u}_t^i|h_t^i)=: V_t^i(h_t^i).\label{eq:KconvVconv}
		\end{align}
	
		Since $g$ is sequentially rational given $\Qfunc$, we have
		\begin{align}
			\sum_{\tilde{u}_t^i} \Qfunc_t^{i}(h_t^i, \tilde{u}_t^i)g_t^{i}(\tilde{u}_t^i|h_t^i) &= \max_{\tilde{u}_t^i} \Qfunc_t^{i}(h_t^i, \tilde{u}_t^i).\label{eq:Veqmax}
		\end{align}
	
		Combining \eqref{eq:hatVconv}\eqref{eq:KconvVconv}\eqref{eq:Veqmax} we have $\hat{V}_t^{(n), i}(h_t^i) \rightarrow V_t^{i}(h_t^i)$ for all $h_t^i\in\mathcal{H}_t^i$. Since $\mathcal{H}_t^i$ is a finite set, we have
		\begin{align}
			\max_{\tilde{h}_t^i}|\hat{V}_t^{(n), i}(\tilde{h}_t^i) - V_t^{(n), i}(\tilde{h}_t^i)| \xrightarrow{n\rightarrow\infty} 0.
		\end{align}
	
		We then have
		\begin{align}
			&\es|\hat{\Qfunc}_{t-1}^{(n), i}(h_t^i, u_t^i) -\Qfunc_{t-1}^{(n), i}(h_t^i, u_t^i)| \\
			&= \left| \sum_{\tilde{z}_{t-1}^i} \left[\hat{V}_{t}^{(n), i}((h_{t-1}^i, \tilde{z}_{t-1}^i)) - V_{t}^{(n), i}((h_{t-1}^i, \tilde{z}_{t-1}^i))\right] \Pr^{g_{t-1}^{(n), -i}}(\tilde{z}_{t-1}^i|h_{t-1}^i, u_{t-1}^i) \right|\label{eq:appSE:2:QV} \\
			&\leq \max_{\tilde{z}_{t-1}^i}|\hat{V}_{t}^{(n), i}((h_{t-1}^i, \tilde{z}_{t-1}^i)) - V_{t}^{(n), i}((h_{t-1}^i, \tilde{z}_{t-1}^i))|\xrightarrow{n\rightarrow\infty} 0,
		\end{align}
        {where we substitute \eqref{eq:appSE:2:Qnt}\eqref{eq:appSE:2:hatQnt} in \eqref{eq:appSE:2:QV}.}
		Since $\Qfunc_{t-1}^{(n), i}(h_t^i, u_t^i) \rightarrow \Qfunc_{t-1}^{i}(h_t^i, u_t^i)$, we {conclude that} $\hat{\Qfunc}_{t-1}^{(n), i}(h_t^i, u_t^i) \rightarrow \Qfunc_{t-1}^{i}(h_t^i, u_t^i)$, establishing the induction step.\\~
		
		\textbf{Step 4}: ``Model-free'' SE version 2 (Definition \ref{def:KSE2}) $\Rightarrow$ Classical SE (Definition \ref{def:classicalSE})\\~
		
		Let $(g, \Qfunc)$ be a sequential equilibrium under Definition \ref{def:KSE2} and let $(g^{(n)}, \hat{\Qfunc}^{(n)})$ satisfies conditions (1)(2')(3) of full consistency in Definition \ref{def:KSE2}.
		
		Define the beliefs $\mu^{(n)}$ on the nodes of the extensive-form game through $\mu^{(n)}(o_t^i|h_t^i) = \Pr^{g^{(n)}}(o_t^i|h_t^i)$. By taking subsequences, without lost of generality, assume that $\mu^{(n)} \rightarrow \mu$.
		
		Let $\hat{g}_{t}^i$ be an arbitrary strategy, then {by condition (2') of Definition \ref{def:KSE2}, we can write}	\begin{equation}\label{eq:maxrew2goghattau}
			\begin{split}
				&\es\sum_{\tilde{u}_t^i} \hat{\Qfunc}_t^{(n), i}(h_t^i, \tilde{u}_t^i) \hat{g}_{t}^i(\tilde{u}_t^i|h_{t}^i) \\
				&= \max_{\tilde{g}_{t+1:T}^i} \sum_{o_t^i} \E^{\hat{g}_t^i, \tilde{g}_{t+1:T}^i, g_{t}^{(n), >i}, g_{t+1:T}^{(n), -i}}[\Lambda_t^i(O_{T+1})|o_t^i]\ \mu_t^{(n), i}(o_t^i|h_t^i).
			\end{split}
		\end{equation}
		
		For each $o_t^i$, $\E^{\tilde{g}_{t}^{\geq i}, \tilde{g}_{t+1:T}}[\Lambda_t^i(O_{t+1})|o_t^i]$ is continuous in $(\tilde{g}_{t}^{\geq i}, \tilde{g}_{t+1:T})$ since it is the sum of product of components of $(\tilde{g}_{t}^{\geq i}, \tilde{g}_{t+1:T})$. Therefore, 
		\begin{align}
			&\sum_{o_t^i}\E^{\hat{g}_t^i, \tilde{g}_{t+1:T}^i, g_{t}^{(n), >i}, g_{t+1:T}^{(n), -i}}[\Lambda_t^i(O_{t+1})|o_t^i]\ \mu_t^{(n), i}(o_t^i|h_t^i)\notag\\
			\xrightarrow{n\rightarrow\infty}& \sum_{o_t^i}\E^{\hat{g}_t^i, \tilde{g}_{t+1:T}^i, g_{t}^{>i}, g_{t+1:T}^{ -i}}[\Lambda_t^i(O_{t+1})|o_t^i]\ \mu_t^{i}(o_t^i|h_t^i),
		\end{align}
		for each behavioral straetegy $\tilde{g}_{t+1:T}^i$. Applying Berge's Maximum Theorem \citep{sundaram1996first}, and taking the limit on both sides of \eqref{eq:maxrew2goghattau}, we obtain
		\begin{align}
			\sum_{\tilde{u}_t^i} \Qfunc_t^{i}(h_t^i, \tilde{u}_t^i)~ \hat{g}_{t}^i(\tilde{u}_t^i|h_{t}^i)=
			\max_{\tilde{g}_{t+1:T}^i} \sum_{o_t^i} \E^{\hat{g}_t^i, \tilde{g}_{t+1:T}^i, g_{t}^{>i}, g_{t+1:T}^{-i}}[\Lambda_t^i(O_{T+1})|o_t^i]~\mu_t^{i}(o_t^i|h_t^i),\label{eq:B50}
		\end{align}
		for all $t\in\mathcal{T}, i\in\mathcal{I}, h_t^i\in\mathcal{H}_t^i$, and all behavioral strategy $\hat{g}_{t}^i$.
		
		Sequential rationality of $g$ to $\Qfunc$ means that
		\begin{equation}
		\begin{split}
			g_{t}^i &\in \underset{\hat{g}_t^i}{\arg\max }~\sum_{\tilde{u}_t^i} \Qfunc_t^{i}(h_t^i, \tilde{u}_t^i) ~\hat{g}_{t}^i(\tilde{u}_t^i|h_{t}^i)\\
			&=\underset{\hat{g}_t^i}{\arg\max }~\max_{\tilde{g}_{t+1:T}^i} \sum_{o_t^i} \E^{\hat{g}_t^i, \tilde{g}_{t+1:T}^i, g_{t}^{>i}, g_{t+1:T}^{-i}}[\Lambda_t^i(O_{T+1})|o_t^i]~\mu_t^{i}(o_t^i|h_t^i),
		\end{split}
		\end{equation}
		for all $t\in\mathcal{T}, i\in\mathcal{I}$, and all $h_t^i\in\mathcal{H}_t^i$.
		
		  {Recall that the node $O_t^i$ uniquely determines $(X_1, W_{1:t-1}, U_{1:t-1})$. Therefore, the instantaneous rewards $R_\tau^i$ for $\tau \leq t-1$ are uniquely determined by $O_t^i$ as well. For $\tau \leq t-1$, let $r_\tau^i$ be realizations of $R_\tau^i$ under $O_t^i=o_t^i$. Recall that $\Lambda^i$ is the total reward function and $\Lambda_t^i$ is the reward-to-go function starting with (and including) time $t$.} We have $\E^{\hat{g}_t^i, \tilde{g}_{t+1:T}^i, g_{t}^{>i}, g_{t+1:T}^{-i}}[\Lambda^i(O_{T+1}) - \Lambda_t^i(O_{T+1})|o_t^i] = \sum_{\tau=1}^{t-1} r_\tau^i$ to be independent of the strategy profile. Therefore we have 
		
        \begin{equation}
			g_{t}^i \in\underset{\hat{g}_t^i}{\arg\max }~\max_{\tilde{g}_{t+1:T}^i} \sum_{o_t^i} \E^{\hat{g}_t^i, \tilde{g}_{t+1:T}^i, g_{t}^{>i}, g_{t+1:T}^{-i}}[\Lambda^i(O_{T+1})|o_t^i]~\mu_t^{i}(o_t^i|h_t^i).\label{eq:dynmu}
		\end{equation}
		
	Fixing $h_\tau^i$, the problem of optimizing 
		\begin{equation}
			J_\tau^i(\tilde{g}_{\tau:T}^i; h_\tau^i, \mu_\tau^i):= \sum_{o_\tau^i} \E^{\tilde{g}_{\tau:T}^i, g_{\tau}^{>i}, g_{\tau+1:T}^{-i}}[\Lambda^i(O_{T+1})|o_\tau^i]~\mu_\tau^{i}(o_\tau^i|h_\tau^i),
		\end{equation}
		over all $\tilde{g}_{\tau:T}^i$ is a POMDP problem with

        \begin{itemize}
			\item \small Timestamps $\tilde{T} = \{\tau, \tau+1, \cdots, T, T+1\}$;
			\item State process $(O_t^i)_{t=\tau}^T\cup (O_{T+1})$;
            \item Control actions $(U_t^i)_{t=\tau}^T$;
			\item Initial state distribution $\mu_\tau^i(h_\tau^i)\in\Delta(\mathcal{O}_{\tau}^i)$;
			\item State transition kernel $\Pr^{g_{t}^{>i}, g_{t+1}^{<i}}(o_{t+1}^i|o_t^i, u_t^i)$ for $t<T$ and $\Pr^{g_{T}^{>i}}(o_{T+1}|o_T^i, u_T^i)$ for $t=T$;
			\item Observation history: $(H_t^i)_{t=\tau}^T$;
			\item Instantaneous rewards are 0. Terminal reward is $\Lambda^i(O_{T+1})$.
		\end{itemize}
 
		The belief $\mu$ is fully consistent with $g$ by construction. From standard results in game theory, we know that $\mu_{t+1}^i(h_{t+1}^i)$ can be updated with Bayes rule from $\mu_{t}^i(h_{t}^i)$ and $g$ whenever applicable. Therefore, $(\mu_t)_{t=\tau}^T$ represent the true beliefs of the state given observations in the above POMDP problem. Therefore, through standard control theory \citep[Section 6.7]{kumar1986stochastic}, \eqref{eq:dynmu} is a sufficient condition for $g_{t:T}^i$ to be optimal for the above POMDP problem, which means that $g$ is sequentially rational given $\mu$.
		
		Therefore we conclude that $(g, \mu)$ is a sequential equilibrium under Definition \ref{def:classicalSE}.
	
\end{proof}

\section{Proofs for Sections \ref{sec:twoinfostates} and \ref{sec:isbe}}\label{app:proofsmain}

\subsection{Proof of Lemma \ref{lem:msi}}\label{app:lem:msi}

\begin{lemma}[Lemma \ref{lem:msi}, restated]
	If for all $i\in\mathcal{I}$ and all $\Comp^{-i}$-based strategy profiles $\rho^{-i}$, there exist functions $(\Phi_t^{i, \rho^{-i}})_{t\in\mathcal{T}}$ where $\Phi_t^{i, \rho^{-i}}\colon\Compset_t^i \mapsto \Delta(\mathcal{X}_t\times \Compset_t^{-i})$ such that
	\begin{equation}
		\Pr^{g^i, \rho^{-i}}(x_t, \comp_t^{-i}|h_t^i) = \Phi_t^{i,\rho^{-i}}(x_t, \comp_t^{-i}|\comp_t^i),
	\end{equation}
	for all behavioral strategies $g^i$, all $t\in\mathcal{T}$, and all $h_t^i$ admissible under $(g^i, \rho^{-i})$, then $\Comp=(\Comp^i)_{i\in\mathcal{I}}$ is mutually sufficient information.
\end{lemma}

\begin{proof}

{Let $g^i$ be an arbitrary behavioral strategy for player $i$ and $\rho^{-i}$ be any $K^{-i}$-based strategy profile. Let $h_t^i$ be admissible under $(g^i, \rho^{-i})$.} We have	
\begin{align}
		\Pr^{g^i, \rho^{-i}}(\tilde{x}_t, \tilde{u}_t^{-i}|h_t^i) &= \sum_{\tilde{h}_t^{-i}} \Pr^{g^i, \rho^{-i}}(\tilde{u}_t^{-i}|\tilde{x}_t, \tilde{h}_t^{-i}, h_t^i, u_t^i)\Pr^{g^i, \rho^{-i}}(\tilde{x}_t, \tilde{h}_t^{-i}|h_t^i, u_t^i)\label{eq:lemmsi:1} \\
		&=\sum_{\tilde{h}_t^{-i}}\left(\prod_{j\neq i} \rho_t^{j}(\tilde{u}_t^j|\tilde{\comp}_t^j) \right)\Pr^{g^i, \rho^{-i}}(\tilde{x}_t, \tilde{h}_t^{-i}|h_t^i)\\
		&=\sum_{\tilde{\comp}_t^{-i}}\left(\prod_{j\neq i} \rho_t^{j}(\tilde{u}_t^j|\tilde{\comp}_t^j) \right)\Pr^{g^i, \rho^{-i}}(\tilde{x}_t, \tilde{\comp}_t^{-i}|h_t^i)\label{eq:lemmsi:3} \\
		&=\sum_{\tilde{\comp}_t^{-i}}\left( \prod_{j\neq i} \rho_t^{j}(\tilde{u}_t^j|\tilde{\comp}_t^j) \right)\Phi_t^{i, \rho^{-i}}(\tilde{x}_t, \tilde{\comp}_t^{-i}|\comp_t^i),\label{eq:lemmsi:4}
	\end{align}
    {where in \eqref{eq:lemmsi:1} we applied the Law of Total Probability. In \eqref{eq:lemmsi:3} we combined the realizations of $\Tilde{h}_t^{-i}$ corresponding to the same compressed information $\Tilde{k}_t^{-i}$. In the final equation, we used the condition of Lemma \ref{lem:msi}.}
 
	By the definition of the model, $Z_t^i = f_t^{i,Z}(X_t, U_t, W_t)$ for some fixed function $f_t^{i,Z}$ independent of the strategy profile. 
    Since the compressed information can be sequentially updated as $\Comp_{t+1}^i=\iota_{t+1}^i(\Comp_t^i, Z_t^i)$, this means that we can write $\Comp_{t+1}^i = \xi_t^i(\Comp_t^i, X_t, U_t, W_t)$ for some fixed function $\xi_t^i$. Since $W_t$ is a primitive random variable, we conclude that $\Pr(\comp_{t+1}^i| \comp_t^i, x_t, u_t)$ is independent of any strategy profile. Therefore,
	\begin{align}
		&\quad~\Pr^{g^i, \rho^{-i}}(\comp_{t+1}^i|h_t^i, u_t^i)\label{eq:C59} \\
		&= \sum_{\tilde{x}_t, \tilde{u}_t^{-i}} \Pr(\comp_{t+1}^i|\comp_t^i, \tilde{x}_t, (\tilde{u}_t^{-i}, u_t^i) )  \Pr^{g^i, \rho^{-i}}(\tilde{x}_t, \tilde{u}_t^{-i}|h_t^i)\label{eq:lemmsi:C59expansion} \\
		&=\sum_{\tilde{x}_t, \tilde{u}_t^{-i}} \left[ \Pr(\comp_{t+1}^i|\comp_t^i, \tilde{x}_t, (\tilde{u}_t^{-i}, u_t^i) ) \sum_{\tilde{\comp}_t^{-i}}\left( \prod_{j\neq i} \rho_t^{j}(\tilde{u}_t^j|\tilde{\comp}_t^j) \right)\Phi_t^{i, \rho^{-i}}(\tilde{x}_t, \tilde{\comp}_t^{-i}|\comp_t^i) \right] \label{eq:lemmsi:giantthing} \\
		&=:P_t^{i, \rho^{-i}}(\comp_{t+1}^i|\comp_t^i, u_t^i),\label{eq:proofmsi:transition}
	\end{align}
	for some function $P_t^{i, g^{-i}}$, {where in \eqref{eq:lemmsi:C59expansion} we used the Law of Total Probability, and we substituted \eqref{eq:lemmsi:4} in \eqref{eq:lemmsi:giantthing}. }
	
	Since $R_t^i = f_t^{i,R}(X_t, U_t, W_t)$ for some fixed function $f_t^{i,R}$ and $W_t$ is a primitive random variable, we have $\E[R_t^i|X_t, U_t]$ to be independent of the strategy profile $g$. By an argument similar to the one that leads from \eqref{eq:C59} to \eqref{eq:proofmsi:transition} we obtain
	\begin{align}
		&\quad~\E^{g^i, \rho^{-i}}[R_t^i|h_t^i, u_t^i] \\
        &= \sum_{\tilde{x}_t, \tilde{u}_t^{-i}} \left[ \E[R_t^i|\tilde{x}_t, (u_t^i, \tilde{u}_t^{-i})] \sum_{\tilde{\comp}_t^{-i}}\left( \prod_{j\neq i} \rho_t^{j}(\tilde{u}_t^j|\tilde{\comp}_t^j) \right)\Phi_t^{i, \rho^{-i}}(\tilde{x}_t, \tilde{\comp}_t^{-i}|\comp_t^i) \right]  \\
		&=: r_t^{i, \rho^{-i}}(\comp_t^i, u_t^i),\label{eq:proofmsi:reward}
	\end{align}
	for some function $r_i^{i, \rho^{-i}}$. With \eqref{eq:proofmsi:transition} and \eqref{eq:proofmsi:reward}, we have shown that $\Comp$ satisfies Definition \ref{def:msi} and hence $\Comp$ is MSI.

\end{proof}

\subsection{Proof of Theorem \ref{thm:msiexist}}\label{app:thm:msiexist}

\begin{thm}[Theorem \ref{thm:msiexist}, restated]
	If $\Comp$ is mutually sufficient information, then there exists at least one $\Comp$-based BNE.
\end{thm}

\begin{proof}
{The proof will proceed as follows: We first construct a best-response correspondence using stochastic control theory, and then we establish the existence of equilibria by applying Kakutani's fixed-point theorem to this correspondence. For technical reasons, we first consider only behavioral strategies where each action has probability at least $\epsilon>0$ of being played at each information set. We then take $\epsilon$ to zero. }\\~

    Fixing {a $K^{-i}$-based strategy profile} $\rho^{-i}$, we first argue that $\Comp_t^i$ is a controlled Markov process controlled by player $i$'s action $U_t^i$. 
	
	From the definition of an information state (Definition \ref{def:infostate}) we know that
	\begin{align}
		\Pr^{\tilde{g}^i, \rho^{-i}}(\comp_{t+1}^i|h_t^i, u_t^i) = P_t^{i, \rho^{-i}}(\comp_{t+1}^i|\comp_t^i, u_t^i).
	\end{align}
	
	Since $(\Comp_{1:t}^i, U_{1:t}^i)$ is a function of $(H_t^i, U_t^i)$, by the smoothing property of conditional probability we have
	\begin{align}
		\Pr^{\tilde{g}^i, \rho^{-i}}(\comp_{t+1}^i|\comp_{1:t}^i, u_{1:t}^i) = P_t^{i, \rho^{-i}}(\comp_{t+1}^i|\comp_t^i, u_t^i).
	\end{align}
	
	Therefore we have shown that $\Comp_t^i$ is a controlled Markov process controlled by player $i$'s action $U_t^i$.
	
	From the definition of information state (Definition \ref{def:infostate}) we know that
	\begin{align}
		\E^{\tilde{g}^i, \rho^{-i}}\left[R_t^i|\comp_t^i, u_t^i \right] = r_t^{i, \rho^{-i}}(\comp_t^i, u_t^i),
	\end{align}
	for all $(\comp_t^i, u_t^i)$ admissible under $(\tilde{g}^i, \rho^{-i})$. 
	
	Therefore, using the Law of Total Expectation we have
	\begin{align}
		J^i(\tilde{g}^i, \rho^{-i}) &= \E^{\tilde{g}^i, \rho^{-i}}\left[\sum_{t=1}^T R_t^i\right] = \E^{\tilde{g}^i, \rho^{-i}}\left[\sum_{t=1}^T \E^{\tilde{g}^i, \rho^{-i}}\left[R_t^i|\Comp_t^i, U_t^i \right] \right]\\
		&=\E^{\tilde{g}^i, \rho^{-i}}\left[\sum_{t=1}^T r_t^{i, \rho^{-i}}(\Comp_t^i, U_t^i) \right].
	\end{align}
	
	By standard MDP theory, there exist $\Comp^i$-based strategies $\rho^i$ that maximize $J^i(\tilde{g}^i, \rho^{-i})$ over all behavioral strategies $\tilde{g}^i$. Furthermore, optimal $\Comp^i$-based strategies can be found through dynamic programming.
	
	%\vjcomment{Add an outline of proof and a short why too.} 
    Assume $\epsilon> 0$, let $\mathcal{P}^{\epsilon, i}$ denote the set of $\Comp^i$-based strategies for player $i$ where each action $u_t^i\in \mathcal{U}_t^i$ is chosen with probability at least $\epsilon$ at any information set. To endow $\mathcal{P}^{\epsilon, i}$ with a topology, we consider it as a product of sets of distributions, i.e.
	\begin{align}
		\mathcal{P}^{\epsilon, i} = \prod_{t\in\mathcal{T}}\prod_{\comp_t^i\in\Compset_t^i} \Delta^\epsilon(\mathcal{U}_t^i),
	\end{align}
	where
	\begin{align}
		\Delta^\epsilon(\mathcal{U}_t^i) = \{\eta \in\Delta(\mathcal{U}_t^i): \eta(u_t^i)\geq \epsilon~\forall u_t^i\in\mathcal{U}_t^i \}.
	\end{align}
	
	Define $\mathcal{P}^\epsilon = \prod_{i\in\mathcal{I}} \mathcal{P}^{\epsilon, i}$. Denote the set of all $\Comp^i$-based strategy profiles by $\mathcal{P}^0$.
	
	For the rest of the proof, assume that $\epsilon$ is small enough such that $\Delta^\epsilon(\mathcal{U}_t^i)$ is non-empty for all $t\in\mathcal{T}$ and $i\in\mathcal{I}$. 
	
	For each $t\in\mathcal{T}$, $i\in\mathcal{I}$ and $\comp_t^i\in\Compset_t^i$, define the correspondence $\mathrm{BR}_t^{\epsilon, i}[\comp_t^i]: \mathcal{P}^{\epsilon, -i} \mapsto \Delta^\epsilon(\mathcal{U}_t^i)$ sequentially through %\vjmargincomment{Check $V$ versus $J$ usage for value functions.}\vjmargincomment{Done.}
 \begin{subequations}\label{defofBR}
		\begin{align}
			\Qfunc_{T}^{\epsilon, i}(\comp_T^i, u_T^i;\rho^{-i})&:=r_T^{i, \rho^{-i}}(\comp_T^i, u_T^i)\label{defofBR:ini}, \\
			\mathrm{BR}_t^{\epsilon, i}[\comp_t^i](\rho^{-i}) &:= \underset{\eta\in \Delta^\epsilon(\mathcal{U}_t^i)}{\arg\max}~ \sum_{\tilde{u}_t^i}  \Qfunc_t^{\epsilon, i}(\comp_t^i, \tilde{u}_t^i;\rho^{-i}) \eta(\tilde{u}_t^i),\label{defofBR:BR} \\
			V_t^{\epsilon, i}(\comp_t^i;\rho^{-i})&:=\max_{\eta\in \Delta^\epsilon(\mathcal{U}_t^i)} \sum_{\tilde{u}_t^i}  \Qfunc_t^{\epsilon, i}(\comp_t^i, \tilde{u}_t^i;\rho^{-i}) \eta(\tilde{u}_t^i), \label{defofBR:V}\\
			\Qfunc_{t-1}^{\epsilon, i}(\comp_{t-1}^i, u_{t-1}^i; \rho^{-i})&:=r_{t-1}^{i, \rho^{-i}}(\comp_{t-1}^i, u_{t-1}^i) +\label{defofBR:K}\\
			+ \sum_{\comp_{t}^i\in\Compset_{t}^i } &V_{t}^{\epsilon, i}(\comp_t^i ; \rho^{-i}) P_{t-1}^{i, \rho^{-i}}(\comp_{t}^i|\comp_{t-1}^i, u_t^i).
		\end{align}
	\end{subequations}
	
	\noeqref{defofBR:ini}\noeqref{defofBR:BR}\noeqref{defofBR:V}\noeqref{defofBR:K}
	
	Define $\mathrm{BR}^{\epsilon}: \mathcal{P}^\epsilon \mapsto \mathcal{P}^\epsilon$ by
	\begin{align}
		\mathrm{BR}^{\epsilon}(\rho) = \prod_{i\in\mathcal{I}} \prod_{t\in\mathcal{T}} \prod_{\comp_t^i\in\Compset_t^i} \mathrm{BR}_t^{\epsilon, i}[\comp_t^i](\rho^{-i}).
	\end{align}
	
	\textbf{Claim}:
	\begin{enumerate}[(a)]
		\item $P_t^{i, \rho^{-i}}(\comp_{t+1}^i|\comp_{t}^i, u_t^i)$ is continuous in $\rho^{-i}$ on $\mathcal{P}^{\epsilon, -i}$ for all $t\in\mathcal{T}$ and all $\comp_{t+1}^i\in\Compset_{t+1}^i, \comp_t^i\in\Compset_t^i, u_t^i\in\mathcal{U}_t^i$.
		\item $r_t^{i, \rho^{-i}}(\comp_t^i, u_t^i)$ is continuous in $\rho^{-i}$ on $\mathcal{P}^{\epsilon, -i}$ for all $t\in\mathcal{T}$ and all $\comp_t^i\in\Compset_t^i, u_t^i\in\mathcal{U}_t^i$.
	\end{enumerate}
	
	Given the claims we prove by induction that $\Qfunc_t^{\epsilon, i}(\comp_t^i, u_{t}^i; \rho^{-i})$ is continuous in $\rho^{-i}$ on $\mathcal{P}^{\epsilon, -i}$ for each $\comp_t^i\in\Compset_t^i, u_t^i\in\mathcal{U}_t^i$.
	
	\textbf{Induction Base}: {$\Qfunc_T^{\epsilon, i}(\comp_T^i, u_{T}^i; \rho^{-i}) = r_T^{i, \rho^{-i}}(\comp_T^i, u_T^i)$ is continuous in $\rho^{-i}$ on $\mathcal{P}^{\epsilon, -i}$ due to part (a) of the claims.}
	
	\textbf{Induction Step}: Suppose that the induction hypothesis is true for $t$. Then $V_t^{\epsilon, i}(\comp_t^i; \rho^{-i})$ is continuous in $\rho^{-i}$ on $\mathcal{P}^{\epsilon, -i}$ due to Berge's Maximum Theorem \citep{sundaram1996first}. Then, $\Qfunc_{t-1}^{\epsilon, i}(\comp_{t-1}^i, u_{t-1}^i; \rho^{-i})$ is continuous in $\rho^{-i}$ on $\mathcal{P}^{\epsilon, -i}$ due to the claims.\\
	
	Applying Berge's Maximum Theorem \citep{sundaram1996first} once again, we conclude that $\mathrm{BR}_t^{\epsilon, i}[\comp_t^i]$ is upper hemicontinuous on $\mathcal{P}^{\epsilon, -i}$. For each $\rho^{-i} \in \mathcal{P}^{\epsilon, -i}$, $\mathrm{BR}_t^{\epsilon, i}[\comp_t^i](\rho^{-i})$ is non-empty and convex since it is the solution set of a linear program.
	
	As a product of compact-valued upper hemicontinuous correspondences, $\mathrm{BR}^\epsilon$ is upper hemicontinuous. For each $\rho \in \mathcal{P}^{\epsilon}$, $\mathrm{BR}^\epsilon(\rho)$ is non-empty and convex. By Kakutani's fixed point theorem, $\mathrm{BR}^\epsilon$ has a fixed point. 
 
% \vjcomment{Maybe pause to highlight that these fixed points are not NE but $\delta$ approximate NE for some $\delta$ depending on $\epsilon$. Then you can say that the idea is to take a limit of $\epsilon$ going to $0$.} 

    {The above construction provides an approximate $K$-based BNE for small $\epsilon$. Next, we show that we can take $\epsilon$ to zero to obtain an exact BNE:}
	Let $(\epsilon_n)_{n=1}^{\infty}$ be a sequence such that $\epsilon_n > 0, \epsilon_n\rightarrow 0$. Let $\rho^{(n)}$ be a fixed point of $\mathrm{BR}^{\epsilon_n}$. Then for each $i\in \mathcal{I}$ we have
	\begin{equation}
		\rho^{(n), i} \in \underset{\tilde{\rho}^i\in \mathcal{P}^{\epsilon_n, i}}{\arg\max}~J^i(\tilde{\rho}^i, \rho^{(n), -i}).
	\end{equation}
	
	Let $\rho^{(\infty)} \in \mathcal{P}^0$ be the limit of a sub-sequence of $(\rho^{(n)})_{n=1}^{\infty}$. Since $J^i(\rho)$ is continuous in $\rho$ on $\mathcal{P}^0$, and $\epsilon \mapsto \mathcal{P}^{\epsilon, i}$ is a continuous correspondence with compact, non-empty value, through applying Berge's Maximum Theorem \citep{sundaram1996first} one last time, we conclude that for each $i\in \mathcal{I}$
	\begin{equation}
		\rho^{(\infty), i} \in \underset{\tilde{\rho}^i\in \mathcal{P}^{0, i}}{\arg\max}~J^i(\tilde{\rho}^i, \rho^{(\infty), -i}),
	\end{equation}
	i.e. $\rho^{(\infty), i}$ is optimal among $\Comp^i$-based strategies in response to $\rho^{(\infty), -i}$. Recall that we have shown that there exist $\Comp^i$-based strategies $\rho^i$ that maximizes $J^i(\tilde{g}^i, \rho^{-i})$ over all behavioral strategies $\tilde{g}^i$. Therefore, we conclude that $\rho^{(\infty)}$ forms a BNE, proving the existence of $\Comp$-based BNE.\\~
	
	\textbf{Proof of Claim}: 
    {We establish the continuity of the two functions by showing that they can be expressed with basic functions (i.e. summation, multiplication, division).}
 
	Let $\hat{g}^i$ be a behavioral strategy where player $i$ chooses actions uniformly at random at every information set. For $\rho^{-i} \in\mathcal{P}^{\epsilon, -i}$, we have $\Pr^{\hat{g}^i, \rho^{-i}}(\comp_t^i) > 0$ for all $\comp_t^i\in\Compset_t^i$ since $(\hat{g}^i, \rho^{-i})$ is a strategy profile that always plays strictly mixed actions. Therefore we have
	\begin{align}
		P_t^{i, \rho^{-i}}(\comp_{t+1}^i|\comp_t^i, u_t^i) &= \Pr^{\hat{g}^i, \rho^{-i}}(\comp_{t+1}^i| \comp_t^i, u_t^i) = \dfrac{\Pr^{\hat{g}^i, \rho^{-i}}(\comp_{t+1}^i, \comp_t^i, u_t^i) }{\Pr^{\hat{g}^i, \rho^{-i}}(\comp_t^i, u_t^i)},\\
		r_t^{i, \rho^{-i}}(\comp_t^i, u_t^i) &= \E^{\hat{g}^i, \rho^{-i}}[R_t^i| \comp_t^i, u_t^i] \\
		&=\sum_{x_t\in \mathcal{X}_t, u_t^{-i}\in\mathcal{U}_t } \E[R_t^i| x_t, u_t] \Pr^{\hat{g}^i, \rho^{-i}}(x_t, u_t^{-i}| \comp_t^i, u_t^i),
	\end{align}
	where $\E[R_t^i| x_t, u_t]$ is independent of the strategy profile.
	
	We know that both $\Pr^{\hat{g}^i, \rho^{-i}}(\comp_{t+1}^i, \comp_t^i, u_t^i)$ and $\Pr^{\hat{g}^i, \rho^{-i}}(\comp_t^i, u_t^i)$ are sums of products of components of $\rho^{-i}$ and $\hat{g}^i$, hence both are continuous in $\rho^{-i}$. Therefore $P_t^{i, \rho^{-i}}(z_{t}^i|\comp_t^i, u_t^i)$ is continuous in $\rho^{-i}$ on $\mathcal{P}^{\epsilon, -i}$. The continuity of $r_t^{i, \rho^{-i}}(\comp_t^i, u_t^i)$ in $\rho^{-i}$ on $\mathcal{P}^{\epsilon, -i}$  can be shown with an analogous argument.

\end{proof}
 
\subsection{Proof of Theorem \ref{thm:usiequiv}}\label{app:thm:usiequiv}
\begin{thm}[Theorem \ref{thm:usiequiv}, restated]
	If $\Comp=(\Comp^i)_{i\in\mathcal{I}}$ where $\Comp^i$ is unilaterally sufficient information for player $i$, then the set of $\Comp$-based BNE payoffs is the same as that of all BNE.
\end{thm}

{
To establish Theorem \ref{thm:usiequiv}, we first introduce Definition \ref{def:infostateforj}, an extension of Definition \ref{def:infostate}, for the convenience of the proof. Then, we establish Lemmas \ref{lem:usiisinfostate}, \ref{lem:usiclaim}, \ref{lem:usireplace}. Finally, we conclude the proof of Theorem \ref{thm:usiequiv} from the three lemmas.}

{In the following definition, we provide an extension of the definition of the information state where not only player $i$'s payoff are considered. This definition allows us to characterize compression maps that preserve payoff profiles, as required in the statement of Theorem \ref{thm:usiequiv}.}

\begin{defn}\label{def:infostateforj}
	Let $g^{-i}$ be a behavioral strategy profile of players other than $i$ and $\mathcal{J}\subseteq \mathcal{I}$ be a subset of players. We say that $\Comp^i$ is an \emph{information state under $g^{-i}$ for the payoffs of $\mathcal{J}$} if there exist functions $(P_t^{i,g^{-i}})_{t\in\mathcal{T}}, (r_t^{j,g^{-i}})_{j\in\mathcal{J}, t\in\mathcal{T}}$, where $P_t^{i,g^{-i}}: \Compset_t^i\times \mathcal{U}_t^i \mapsto \Delta(\Compset_{t+1}^i)$ and $r_t^{j,g^{-i}}: \Compset_t^i\times \mathcal{U}_t^i \mapsto [-1, 1]$, such that
	\begin{enumerate}[(1)]
		\item $\Pr^{g^i, g^{-i}}(\comp_{t+1}^i|h_t^i, u_t^i) = P_t^{i,g^{-i}}(\comp_{t+1}^i|\comp_t^i, u_t^i)$ for all $t\in\mathcal{T}\backslash\{T\}$; and
		\item $\E^{g^i, g^{-i}}[R_t^j|h_t^i, u_t^i] = r_t^{j,g^{-i}}(\comp_t^i, u_t^i)$ for all $j\in\mathcal{J}$ and all $t\in \mathcal{T}$,
	\end{enumerate}
	for all $g^i$, and all $(h_t^i, u_t^i)$ admissible under $(g^i, g^{-i})$.
\end{defn}

Notice that condition (2) of Definition \ref{def:infostateforj} means that the information state $K^i$ is sufficient for evaluating \emph{other} agents' payoffs as well. This property is essential in establishing the preservation of payoff profiles of other agents when player $i$ switches to a compression-based strategy.

\begin{lemma}\label{lem:usiisinfostate}
	If $\Comp^i$ is unilaterally sufficient information, then $\Comp^i$ is an information state under $g^{-i}$ for the payoffs of $\mathcal{I}$ under all behavioral strategy profiles $g^{-i}$.
\end{lemma}

\begin{proof}[Proof of Lemma \ref{lem:usiisinfostate}]
	Let $\Phi_t^{i, g^{-i}}$ be as in the definition of USI (Definition \ref{def:usi}), we have
	\begin{align}
		\Pr^g(x_t, h_t^{-i}| h_t^i) &= \Phi_t^{i, g^{-i}}(x_t, h_t^{-i}|\comp_t^i).
	\end{align}
	
	{Applying the Law of Total Probability,}
	\begin{align}
		\Pr^g(\tilde{x}_t, \tilde{u}_t^{-i}|h_t^i) &= \sum_{\tilde{h}_t^{-i}} \Pr^g(\tilde{u}_t^{-i}|\tilde{x}_t, \tilde{h}_t^{-i}, h_t^i) \Pr^g(\tilde{x}_t, \tilde{h}_t^{-i}| h_t^i)\label{eq:usiproof:1} \\
		&= \sum_{\tilde{h}_t^{-i}} \left(\prod_{j\neq i}g_t^j(\tilde{u}_t^j|\tilde{h}_t^j) \right) \Phi_t^{i, g^{-i}}(\tilde{x}_t, \tilde{h}_t^{-i}| \comp_t^i)\\
		&=: \tilde{P}_t^{i, g^{-i}}(\tilde{x}_t, \tilde{u}_t^{-i}|\comp_t^i).
	\end{align}

	{We know that $\Comp_{t+1}^i = \iota_{t+1}^i(\Comp_t^i, Z_t^i) = \xi_t^i(\Comp_t^i, X_t, U_t, W_t)$ for some fixed function $\xi_t^i$ independent of the strategy profile $g$. Since $W_t$ is a primitive random variable, $\Pr(\comp_{t+1}^i | \comp_t^i, x_t, u_t)$ is independent of the strategy profile $g$.} Therefore,
	\begin{align}
		\Pr^g(\comp_{t+1}^i|h_t^i, u_t^i) &= \sum_{\tilde{x}_t, \tilde{u}_t^{-i}} \Pr(\comp_{t+1}^i | \comp_t^i, \tilde{x}_t, (\tilde{u}_t^{-i}, u_t^i)) \tilde{P}_t^{i, g^{-i}}(\tilde{x}_t, \tilde{u}_t^{-i}|\comp_t^i)\\
		&=:P_t^{i, g^{-i}}(\comp_{t+1}^i|\comp_t^i, u_t^i),
	\end{align}
    {establishing part (1) of Definition \ref{def:infostateforj}.}
	
	Consider any $j\in\mathcal{I}$. Since $R_t^j$ is a strategy-independent function of $(X_t, U_t, W_t)$, $\E[R_t^j|x_t, u_t]$ is independent of $g$. Therefore
	\begin{align}
		\E^g[R_t^j|h_t^i, u_t^i] &= \sum_{\tilde{x}_t, \tilde{u}_t^{-i}} \E[R_t^j|\tilde{x}_t, (u_t^i, \tilde{u}_t^{-i})] \tilde{P}_t^{i, g^{-i}}(\tilde{x}_t, \tilde{u}_t^{-i}|\comp_t^i)\\
		&=:r_t^{j, g^{-i}}(\comp_t^i, u_t^i),
	\end{align}
    {establishing part (2) of Definition \ref{def:infostateforj}.}
\end{proof}

{In Lemma \ref{lem:usiclaim}, we show that any behavioral strategy of player $i$ can be replaced by an equivalent randomized USI-based strategy while preserving payoffs of \emph{all} players.}

\begin{lemma}\label{lem:usiclaim}
	Let $\Comp^i$ be unilaterally sufficient information. Then for every behavioral strategy profile $g^i$, if the $\Comp^i$ based strategy $\rho^i$ is given by
	\begin{equation}\label{rhotheultimateequivalentstrategy}
		\rho_t^i(u_t^i|\comp_t^i) = \sum_{\tilde{h}_t^i\in\mathcal{H}_t^i}g_t^i(u_t^i|\tilde{h}_t^i)F_t^{i, g^i}(\tilde{h}_t^i|\comp_t^i),
	\end{equation}
	where $F_t^{i, g^i}(\tilde{h}_t^i|\comp_t^i)$ is defined in Definition \ref{def:usi}, then
	$$J^j(g^i, g^{-i}) = J^j(\rho^i, g^{-i}),$$
	for all $j\in\mathcal{I}$ and all behavioral strategy profiles $g^{-i}$ of players other than $i$.
\end{lemma}

\begin{proof}[Proof of Lemma \ref{lem:usiclaim}]
Let $j\in\mathcal{I}$. Consider an MDP with state $H_t^i$, action $U_t^i$ and instantaneous reward $\tilde{r}_t^{i, j}(h_t^i, u_t^i):=\E^{g^{-i}}[R_t^j|h_t^i, u_t^i]$. By Lemma \ref{lem:usiisinfostate}, $\Comp^i$ is an information state (as defined in Definition \ref{def:app:infostate}) for this MDP. Hence $J^j(g^i, g^{-i}) = J^j(\rho^i, g^{-i})$ follows from the Policy Equivalence Lemma (Lemma \ref{lem:policyeval}).
\end{proof} 

In Lemma \ref{lem:usireplace}, we proceed to show that a behavioral strategy can be replaced with an USI-based strategy while preserving not only the payoffs of all players, but also the \emph{equilibrium}.

\begin{lemma}\label{lem:usireplace}
	If $\Comp^i$ is unilaterally sufficient information for player $i$, then for any BNE strategy profile $g=(g^i)_{i\in\mathcal{I}}$ there exists a $\Comp^i$-based strategy $\rho^i$ such that $(\rho^i, g^{-i})$ forms a BNE with the same expected payoff profile as $g$.
\end{lemma}

\begin{proof}[Proof of Lemma \ref{lem:usireplace}]
	Let $\rho^i$ be associated with $g^i$ as specified in Lemma \ref{lem:usiclaim}. Set $\bar{g} = (\rho^i, g^{-i})$. Since $J^i(\rho^i, g^{-i}) = J^i(g^i, g^{-i})$ and $g^i$ is a best response to $g^{-i}$, $\rho^i$ is also a best response to $g^{-i}$.
	
	Consider $j\neq i$. Let $\tilde{g}^j$ be an arbitrary behavioral strategy of player $j$. By using Lemma \ref{lem:usiclaim} twice we have
	\begin{align}
		J^j(\bar{g}^j, \bar{g}^{-j}) &= J^j(\rho^i, g^{-i}) = J^j(g)\geq J^j(\tilde{g}^j, g^{-j}) \\
		&= J^j(\tilde{g}^j, (\rho^i, g^{-\{i, j\}})) = J^j(\tilde{g}^j, \bar{g}^{-j}).
	\end{align}
	Therefore $\bar{g}^j$ is a best response to $(\rho^i, g^{-\{i, j\}})$. We conclude that $\bar{g} = (\rho^i, g^{-i})$ is also a BNE.
\end{proof}

\begin{proof}[Proof of Theorem \ref{thm:usiequiv}]
	Given any BNE strategy profile $g$, applying Lemma \ref{lem:usireplace} iteratively for each $i\in\mathcal{I}$, we obtain a $\Comp$-based BNE strategy profile $\rho$ with the same expected payoff profile as $g$. Therefore the set of $\Comp$-based BNE payoffs is the same as that of all BNE.
\end{proof}

\subsection{Proof of Theorem \ref{thm:msiseexist}}\label{app:thm:msiseexist}
\begin{thm}[Theorem \ref{thm:msiseexist}, restated]
	If $\Comp$ is mutually sufficient information, then there exists at least one $\Comp$-based sequential equilibrium.
\end{thm}

\begin{proof}
    {The proof of Theorem \ref{thm:msiseexist} follows similar steps to that of Theorem \ref{thm:msiexist}, where we construct a sequence of strictly mixed strategy profiles via the fixed points of dynamic program based best response mappings. In addition, we show the sequential rationality of the strategy profile constructed.}
    
	Let $(\rho^{(n)})_{n=1}^\infty$ be a sequence of $\Comp$-based strategy profiles that always assigns strictly mixed actions as constructed in the proof of Theorem \ref{thm:msiexist}. By taking a sub-sequence, without loss of generality, assume that $\rho^{(n)}\rightarrow \rho^{(\infty)}$ for some $\Comp$-based strategy profile $\rho^{(\infty)}$. 
	
	Let $\Qfunc^{(n)}$ be conjectures of reward-to-go functions consistent (in the sense of Definition \ref{def:KSE}) with $\rho^{(n)}$, i.e.
	\begin{align}
		\Qfunc_\tau^{(n), i}(h_t^i, u_t^i) := \E^{\rho^{(n)}}\left[\sum_{t=\tau}^T R_t^i\Big|h_\tau^i, u_\tau^i\right].
	\end{align}
	
	Let $\Qfunc^{(\infty)}$ be the limit of a sub-sequence of $(\Qfunc^{(n)})_{n=1}^\infty$ (such a limit exists since the range of each $\Qfunc_\tau^{(n), i}$ is a compact set). We proceed to show that $(\rho^{(\infty)}, \Qfunc^{(\infty)})$ forms a sequential equilibrium (as defined in Definition \ref{def:KSE}). Note that by construction, $\Qfunc^{(\infty)}$ is fully consistent with $\rho^{(\infty)}$. We only need to show sequential rationality.
	
	\textbf{Claim:} Let $\Qfunc_t^{\epsilon, i}$ be as defined in \eqref{defofBR} in the proof of Theorem \ref{thm:msiexist}, then
	\begin{equation}
		\Qfunc_t^{(n), i}(h_t^i, u_t^i) = \Qfunc_t^{\epsilon_n, i}(\comp_{t}^i, u_t^i; \rho^{(n), -i}),
	\end{equation}
	for all $i\in\mathcal{I}, t\in\mathcal{T}, h_t^i\in\mathcal{H}_t^i$, and $u_t^i\in\mathcal{U}_t^i$.
	
	By construction in the proof of Theorem \ref{thm:msiexist}, $\rho_t^{(n), i}(\comp_t^i)\in \mathrm{BR}_t^{\epsilon_n, i}[\comp_t^i](\rho^{(n), -i})$. Given the claim, this means that
	\begin{equation}
		\rho_{t}^{(n), i}(\comp_t^i) \in \underset{\eta\in \Delta^{\epsilon_n}(\mathcal{U}_{t}^{i}) }{\arg\max}~\sum_{\tilde{u}_{t}^i} \Qfunc_t^{(n), i}(h_{t}^i, \tilde{u}_{t}^i) \eta(\tilde{u}_{t}^i),
	\end{equation}
	for all $i\in\mathcal{I}, t\in \mathcal{T}$ and $h_{t}^i\in \mathcal{H}_{t}^i$.
	
	Applying Berge's Maximum Theorem \citep{sundaram1996first} in a similar manner to the proof of Theorem \ref{thm:msiexist} we obtain
	\begin{equation}
		\rho_{t}^{(\infty), i}(\comp_t^i) \in \underset{\eta\in \Delta(\mathcal{U}_{t}^{i}) }{\arg\max}~\sum_{\tilde{u}_{t}^i} \Qfunc_t^{(\infty), i}(h_{t}^i, \tilde{u}_{t}^i) \eta(\tilde{u}_{t}^i),
	\end{equation}
	for all $i\in\mathcal{I}, t\in \mathcal{T}$ and $h_{t}^i\in \mathcal{H}_{t}^i$. 
	
	Therefore, we have shown that $\rho^{(\infty)}$ is sequentially rational under $\Qfunc^{(\infty)}$ and we have completed the proof.
	
	\textbf{Proof of Claim:} For clarity of exposition we drop the superscript $(n)$ of $\rho^{(n)}$. We know that $\Qfunc_t^{(n), i}$ satisfies the following equations: %\vjmargincomment{Check $V$ and $J$ usage.}\vjmargincomment{Done}
    \begin{subequations}
	\begin{align}
		\Qfunc_{T}^{(n), i}(h_T^i, u_T^i)&=\E^{\rho}[R_T^i|h_T^i, u_T^i], \\
		V_t^{(n), i}(h_t^i)&:=\sum_{\tilde{u}_t^i}  \Qfunc_t^{(n), i}(h_t^i, \tilde{u}_t^i) \rho^{ i}_t(\tilde{u}_t^i|\comp_t^i),\\
		\Qfunc_{t-1}^{(n), i}(h_{t-1}^i, u_{t-1}^i)&:=
		\E^{\rho}[R_{t-1}^i|h_{t-1}^i, u_{t-1}^i] + \sum_{\tilde{h}_{t}^i\in\mathcal{H}_{t}^i } V_{t}^{(n), i}(\tilde{h}_t^i) \Pr^{\rho}(\tilde{h}_{t}^i|h_{t-1}^i, u_t^i).
	\end{align}
    \end{subequations}
	
	Since $\Comp$ is mutually sufficient information, we have
	\begin{align}
		\Pr^{\rho}(\comp_{t+1}^i|h_t^i, u_t^i) &:= P_t^{i, \rho^{-i}}(\comp_{t+1}^i|\comp_t^i, u_t^i),\\
		\E^{\rho}[R_t^i|h_t^i, u_t^i] &:= r_t^{i, \rho^{-i}}(\comp_t^i, u_t^i),
	\end{align}
	where $P_t^{i, \rho^{-i}}$ and $r_t^{i, \rho^{-i}}$ are as specified in Definition \ref{def:infostate}.
	
	Therefore, through an inductive argument, one can show then $\Qfunc_t^{(n), i}(h_t^i, u_t^i)$ depends on $h_t^i$ only through $\comp_t^i$, and
	\begin{subequations}\label{defofKinSE}
		\begin{align}
			\Qfunc_{T}^{(n), i}(\comp_T^i, u_T^i)&=r_{T}^{i, \rho^{-i}}(\comp_T^i, u_T^i), \label{defofKinSE:throwaway1}\\
			V_t^{(n), i}(\comp_t^i)&:=\sum_{\tilde{u}_t^i}  \Qfunc_t^{i}(\comp_t^i, \tilde{u}_t^i;\rho^{-i}) \rho^{i}_t(\tilde{u}_t^i|\comp_t^i),\label{defofKinSE:throwaway2}\\
			\quad \Qfunc_{t-1}^{(n), i}(\comp_{t-1}^i, u_{t-1}^i)&:=
			r_{t-1}^{i,\rho^{-i}}(\comp_{t-1}^i, u_{t-1}^i)+ \sum_{\tilde{\comp}_{t}^i\in\Compset_{t}^i } V_{t}^{(n), i}(\tilde{\comp}_t^i) P_{t-1}^{i,\rho^{-i}}(\tilde{\comp}_{t}^i|\comp_{t-1}^i, u_t^i).\label{defofKinSE:throwaway3}
		\end{align}
	\end{subequations}
	
	The claim is then established by comparing \eqref{defofKinSE} with \eqref{defofBR} and combining with the fact that $\rho_t^i(\comp_t^i)\in \mathrm{BR}_t^{\epsilon, i}[\comp_t^i] (\rho^{-i})$.
\end{proof}

\subsection{Proof of Theorem \ref{thm:usiseequiv}}\label{app:thm:usiseequiv}

\begin{thm}[Theorem \ref{thm:usiseequiv}, restated]
	If $\Comp=(\Comp^i)_{i\in\mathcal{I}}$ where $\Comp^i$ is unilaterally sufficient information for player $i$, then the set of $\Comp$-based sequential equilibrium payoffs is the same as that of all sequential equilibria.
\end{thm}

{To prove the assertion of Theorem \ref{thm:usiseequiv} we establish a series of technical results that appear in Lemmas \ref{lem:htihtjcondindepqti} - \ref{lem:usisereplace} below. The two key results needed for the proof of the theorem are provided by Lemmas \ref{lem:prequiv} and \ref{lem:usisereplace}. Lemma \ref{lem:prequiv} asserts that a player can switch to a USI-based strategy without changing the dynamic decision problems faced by the other players. The result of Lemma \ref{lem:prequiv} allows to establish the analogue of the payoff equivalence result of Lemma \ref{lem:policyeval} under the concept of sequential equilibrium. Lemma \ref{lem:usisereplace} asserts that any one player can switch to a USI-based strategy without affecting the sequential equilibrium (under perfect recall) and its payoffs. The proof of Lemma \ref{lem:prequiv} is based on two technical results provided by Lemmas \ref{lem:htihtjcondindepqti} and \ref{lem:usisatisoldconds}. The proof of Lemma \ref{lem:usisereplace} is based on Lemmas \ref{lem:prequiv} and \ref{lem:selfKfunc} which states that the history-action value function of a player $i\in\mathcal{I}$ can be expressed with their USI.}

\begin{lemma}\label{lem:htihtjcondindepqti}
	Suppose that $\Comp^i$ is unilaterally sufficient information. Then
	\begin{align}
		\Pr^g(h_t^i|h_t^j) &= \Pr^{g}(h_t^i|\comp_t^i)\Pr^g(\comp_t^i|h_t^j),\label{delaware}
	\end{align}
	whenever $\Pr^g(\comp_t^i) > 0, \Pr^g(h_t^j) > 0$.
\end{lemma}

\begin{proof}
	From the definition of unilaterally sufficient information (Definition \ref{def:usi}) we have 
	\begin{align}
		\Pr^{g}(\tilde{h}_t^i, \tilde{h}_t^j|\comp_t^i) &= F_t^{i, g^i}(\tilde{h}_t^i|\comp_t^i) F_t^{i,j, g^{-i}}(\tilde{h}_t^j|\comp_t^i), 
	\end{align}
	where
	\begin{align}
		F_t^{i,j, g^{-i}}(h_t^j|\comp_t^i) := \sum_{\tilde{x}_t, \tilde{h}_t^{-\{i, j\} }}\Phi_t^{i,g^{-i}}(\tilde{x}_t, (h_t^j, \tilde{h}_t^{-\{i, j\}})|\comp_t^i).
	\end{align}

	Therefore, we conclude that $H_t^i$ and $H_t^j$ are conditionally independent given $\Comp_t^i$. Since $\Comp_t^i$ is a function of $H_t^i$, we have
	\begin{align}
		\Pr^g(h_t^i|h_t^j) &= \Pr^g(h_t^i, \comp_t^i|h_t^j) =\Pr^{g}(h_t^i|\comp_t^i)\Pr^g(\comp_t^i|h_t^j).
	\end{align}
\end{proof}

\begin{lemma}\label{lem:usisatisoldconds}
 	Suppose that $\Comp^i$ is unilaterally sufficient information for player $i\in\mathcal{I}$. Then there exist functions $(\Pi_t^{j,i,g^{-\{i,j\}}})_{j\in\mathcal{I}\backslash\{i\},t\in\mathcal{T}}$, $(r_t^{i, j,g^{-\{i, j\}}})_{j\in\mathcal{I}\backslash\{i\}, t\in\mathcal{T}}$, where $\Pi_t^{i,j,g^{-\{i, j\}}}: \Compset_t^i\times \mathcal{H}_t^j \times \mathcal{U}_t^i \times \mathcal{U}_t^j \mapsto \Delta(\mathcal{H}_{t+1}^j)$, $r_t^{i,j,g^{-\{i,j\}}}: \Compset_t^i\times \mathcal{H}_t^j  \times \mathcal{U}_t^i \times \mathcal{U}_t^j \mapsto [-1, 1]$ such that
	\begin{enumerate}[(1)]
		%\item $\Pr^{g}(\comp_{t+1}^i|h_t^i, h_t^j, u_t^i, u_t^j) = P_t^{i, j, g^{-\{i, j\} }}(\comp_{t+1}^i|\comp_t^i, h_t^j, u_t^i, u_t^j)$ for all $t\in\mathcal{T}\backslash\{T\}$;
		\item $\Pr^{g}(\tilde{h}_{t+1}^j|h_t^i, h_t^j, u_t^i, u_t^j) = \Pi_t^{j, i, g^{-\{i, j\}}}(\tilde{h}_{t+1}^j|\comp_t^i, h_t^j, u_t^i, u_t^j)$ for all $t\in\mathcal{T}\backslash\{T\}$; and
		\item $\E^{g}[R_t^j|h_t^i, h_t^j, u_t^i, u_t^j] = r_t^{i, j, g^{-\{i, j\}}}(\comp_t^i, h_t^j, u_t^i, u_t^j)$ for all $t\in \mathcal{T}$,
		
	\end{enumerate}
	for all $j\in\mathcal{I}\backslash\{i\}$ and all behavioral strategy profiles $g$ whenever the left-hand side expressions are well-defined.
\end{lemma}

\begin{proof}[Proof of Lemma \ref{lem:usisatisoldconds}]
	Let $\hat{g}^l$ be some fixed, fully mixed behavioral strategy for player $l\in\mathcal{I}$. 
	
	Fix $j\neq i$. First,
	\begin{align}
		\Pr^g(x_t, h_t^{-\{i, j\}}| h_t^i, h_t^j) &= \Pr^{\hat{g}^{\{i, j\}}, g^{-\{i, j\}} } (x_t, h_t^{-\{i, j\}}| h_t^i, h_t^j)\label{eq:usisatisoldconds:1} \\
		&=\dfrac{\Phi_t^{i, (\hat{g}^j, g^{-\{i, j\}}) }(x_t, h_t^{-i}|\comp_t^i) }{\sum_{\tilde{x}_t,\tilde{h}_t^{-\{i, j\}}} \Phi_t^{i, (\hat{g}^j, g^{-\{i, j\}})}(\tilde{x}_t, (\tilde{h}_t^{-\{i, j\}}, h_t^j)|\comp_t^i)}\label{eq:usisatisoldconds:2} \\
		&=:\Phi_t^{i, j, g^{-\{i, j\}}}(x_t, h_t^{-\{i, j\}}| \comp_t^i, h_t^j),
	\end{align}
	for any behavioral strategy profile $g$, {where in \eqref{eq:usisatisoldconds:1} we used the fact that since $(h_t^i, h_t^j)$ are included in the conditioning, the conditional probability is independent of the strategies of player $i$ and $j$ \citep[Section 6.5]{kumar1986stochastic}. 
    In \eqref{eq:usisatisoldconds:2} we used Bayes rule and the definition of USI (Definition \ref{def:usi}).}
	
	Therefore, {using the Law of Total Probability,}
	\begin{align}
		\Pr^g(\tilde{x}_t, \tilde{u}_t^{-\{i,j\}}|h_t^i, h_t^j) &= \sum_{\tilde{h}_t^{-\{i, j\}}} \Pr^g(\tilde{u}_t^{-\{i,j\}}|\tilde{x}_t, \tilde{h}_t^{-\{i, j\}}, h_t^i, h_t^j) \Pr^g(\tilde{x}_t, \tilde{h}_t^{-\{i, j\}}| h_t^i, h_t^j)\\
		&= \sum_{\tilde{h}_t^{-\{i, j\}}} \left(\prod_{l \in \mathcal{I}\backslash\{i,j\} }g_t^l(\tilde{u}_t^l|\tilde{h}_t^l) \right) \Phi_t^{i, j, g^{-\{i, j\}}}(\tilde{x}_t, \tilde{h}_t^{-\{i, j\}}| \comp_t^i, h_t^j)\\
		&=: \tilde{P}_t^{i, j, g^{-\{i, j\}}}(\tilde{x}_t, \tilde{u}_t^{-\{i, j\}}|\comp_t^i, h_t^j),
	\end{align}
	for any behavioral strategy profile $g$.
	
	We know that $H_{t+1}^j = \xi_t^j(X_t, U_t, H_t^j)$ for some function $\xi_t^j$ independent of the strategy profile $g$, hence {using the Law of Total Probability we have}
	\begin{align}
		&\es\Pr^g(\tilde{h}_{t+1}^j|h_t^i, h_t^j, u_t^i, u_t^j) \\
		&= \sum_{\tilde{x}_t, \tilde{u}_t^{-\{i, j\}}} \bm{1}_{\{\tilde{h}_{t+1}^j = \xi_t^i(\tilde{x}_t, (u_t^{\{i, j\}},  \tilde{u}_t^{-\{i, j\}}), h_t^j) \} } \tilde{P}_t^{i, j, g^{-\{i, j\}}}(\tilde{x}_t, \tilde{u}_t^{-\{i, j\}}|\comp_t^i, h_t^j)\\
		&=:\Pi_t^{j, i, g^{-\{i, j\}}}(\tilde{h}_{t+1}^j|\comp_t^i, h_t^j, u_t^i, u_t^j),
	\end{align}
	establishing part (1) of Lemma \ref{lem:usisatisoldconds}.
	
	Since $\E[R_t^j|x_t, u_t]$ is strategy-independent, for $j\in\mathcal{I}\backslash\{i\}$, {using the Law of Total Expectation we have}
	\begin{align}
		\E^g[R_t^j|h_t^i, h_t^j, u_t^i, u_t^j] &= \sum_{\tilde{x}_t, \tilde{u}_t^{-i}} \E[R_t^j|\tilde{x}_t, (u_t^{\{i, j\}}, \tilde{u}_t^{-\{i, j\}})] \tilde{P}_t^{i, j, g^{-\{i, j\}}}(\tilde{x}_t, \tilde{u}_t^{-\{i, j\}}|\comp_t^i, h_t^j)\\
		&=:r_t^{i, j, g^{-\{i, j\}}}(\comp_t^i, h_t^j, u_t^i, u_t^j),
	\end{align}
	establishing part (2) of Lemma \ref{lem:usisatisoldconds}.
\end{proof}

\begin{lemma}\label{lem:prequiv}
	Suppose that $\Comp^i$ is unilaterally sufficient information. Let $g=(g^j)_{j\in\mathcal{I}}$ be a fully mixed behavioral strategy profile. Let a $\Comp^i$-based strategy $\rho^i$ be such that
	\begin{equation}
		\rho_t^{i}(u_t^i|\comp_t^i) = \sum_{\tilde{h}_t^i}g_t^{ i}(u_t^i|\tilde{h}_t^i)F_t^{i, g^{i}}(\tilde{h}_t^i|\comp_t^i).\label{florida}
	\end{equation}
	Then
	\begin{enumerate}[(1)]
		\item $\Pr^{g}(\tilde{h}_{t+1}^j|h_t^j, u_t^j) = \Pr^{\rho^i, g^{-i}}(\tilde{h}_{t+1}^j|h_t^j, u_t^j)$ for all $t\in\mathcal{T}\backslash\{T\}$; and
		\item $\E^g[R_t^j|h_t^j, u_t^j] = \E^{\rho^i, g^{-i}}[R_t^j|h_t^j, u_t^j]$ for all $t\in\mathcal{T}$,
	\end{enumerate}
	for all $j\in\mathcal{I}\backslash\{i\}$ and all $h_t^j\in \mathcal{H}_t^j, u_t^j\in\mathcal{U}_t^j$.
\end{lemma}

\begin{proof}
	Fixing $g^{-i}$, $H_t^i$ is a controlled Markov Chain controlled by $U_t^i$ and player $i$ faces a Markov Decision Problem. By Lemma \ref{lem:usiisinfostate}, $\Comp_t^i$ is an information state (as defined in \ref{def:app:infostate}) of this MDP. Therefore, by {the Policy Equivalence Lemma (Lemma \ref{lem:policyeval})} we have
	\begin{equation}\label{prgniqt}
		\Pr^{g^{i}, g^{-i}}(\comp_t^i) = \Pr^{\rho^{i}, g^{-i}}(\comp_t^i).
	\end{equation}
	
	Furthermore, {from the definition of USI we have}
	\begin{align}
		\Pr^{g^{i}, g^{-i}}(h_t^j|\comp_t^i) &= \sum_{\tilde{x}_t, \tilde{h}_t^{-\{i, j\} }}\Phi_t^{i,g^{-i}}(\tilde{x}_t, (h_t^j, h_t^{-\{i, j\}})|\comp_t^i)\\
		&=: F_t^{i,j, g^{-i}}(h_t^j|\comp_t^i).
    \end{align}
    
    {Using Bayes Rule, we then have}
    \begin{align}
		\Pr^{g^{i}, g^{-i}}(\comp_t^i|h_t^j) &= \dfrac{\Pr^{g^{i}, g^{-i}}(h_t^j|\comp_t^i)\Pr^{g^{i}, g^{-i}}(\comp_t^i) }{\sum_{\tilde{\comp}_t^i} \Pr^{g^{i}, g^{-i}}(h_t^j|\tilde{\comp}_t^i)\Pr^{g^{i}, g^{-i}}(\tilde{\comp}_t^i) }\\
		&=\dfrac{F_t^{i, j, g^{-i}}(h_t^j|\comp_t^i)  \Pr^{g^{i}, g^{-i}}(\comp_t^i)}{\sum_{\tilde{\comp}_t^i} F_t^{i, j, g^{-i}}(h_t^j|\tilde{\comp}_t^i)  \Pr^{g^{i}, g^{-i}}(\tilde{\comp}_t^i)}.\label{prgniqtbayes}
	\end{align}
	
	{Note that \eqref{prgniqtbayes} applies for all strategies $g^i$. Replacing $g^i$ with $\rho^i$ we have}
    \begin{align}
		\Pr^{\rho^{i}, g^{-i}}(\comp_t^i|h_t^j) &= \dfrac{F_t^{i, j, g^{-i}}(h_t^j|\comp_t^i)  \Pr^{\rho^{i}, g^{-i}}(\comp_t^i)}{\sum_{\tilde{\comp}_t^i} F_t^{i, j, g^{-i}}(h_t^j|\tilde{\comp}_t^i)  \Pr^{\rho^{i}, g^{-i}}(\tilde{\comp}_t^i)}.\label{prgniqtbayesrho}
	\end{align}
	
    Combining \eqref{prgniqt}, \eqref{prgniqtbayes}, and \eqref{prgniqtbayesrho} we conclude that
	\begin{align}
		\Pr^{g^{i}, g^{-i}}(\comp_t^i|h_t^j) = \Pr^{\rho^{i}, g^{-i}}(\comp_t^i|h_t^j).\label{eq:initdistalsook}
	\end{align}
	
	Using \eqref{florida}, Lemma \ref{lem:htihtjcondindepqti}, and Lemma \ref{lem:usisatisoldconds} we have
	\begin{align}
		&\quad~\Pr^g(\tilde{h}_{t+1}^j|h_t^j, u_t^j) \\
		&= \sum_{\tilde{h}_t^i: \Pr^{g}(\tilde{h}_t^i, h_t^j) > 0}\sum_{\tilde{u}_t^i} \Pr^g(\tilde{h}_{t+1}^j|\tilde{h}_t^i, h_t^j, \tilde{u}_t^i, u_t^j)\Pr^g(\tilde{u}_t^i|\tilde{h}_t^i, h_t^j, u_t^j)\Pr^g(\tilde{h}_t^i|h_t^j, u_t^j)\\
		&= \sum_{\tilde{h}_t^i, \tilde{u}_t^i} \Pi_t^{j, i, g^{-\{i, j\}}} (\tilde{h}_{t+1}^j|\tilde{\comp}_t^i, h_t^j, \tilde{u}_t^i, u_t^j) g_t^i(\tilde{u}_t^i|\tilde{h}_t^i) \Pr^g(\tilde{h}_t^i|h_t^j)\label{eq:whateverlabelthismeans} \\
		&=\sum_{\tilde{h}_t^i, \tilde{u}_t^i} \Pi_t^{j, i, g^{-\{i, j\}}} (\tilde{h}_{t+1}^j|\tilde{\comp}_t^i, h_t^j, \tilde{u}_t^i, u_t^j) g_t^i(\tilde{u}_t^i|\tilde{h}_t^i) \Pr^{g}(\tilde{h}_t^i|\tilde{\comp}_t^i)\Pr^g(\tilde{\comp}_t^i|h_t^j)\label{eq:algreenhowcanyouamendabrokenheart} \\
		&=\sum_{\tilde{\comp}_t^i, \tilde{u}_t^i} \Pi_t^{j, i, g^{-\{i, j\}}} (\tilde{h}_{t+1}^j|\tilde{\comp}_t^i, h_t^j, 	\tilde{u}_t^i, u_t^j) \left(\sum_{\hat{h}_t^i} g_t^i(\tilde{u}_t^i|\hat{h}_t^i) \Pr^{g}(\hat{h}_t^i|\tilde{\comp}_t^i)\right) \Pr^g(\tilde{\comp}_t^i|h_t^j)\\
		&=\sum_{\tilde{\comp}_t^i, \tilde{u}_t^i} \Pi_t^{j, i, g^{-\{i, j\}}} (\tilde{h}_{t+1}^j|\tilde{\comp}_t^i, h_t^j, 	\tilde{u}_t^i, u_t^j) \rho_t^i(\tilde{u}_t^i|\tilde{\comp}_t^i) \Pr^g(\tilde{\comp}_t^i|h_t^j),\label{georgia}
	\end{align}
    {where in \eqref{eq:whateverlabelthismeans} we utilized Lemma \ref{lem:usisatisoldconds} and the function $\Pi_t^{j, i, g^{-\{i, j\}}}$ defined in it. In \eqref{eq:algreenhowcanyouamendabrokenheart} we applied Lemma \ref{lem:htihtjcondindepqti}. In the last equation we used \eqref{florida} and the definition of USI. }
	
	Following a similar argument, we can show that
	\begin{align}
		&\quad~\Pr^{\rho^i, g^{-i}}(\tilde{h}_{t+1}^j|h_t^j, u_t^j) \\
		&=\sum_{\tilde{\comp}_t^i, \tilde{u}_t^i} \Pi_t^{j, i, g^{-\{i, j\}}} (\tilde{h}_{t+1}^j|\tilde{\comp}_t^i, h_t^j, 	\tilde{u}_t^i, u_t^j) \rho_t^i(\tilde{u}_t^i|\tilde{\comp}_t^i) \Pr^{\rho^i, g^{-i}}(\tilde{\comp}_t^i|h_t^j).\label{hawaii}
	\end{align}
	
	{Using \eqref{eq:initdistalsook} and comparing \eqref{georgia} with \eqref{hawaii},} we conclude that
	\begin{align}
		\Pr^g(\tilde{h}_{t+1}^j|h_t^j, u_t^j) = \Pr^{\rho^i, g^{-i}}(\tilde{h}_{t+1}^j|h_t^j, u_t^j),
	\end{align}
    {proving statement (1) of the Lemma.}
	
	Following an analogous argument, we can show that
	\begin{align}
		\E^g[R_{t}^j|h_t^j, u_t^j] &=\sum_{\tilde{\comp}_t^i, \tilde{u}_t^i} r_t^{i, j, g^{-\{i, j\}}} (\tilde{\comp}_t^i, h_t^j, 	\tilde{u}_t^i, u_t^j) \rho_t^i(\tilde{u}_t^i|\tilde{\comp}_t^i) \Pr^g(\tilde{\comp}_t^i|h_t^j)\\
		\E^{\rho^i, g^{-i}}[R_{t}^j|h_t^j, u_t^j] &=\sum_{\tilde{\comp}_t^i, \tilde{u}_t^i} r_t^{i, j, g^{-\{i, j\}}} (\tilde{\comp}_t^i, h_t^j, 	\tilde{u}_t^i, u_t^j) \rho_t^i(\tilde{u}_t^i|\tilde{\comp}_t^i) \Pr^{\rho^i, g^{-i}}(\tilde{\comp}_t^i|h_t^j),
	\end{align}
	where $r_t^{i, j, g^{-\{i, j\}}}$ is defined in Lemma \ref{lem:usisatisoldconds}.
	{We similarly conclude that} 
	\begin{align}
		\E^g[R_{t}^j|h_t^j, u_t^j] &= \E^{\rho^i, g^{-i}}[R_{t}^j|h_t^j, u_t^j],
	\end{align}
	{proving statement (2) of the Lemma.}
\end{proof}

\begin{lemma}\label{lem:selfKfunc}
	{Suppose that $\Comp^i$ is unilaterally sufficient information for player $i$. Let $g^{-i}$ be a fully mixed behavioral strategy profile for players other than $i$.} Define $\Qfunc_\tau^i$ through
	\begin{equation}
		\Qfunc_\tau^i(h_\tau^i, u_\tau^i) = \E^{g^{-i}}[R_\tau^i|h_\tau^i, u_\tau^i] + \underset{\tilde{g}_{\tau+1:T}^i}{\max}~\E^{\tilde{g}_{\tau+1:T}^i, g^{-i}}\left[\sum_{t=\tau+1}^{T} R_t^i\Big|h_\tau^i, u_\tau^i\right].
	\end{equation}
	Then there exist a function $\hat{\Qfunc}_\tau^{i}: \Compset_\tau^i\times \mathcal{U}_\tau^i\mapsto [-T, T]$ such that
	\begin{equation}
		\Qfunc_\tau^i(h_\tau^i, u_\tau^i) = \hat{\Qfunc}_\tau^{i}(\comp_\tau^i, u_\tau^i).
	\end{equation}
\end{lemma}

\begin{proof}
	{By Lemma \ref{lem:usiisinfostate}, $\Comp^i$ is an information state for the payoff of player $i$ under $g^{-i}$. Fixing $g^{-i}$, $H_t^i$ is a controlled Markov Chain controlled by $U_t^i$. Through Definition \ref{def:app:infostate}, $\Comp_t^i$ is an information state of this controlled Markov Chain. The Lemma then follows from a direct application of Lemma \ref{lemma:twoMDPequiv}.} 
\end{proof}

\begin{lemma}\label{lem:usisereplace}
	Suppose that $\Comp^i$ is unilaterally sufficient information for player $i$. Let $g$ be (the strategy part of) a sequential equilibrium. Then there exist a $\Comp^i$-based strategy $\rho^i$ such that $(\rho^i, g^{-i})$ is (the strategy part of) a sequential equilibrium with the same expected payoff profile as $g$. 
\end{lemma}

\begin{proof}[Proof of Lemma \ref{lem:usisereplace}]
	  {Recall that in Theorem \ref{thm:SEdefequiv} we established the equivalence of a variety of definitions of Sequential Equilibrium for strategy profiles.} Let $(g, \Qfunc)$ be a sequential equilibrium under Definition \ref{def:KSE2}. Let $(g^{(n)}, \Qfunc^{(n)})$ be a sequence of strategy and conjecture profiles that satisfies conditions (1)(2')(3) of Definition \ref{def:KSE2}.
	
	Set $\rho^{(n), i}$ through
	\begin{equation}\label{rhotheultimateequivalentstrategy2}
		\rho_t^{(n), i}(u_t^i|\comp_t^i) = \sum_{\tilde{h}_t^i}g_t^{(n), i}(u_t^i|\tilde{h}_t^i)F_t^{i, g^{(n), i}}(\tilde{h}_t^i|\comp_t^i),
	\end{equation}
	where $F_t^{i, g^{(n), i}}$ is defined in Definition \ref{def:usi}.
	By replacing the sequence with one of its sub-sequences, without loss of generality, assume that $\rho^{(n), i} \rightarrow \rho^i$ for some $\rho^i$.
	
	  {For the ease of notation,} denote $\bar{g}^{(n)} = (\rho^{(n), i}, g^{(n), -i})$ and $\bar{g} = (\rho^i, g^{-i})$. We have $\bar{g}^{(n)}\rightarrow \bar{g}$. 
	  {In the rest of the proof, we will} show that $(\bar{g}, \Qfunc)$ is a sequential equilibrium. 
 
	We only need to show that $\bar{g}$ is sequentially rational to $\Qfunc$ and $(\bar{g}^{(n)}, \Qfunc^{(n)})$ satisfies conditions (2') of Definition \ref{def:KSE2},
	as conditions (1)(3) of Definition \ref{def:KSE2} are {true by construction}. Since $\bar{g}^{-i} = g^{-i}$, we automatically have $\bar{g}^{j}$ to be sequentially rational given $\Qfunc^{j}$ for all $j\in\mathcal{I}\backslash\{i\}$, and $\Qfunc^{(n), i}$ to be consistent with $\bar{g}^{(n), -i}$ for each $n$. It suffices to establish 
	\begin{enumerate}[(i)]
		\item $\rho^i$ is sequentially rational with respect to $\Qfunc^{i}$; and 
		\item $\Qfunc^{(n), j}$ is consistent with $\bar{g}^{(n), -j}$ for each $j\in\mathcal{I}\backslash\{i\}$.
	\end{enumerate}
	
	  To establish (i), we will use the Lemma \ref{lem:selfKfunc} to show that $\Qfunc_t^i(h_t^i, u_t^i)$ is a function of $(\comp_t^i, u_t^i)$, and hence one can use an $\comp_t^i$ based strategy to optimize $\Qfunc_t^i$.
	
	\textbf{Proof of (i):} By construction,
	\begin{equation}
		\rho_t^{(n), i}(\comp_t^i) = \sum_{\tilde{h}_t^i: \tilde{\comp}_t^i = \comp_t^i} g_t^{(n), i}(\tilde{h}_t^i) \cdot\eta_t^{(n)} (\tilde{h}_t^i|\comp_t^i),
	\end{equation}
	for some distribution $\eta_t^{(n)}(\comp_t^i) \in \Delta(\mathcal{H}_t^i)$. 
    Let $\eta_t(\comp_t^i)$ be an accumulation point of the sequence $[\eta_t^{(n)}(\comp_t^i)]_{n=1}^\infty$. We have
	\begin{equation}
		\rho_t^i(\comp_t^i) = \sum_{\tilde{h}_t^i: \tilde{\comp}_t^i = \comp_t^i} g_t^{i}(\tilde{h}_t^i)\cdot \eta_t (\tilde{h}_t^i|\comp_t^i).
	\end{equation}

	As a result, we have
	\begin{equation}\label{colorado}
		\mathrm{supp}(\rho_t^{i}(\comp_t^i)) \subseteq \bigcup_{\tilde{h}_t^i: \tilde{\comp}_t^i = \comp_t^i} \mathrm{supp}(g_t^{i}(\tilde{h}_t^i)).
	\end{equation}

    By Lemma \ref{lem:selfKfunc} we have $\Qfunc_t^{(n), i}(h_t^i, u_t^i) = \hat{\Qfunc}_t^{(n), i}(\comp_t^i, u_t^i)$ for some function $\hat{\Qfunc}_t^{(n), i}$. Since $\Qfunc^{(n), i} \rightarrow \Qfunc^i$, we have $\Qfunc_t^i(h_t^i, u_t^i) = \hat{\Qfunc}_t^{i}(\comp_t^i, u_t^i)$ for some function $\hat{\Qfunc}^{i}$. By sequential rationality we have
	\begin{equation}\label{arizona}
		\mathrm{supp}(g_t^{i}(\tilde{h}_t^i)) \subseteq \underset{u_t^i}{\arg\max}~ \hat{\Qfunc}_t^{i}(\comp_t^i, u_t^i),
	\end{equation}
	for all $\tilde{h}_t^i$ whose corresponding compression $\tilde{\comp}_t^i$ satisfies $\tilde{\comp}_t^i = \comp_t^i$. Therefore, by \eqref{colorado} and \eqref{arizona} we conclude that
	\begin{equation}
		\mathrm{supp}(\rho_t^{i}(\comp_t^i)) \subseteq \underset{u_t^i}{\arg\max}~ \hat{\Qfunc}_t^{i}(\comp_t^i, u_t^i),
	\end{equation}
	establishing sequential rationality of $\rho^i$ with respect to $\Qfunc^{i}$.\\~
	
	To establish (ii), we will use the Lemmas \ref{lem:usiclaim} and \ref{lem:prequiv} to show that when player $i$ switches their strategy from $g^{(n), i}$ to $\rho^{(n), i}$, other players face the same control problem at every information set. As a result, their $\Qfunc^{(n), j}$ functions stays the same.
	
	\textbf{Proof of (ii):} Consider player $j\neq i$. Through standard control theory, we know that a collection of functions $\tilde{\Qfunc}^{j}$ is consistent (in the sense of condition (2') of Definition \ref{def:KSE2}) with a fully mixed strategy profile $\tilde{g}^{-j}$ if and only if it satisfies the following equations:
    \begin{subequations}
	\begin{align}
		\tilde{\Qfunc}_T^{j}(h_T^j, u_T^j) &= \E^{\tilde{g}^{-j}}[R_T^j|h_T^j, u_T^j],\\
		\tilde{V}_t^{j}(h_t^j) &= \max_{\tilde{u}_t^j} \tilde{\Qfunc}_t^{j}(h_t^j, \tilde{u}_t^j), \qquad\forall t\in\mathcal{T},\\
		\tilde{\Qfunc}_{t}^{j}(h_t^j, u_t^j) &= \E^{\tilde{g}^{-j}}[R_t^j|h_t^j, u_t^j] + \sum_{\tilde{h}_{t+1}^j} \tilde{V}_{t+1}^{ j}(\tilde{h}_{t+1}^j)\Pr^{\tilde{g}^{-j}}(\tilde{h}_{t+1}^j|h_t^j, u_t^j),\qquad\forall t\in\mathcal{T}\backslash\{T\}.        
	\end{align}
    \end{subequations}
	
	By Lemma \ref{lem:prequiv}, we have
	\begin{align}
		\Pr^{g^{(n), -j}}(\tilde{h}_{t+1}^j|h_t^j, u_t^j) &= \Pr^{\rho^{(n), i}, g^{(n), -\{i, j\}}}(\tilde{h}_{t+1}^j|h_t^j, u_t^j),\\
		\E^{g^{(n), -j}}[R_t^j|h_t^j, u_t^j] &= \E^{\rho^{(n), i}, g^{(n), -\{i, j\}}}[R_t^j|h_t^j, u_t^j],
	\end{align}
	and hence we conclude that $\Qfunc^{(n), j}$ is also consistent with $\bar{g}^{(n), -j} = (\rho^{(n), i}, g^{(n), -\{i, j\}})$. \\~
	
	{Now we have shown that $(\bar{g}, \Qfunc)$ forms a sequential equilibrium. The second half of the Lemma, which states that $\bar{g}$ yields the same expected payoff as $g$, can be shown with the following argument:}
	By Lemma \ref{lem:usiclaim}, $\bar{g}^{(n)}$ yields the same expected payoff profile as $g^{(n)}$. Since the expected payoff of each player is a continuous function of the behavioral strategy profile, we conclude that $\bar{g}$ yields the same expected payoff as $g$.
\end{proof}

{Finally, we conclude Theorem \ref{thm:usiseequiv} from Lemma \ref{lem:usisereplace}.}

\begin{proof}[Proof of Theorem \ref{thm:usiseequiv}]
	Given any SE strategy profile $g$, applying Lemma \ref{lem:usisereplace} iteratively for each $i\in\mathcal{I}$, we obtain a $\Comp$-based SE strategy profile $\rho$ with the same expected payoff profile as $g$. Therefore the set of $\Comp$-based SE payoffs is the same as that of all SE.
\end{proof}

\section{Proofs for Section \ref{sec:siapplications} and Section \ref{sec:openproblems}}\label{app:siapplications}
\subsection{Proof of Proposition \ref{prop:exwpbe}}\label{app:prop:exwpbe}
\begin{prop}[Proposition \ref{prop:exwpbe}, restated]
		In the game defined in Example \ref{ex:wpbeproblem}, the set of $\Comp$-based wPBE payoffs is a proper subset of that of all wPBE payoffs.
	\end{prop}

\begin{proof}
    Set $g_1^B$ to be the strategy of Bob where he always chooses $U_1^B=+1$, and $g_2^A: \mathcal{X}_1^A \times \mathcal{U}_1^B\mapsto \Delta(\mathcal{U}_2^A)$ is given by
	\begin{align*}
		g_2^A(x_1^A, u_1^B) = \begin{cases}
			0\text{ w.p. }1,&\text{if }u_1^B=+1;\\
			x_1^A\text{ w.p. }\frac{2}{3},~ 0 \text{ w.p. }\frac{1}{3},&\text{otherwise},
		\end{cases}
	\end{align*}
	and $g_2^B: \mathcal{X}_{1}^B\times \mathcal{U}_1^B \mapsto \Delta(\mathcal{U}_2^B)$ is the strategy of Bob where he always chooses $U_2^B=-1$ irrespective of $U_1^B$.
	
	The beliefs $\mu_1^B: \mathcal{X}_{1}^{B} \mapsto \Delta(\mathcal{X}_{1}^{A})$, $\mu_2^A: \mathcal{X}_1^A \times \mathcal{U}_1^B \mapsto \Delta(\mathcal{X}_{1}^B)$, and $\mu_2^B:\mathcal{X}_{1}^B \times \mathcal{U}_1^B \mapsto \Delta(\mathcal{X}_1^A)$ are given by
	\begin{align*}
		\mu_1^B(x_1^B) &= \text{the prior of }X_1^{A},\\
		\mu_2^A(x_1^A, u_1^B) &= \begin{cases}
			-1 \text{ w.p. }\frac{1}{2},\;+1 \text{ w.p. }\frac{1}{2},&\text{if }u_1^B=+1;\\
			x_1^A \text{ w.p. }1,&\text{otherwise},
		\end{cases}\\
		\mu_2^B(x_{1}^B, u_1^B) &= \text{the prior of }X_1^{A}.
	\end{align*}
	
	One can verify that $g$ is sequentially rational given $\mu$, and $\mu$ is ``preconsistent'' \citep{hendon1996one} with $g$, i.e. the beliefs can be updated with Bayes rule for consecutive information sets on and off-equilibrium paths. In particular, $(g, \mu)$ is a wPBE. (It can also be shown that $(g, \mu)$ satisfies Watson's {PBE definition \citep{watson2017general}}. However, $(g, \mu)$ is not a PBE in the sense of Fudenberg and Tirole \citep{fudenberg1991perfect}, since $\mu$ violates their ``no-signaling-what-you-don't-know'' condition.)
 
    We proceed to show that no $\Comp$-based wPBE can attain the payoff profile of $g$.
	
	Suppose that $\rho = (\rho^A, \rho^B)$ is a $\Comp$-based weak PBE strategy profile. First, observe that at $t=2$, Alice can only choose her actions based on $U_1^B$ according to the definition of $\Comp^A$-based strategies. Let $\alpha, \beta\in \Delta(\{-1, 0, 1\})$ be Alice's mixed action at time $t=2$ under $U_2^A=-1$ and $U_2^A=+1$ respectively under strategy $\rho^A$. With some abuse of notation, denote $\rho^A = (\alpha, \beta)$.
	There exists no belief system under which Alice is indifferent between all of her three actions at time $t=2$. Therefore, no strictly mixed action at $t=2$ would be sequentially rational. Therefore, sequential rationally of $\rho^A$ (with respect to some belief) implies that $\min\{\alpha(-1), \alpha(0), \alpha(+1) \} = \min\{\beta(-1), \beta(0), \beta(+1) \} = 0$.
	
	To respond to $\rho^A = (\alpha, \beta)$, Bob can always maximizes his stage 2 instantaneous reward to 0 by using a suitable response strategy. If Bob plays $-1$ at $t=1$, his best total payoff is given by $0.2$; if Bob plays $+1$ at $t=1$, his best total payoff is given by $0$. Hence Bob strictly prefers $-1$ to $+1$. Therefore, in any best response (in terms of total expected payoff) to Alice's strategy $\rho^A$, Bob plays $U_1^B = -1$ irrespective of his private type. Therefore, Alice has an instantaneous payoff of $-1$ at $t=1$ and a total payoff $\leq 0$ under $\rho$, proving that the payoff profile of $\rho$ is different from that of $g$.
\end{proof}

\subsection{Proof of Proposition \ref{thm:ouyangusi}}\label{app:thm:ouyangusi}
\begin{prop}[Proposition \ref{thm:ouyangusi}, restated]
	In the model of Example \ref{ex:ouyang}, $\Comp_t^i=(Y_{1:t-1}, U_{1:t-1}, X_t^i)$ is unilaterally sufficient information.
\end{prop}

{We first prove Lemma \ref{lem:condindepouyang}, which establish the conditional independence of the state processes given the common information.}

\begin{lemma}\label{lem:condindepouyang}
	In the model of Example \ref{ex:ouyang}, there exists functions $(\xi_t^{g^i})_{g^i\in\mathcal{G}^i, i\in\mathcal
		I}, \xi_t^{g^i}: \mathcal{Y}_{1:t-1}\times \mathcal{U}_{1:t-1}\mapsto \Delta(\mathcal{X}_{1:t}^i)$ such that
	\begin{equation}
		\Pr^g(x_{1:t}|y_{1:t-1}, u_{1:t-1}) = \prod_{i\in\mathcal{I}} \xi_t^{g^i}(x_{1:t}^i|y_{1:t-1}, u_{1:t-1}),
	\end{equation}
	for all strategy profiles $g$ and all $(y_{1:t-1}, u_{1:t-1})$ admissible under $g$.
\end{lemma}

\begin{proof}[Proof of Lemma \ref{lem:condindepouyang}]
	Denote $H_t^0=(\bY_{1:t-1}, \bU_{1:t-1})$. We prove the result by induction on time $t$.
	
	\textbf{Induction Base:} The result is true for $t=1$ since $H_1^0=\varnothing$ and the random variables $(X_1^i)_{i\in\mathcal{I}}$ are assumed to be mutually independent. 
	
	\textbf{Induction Step:} Suppose that we have proved Lemma \ref{lem:condindepouyang} for time $t$. We then prove the result for time $t+1$.
	
	We have
	\begin{align}
		\Pr^g(x_{1:t+1}, y_t, u_t|h_t^0) &= \Pr^g(x_{t+1}, y_t|x_{1:t}, u_t, h_t^0)\Pr^g(u_t|x_{1:t}, h_t^0)\Pr^g(x_{1:t}| h_t^0)\\
		&=\prod_{i\in\mathcal{I}} \left(\Pr(x_{t+1}^i, y_t^i|x_t^i, u_t) g_t^i(u_t^i|x_{1:t}^i, h_t^0) \xi_t^{g^i}(x_{1:t}^i| h_t^0) \right)\label{eq:147}\\
		&=:\prod_{i\in\mathcal{I}} \nu_t^{g^i}(x_{1:t+1}^i, y_t, u_t, h_t^0) = \prod_{i\in\mathcal{I}} \nu_t^{g^i}(x_{1:t+1}^i, h_{t+1}^0),
	\end{align}
    {where the induction hypothesis is utilized in \eqref{eq:147}.}
	
	Therefore, {using Bayes rule,}
	\begin{align}
		\Pr^g(x_{1:t+1}| h_{t+1}^0) &= \dfrac{\Pr^g(x_{1:t+1}, y_t, u_t| h_{t}^0)}{\sum_{\tilde{y}_t, \tilde{u}_t} \Pr^g(\tilde{x}_{1:t+1}, y_t, u_t| h_{t+1}^0)}\\
		&=\dfrac{\prod_{i\in\mathcal{I}} \nu_t^{g^i}(x_{1:t+1}^i, h_{t+1}^0)}{\sum_{\tilde{x}_{1:t+1}}\prod_{i\in\mathcal{I}} \nu_t^{g^i}(\tilde{x}_{1:t+1}^i, h_{t+1}^0)}\\
		&=\dfrac{\prod_{i\in\mathcal{I}} \nu_t^{g^i}(x_{1:t+1}^i, h_{t+1}^0)}{\prod_{i\in\mathcal{I}} \sum_{\tilde{x}_{1:t+1}^i} \nu_t^{g^i}(\tilde{x}_{1:t+1}^i, h_{t+1}^0)}\\
		&=:\prod_{i\in\mathcal{I}} \xi_{t+1}^{g^i} (x_{1:t+1}^i| h_{t+1}^0),
	\end{align}
	where
	\begin{equation}
		\xi_{t+1}^{g^i} (x_{1:t+1}^i| h_{t+1}^0) := \dfrac{ \nu_t^{g^i}(x_{1:t+1}^i, h_{t+1}^0)}{ \sum_{\tilde{x}_{1:t+1}^i} \nu_t^{g^i}(\tilde{x}_{1:t+1}^i, h_{t+1}^0)},
	\end{equation}
	establishing the induction step.
\end{proof}

\begin{proof}[Proof of Proposition \ref{thm:ouyangusi}]
	Denote $H_t^0=(\bY_{1:t-1}, \bU_{1:t-1})$. Then $\Comp_t^i=(H_t^0, X_t^i)$. Given Lemma \ref{lem:condindepouyang}, we have
	\begin{align}
		\Pr^{g}(x_{1:t-1}^i|\comp_t^i) &= \dfrac{\Pr^g(x_{1:t}^i|h_t^0)}{\Pr^g(x_{t}^i|h_t^0)}=\dfrac{\xi_t^{g^i}(x_{1:t}^i|h_t^0) }{\sum_{\tilde{x}_{1:t-1}^i} \xi_t^{g^i}((\tilde{x}_{1:t-1}^i, x_{t}^i)|h_t^0)}\\
		&=:\tilde{F}_t^{i, g^i}(x_{1:t-1}^i|\comp_t^i).
	\end{align}
	
	Since $H_t^i=(\Comp_t^i, X_{1:t-1}^i)$, we conclude that
	\begin{equation}\label{eq:selfes}
		\Pr^{g}(\tilde{h}_t^i|\comp_t^i) = F_t^{i, g^i}(\tilde{h}_t^i|\comp_t^i),
	\end{equation}
	for some function $F_t^{i, g^i}$.
	
	Given Lemma \ref{lem:condindepouyang}, we have
	\begin{align}
		\Pr^{g}(\tilde{x}_{1:t}^{-i}|h_t^i) &= \dfrac{\Pr^g(\tilde{x}_{1:t}^{-i}, x_{1:t}^i|h_t^0)}{\Pr^g(x_{1:t}^i|h_t^0)}=\prod_{j\neq i} \xi_t^{g^j}(\tilde{x}_{1:t}^j|h_t^0).
	\end{align}
	
	As a result, we have
	\begin{align}
		\Pr^{g}(\tilde{x}_{1:t}^{-i}, \tilde{\comp}_t^i|h_t^i) &= \bm{1}_{\{\tilde{\comp}_t^i = \comp_t^i\} } \prod_{j\neq i} \xi_t^{g^j}(x_{1:t}^j|h_t^0)\\
		&=:\tilde{\Phi}_t^{i,g^{-i}}(\tilde{x}_{1:t}^{-i}|\comp_t^i).
	\end{align}
	
	Since $(\bX_t, H_t^{-i})$ is {a fixed function of} $(\bX_{1:t}^{-i}, \Comp_t^i)$, we conclude that
	\begin{equation}\label{eq:otheres}
		\Pr^{g}(\tilde{x}_t, \tilde{h}_{t}^{-i}|h_t^i) = \Phi_t^{i, g^{-i}}(\tilde{x}_t, \tilde{h}_{t}^{-i}|\comp_t^i),
	\end{equation}
	for some function $\Phi_t^{i, g^{-i}}$. 
	
	Combining \eqref{eq:selfes} and \eqref{eq:otheres} while using the fact that $\Comp_t^i$ is a function of $H_t^i$, we obtain
	\begin{align}
		\Pr^{g}(\tilde{x}_t, \tilde{h}_{t}|\comp_t^i) &= F_t^{i, g^i}(\tilde{h}_t^i|\comp_t^i)\Phi_t^{i, g^{-i}}(\tilde{x}_t, \tilde{h}_{t}^{-i}|\comp_t^i).
	\end{align}
	
	We conclude that $\Comp^i$ is unilaterally sufficient information.
\end{proof}

\subsection{Proof of Proposition \ref{prop:zerosumobsac}}\label{app:prop:zerosumobsac}
\begin{prop}[Proposition \ref{prop:zerosumobsac}, restated]
	In the game of Example \ref{ex:zerosumobsac} belief-based equilibria do not exist.
\end{prop}

\begin{proof}
	%Similar to the proof of Proposition \ref{prop:exnonobsac}, 
    We first characterize all the Bayes-Nash equilibria of Example \ref{ex:zerosumobsac} in behavioral strategy profiles. Then we will show that none of the BNE corresponds to a belief-based equilibrium. 
	
	Let $\alpha=(\alpha_1, \alpha_2)\in [0, 1]^2$ describe Alice's behavioral strategy: $\alpha_1$ is the probability that Alice plays $U_1^A=-1$ given $X_1^A=-1$; $\alpha_2$ is the probability that Alice plays $U_1^A=+1$ given $X_1^A=+1$. Let $\beta=(\beta_1, \beta_2)\in [0, 1]^2$ denote Bob's behavioral strategy: $\beta_1$ is the probability that Bob plays $U_2^B=\mathrm{U}$ when observing $U_1^A=-1$, $\beta_2$ is the probability that Bob plays $U_2^B=\mathrm{U}$ when observing $U_1^A=+1$.
	
	\textbf{Claim:}
	\begin{equation}
		\alpha^*=\left(\frac{1}{3}, \frac{1}{3}\right),\quad \beta^*=\left(\frac{1}{3}+c, \frac{1}{3}-c\right),
	\end{equation}
	is the unique BNE of Example \ref{ex:zerosumobsac}.
	
	Given the claim, one can conclude that a belief based equilibrium does not exist in this game: {Bob's true belief $b_2$ on $X_2$ at the beginning of stage 2, given his information $H_2^B = U_1^A$, would satisfy}
    \begin{align}
		b_2^-(+1) &= \dfrac{\alpha_1}{\alpha_1+1-\alpha_2},\quad\text{ if }\alpha\neq (0, 1);\\
		b_2^+(+1) &= \dfrac{\alpha_2}{\alpha_2+1-\alpha_1},\quad\text{ if }\alpha\neq (1, 0),
	\end{align}
    {where $b_2^{-}$ represents the belief under $U_1^A = -1$ and $b_2^{+}$ represents the belief under $U_1^A = +1$. If Alice plays $\alpha^*=\left(\frac{1}{3}, \frac{1}{3}\right)$, then $b_2^- = b_2^+$. Under a belief-based equilibrium concept (e.g. \cite{ouyang2016dynamic,vasal2019spbe}), Bob's stage behavioral strategy $\beta$ should yield the same action distribution under the same belief, which means that $\beta_1=\beta_2$. However we have $\beta^*=\left(\frac{1}{3}+c, \frac{1}{3}-c\right)$. Therefore, $(\alpha^*, \beta^*)$, the unique BNE of the game, is not a belief-based equilibrium. We conclude that a belief-based equilibrium does not exist in Example \ref{ex:zerosumobsac}.}\\  
	
	\textbf{Proof of Claim:}
	Denote Alice's total expected payoff to be $J(\alpha, \beta)$. Then 
	\begin{align*}
		&\es J(\alpha, \beta) \\
		&= \frac{1}{2}c(1-\alpha_1+\alpha_2) + \frac{1}{2} \alpha_1 \cdot 2\beta_1 + \frac{1}{2} (1-\alpha_1)(1-\beta_2) +\frac{1}{2} (1-\alpha_2)(1-\beta_1) + \frac{1}{2} \alpha_2\cdot 2\beta_2 \\
		&=\frac{1}{2}c(1-\alpha_1+\alpha_2) + \frac{1}{2} (2 - \alpha_1 - \alpha_2) + \frac{1}{2}(2\alpha_1 + \alpha_2 - 1)\beta_1 + \frac{1}{2}(2\alpha_2 + \alpha_1 - 1)\beta_2.
	\end{align*}
	
	Define $J^*(\alpha) = \min_\beta J(\alpha, \beta)$. Since the game is zero-sum, Alice plays $\alpha$ at some equilibrium if and only if $\alpha$ maximizes $J^*(\alpha)$. We compute
	\begin{align*}
		J^*(\alpha) &= \frac{1}{2}c(1-\alpha_1+\alpha_2) + \frac{1}{2} (2 - \alpha_1 - \alpha_2) + \\
		&+\frac{1}{2}\min\{2\alpha_1 + \alpha_2 - 1, 0 \} + \frac{1}{2}\min\{\alpha_1 + 2\alpha_2 - 1, 0 \}.
	\end{align*}
	
	Since $J^*(\alpha)$ is a continuous piecewise linear function, the set of maximizers can be found by comparing the values at the extreme points of the pieces.	
	We have
	\begin{align*}
		J^*(0, 0) &= \frac{1}{2}c + 1 - \frac{1}{2} - \frac{1}{2} = \frac{1}{2}c;\\
		J^*\left(\frac{1}{2}, 0\right) &= \frac{1}{2}c \cdot \frac{1}{2} + \frac{1}{2} \cdot \frac{3}{2} + \frac{1}{2} \cdot 0 - \frac{1}{2}\cdot \frac{1}{2}= \frac{1}{4}c + \frac{1}{2};\\
		J^*\left(0, \frac{1}{2}\right) &= \frac{1}{2}c\cdot \frac{3}{2} + \frac{1}{2} \cdot \frac{3}{2} - \frac{1}{2} \cdot \frac{1}{2} - \frac{1}{2}\cdot 0= \frac{3}{4}c + \frac{1}{2};\\
		J^*(1, 0) &= \frac{1}{2}c \cdot 0+ \frac{1}{2}\cdot 1 + \frac{1}{2}\cdot 0 + \frac{1}{2}\cdot 0 = \frac{1}{2};\\
		J^*(0, 1) &= \frac{1}{2}c \cdot 2 + \frac{1}{2}\cdot 1 + \frac{1}{2} \cdot 0 + \frac{1}{2}\cdot 0 = c + \frac{1}{2};\\
		J^*\left(\frac{1}{3}, \frac{1}{3}\right) &= \frac{1}{2}c + \frac{1}{2}\cdot \frac{4}{3} + \frac{1}{2} \cdot 0  + \frac{1}{2}\cdot 0 = \frac{1}{2}c + \frac{2}{3};\\
		J^*(1, 1) &= \frac{1}{2}c + \frac{1}{2}\cdot 0 + \frac{1}{2}\cdot 0 + \frac{1}{2}\cdot 0 = \frac{1}{2}c.
	\end{align*}
	
	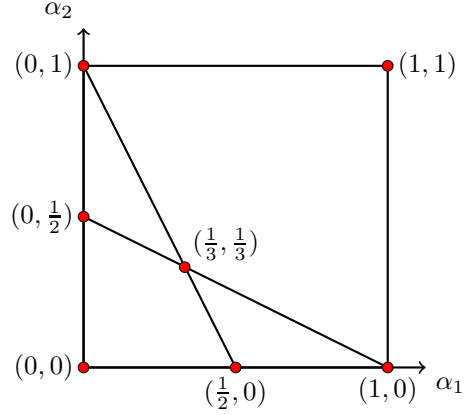
\begin{figure}[!ht]
		\centering
		\begin{tikzpicture}[scale=1.0]
			\draw[thick, ->] (0, 0) -- (4.5, 0) node[anchor=north west] {$\alpha_1$};
			\draw[thick, ->] (0, 0) -- (0, 4.5) node[anchor=south east] {$\alpha_2$};
			
			\draw[thick] (4, 0) -- (4, 4) node[anchor=west] {$(1, 1)$};
			\draw[thick] (4, 4) -- (0, 4) node[anchor=east] {$(0, 1)$};
			\draw[thick] (0, 4) -- (0, 0) node[anchor=east] {$(0, 0)$};
			\draw[thick] (0, 0) -- (4, 0) node[anchor=north] {$(1, 0)$};
			
			\draw[thick] (0, 4) -- (2, 0) node[anchor=north] {$(\frac{1}{2}, 0)$};
			\draw[thick] (4, 0) -- (0, 2) node[anchor=east] {$(0, \frac{1}{2})$};
			\draw[thick] (1.3, 1.3) node[anchor=south west] {$(\frac{1}{3}, \frac{1}{3})$};
			
			\foreach \point in {(0, 0), (4, 0), (0, 4), (2, 0), (0, 2), (1.33, 1.33), (4, 4)}
			\draw [fill=red] \point circle [radius=0.07];
		\end{tikzpicture}
		\caption{The pieces (polygons) for which $J^*(\alpha)$ is linear on. The extreme points of the pieces are labeled.} \label{fig: extremepoints}
	\end{figure}
	
	Since $c < \frac{1}{3}$, we have $(\frac{1}{3}, \frac{1}{3})$ to be the unique maximum among the extreme points. Hence we have $\arg\max_{\alpha} J^*(\alpha) = \{(\frac{1}{3}, \frac{1}{3}) \}$, i.e. Alice always plays $\alpha^*=(\frac{1}{3}, \frac{1}{3})$ in any BNE of the game.
	
	Now, consider Bob's equilibrium strategy. $\beta^*$ is an equilibrium strategy of Bob only if $\alpha^* \in \arg\max_{\alpha} J(\alpha, \beta^*)$.
	
	For each $\beta$, $J(\alpha, \beta)$ is a linear function of $\alpha$ and
	\begin{align*}
		\nabla_{\alpha} J(\alpha, \beta) = \left(-\frac{1}{2}c - \frac{1}{2} + \beta_1 + \frac{1}{2}\beta_2, \frac{1}{2}c - \frac{1}{2} + \frac{1}{2}\beta_1 + \beta_2 \right),\quad \forall \alpha\in (0, 1)^2.&
	\end{align*}
	
	We need $\nabla_{\alpha} J(\alpha, \beta^*)\Big|_{\alpha=\alpha^*} = (0, 0)$. Hence
	\begin{align*}
		-\frac{1}{2}c - \frac{1}{2} + \beta_1^* + \frac{1}{2}\beta_2^* &= 0;\\
		\frac{1}{2}c - \frac{1}{2} + \frac{1}{2}\beta_1^* + \beta_2^* &= 0,
	\end{align*}
	which implies that $\beta^*=(\frac{1}{3}+c, \frac{1}{3}-c)$, proving the claim.
\end{proof}

\end{appendices}

\bibliography{mybib}

\end{document}